\documentclass[11pt]{article}

\usepackage[hyphens]{url}
\usepackage{fullpage}
\usepackage{color}
\usepackage{amsmath,amsfonts,amsthm}
\usepackage{hyperref}
\usepackage{mathtools}
\usepackage{lmodern}
\usepackage{multirow}
\usepackage{array}
\usepackage{microtype}

\usepackage{fullpage}
\usepackage{amsmath,amsthm,amsfonts,dsfont}
\usepackage{amssymb,latexsym,graphicx}
\usepackage{mathtools}
\usepackage{hyperref}
\usepackage{xcolor}
\hypersetup{
	colorlinks,
	linkcolor={red!75!black},
	citecolor={blue!75!black},
	urlcolor={blue!75!black}
}

\newtheorem{theorem}{Theorem}[section]

\newtheorem{proposition}[theorem]{Proposition}
\newtheorem{lemma}[theorem]{Lemma}

\newtheorem{fact}[theorem]{Fact}

\newtheorem{corollary}[theorem]{Corollary}
\newtheorem{definition}[theorem]{Definition}

\newcommand{\C}{\ensuremath{\mathbb{C}}}

\newcommand{\F}{\ensuremath{\mathbb{F}}}
\newcommand{\R}{\ensuremath{\mathbb{R}}}
\newcommand{\Z}{\ensuremath{\mathbb{Z}}}

\newcommand{\lat}{\mathcal{L}}

\newcommand{\eps}{\varepsilon} 
\renewcommand{\epsilon}{\varepsilon}
\newcommand{\poly}{\mathrm{poly}}

\DeclareMathOperator{\dist}{dist}

\renewcommand{\vec}[1]{\ensuremath{\boldsymbol{#1}}}
\newcommand{\basis}{\ensuremath{\mathbf{B}}}

\newcommand{\problem}[1]{\ensuremath{\mathrm{#1}} }

\DeclarePairedDelimiter\inner{\langle}{\rangle}
\DeclarePairedDelimiter\floor{\lfloor}{\rfloor}
\DeclarePairedDelimiter\ceil{\lceil}{\rceil}

\renewcommand{\Re}{\mathfrak{Re}}
\renewcommand{\Im}{\mathfrak{Im}}

\newcommand{\blue}[1]{{\color{blue} #1}}
\newcommand{\CVP}{\problem{CVP}}
\newcommand{\SVP}{\problem{SVP}}
\newcommand{\CVPP}{\problem{CVPP}}

\DeclareMathOperator{\var}

\newcommand{\norm}[1]{\|#1\|}
\newcommand{\set}[1]{\{#1\}}

\newcommand{\ind}{\textrm{ind}}
\newcommand{\val}{\textrm{val}}

\DeclareMathOperator*{\E}{\mathbb{E}}

\newcommand{\ton}{\vec{t}_{\textrm{on}}}
\newcommand{\toff}{\vec{t}_{\textrm{off}}}
\newcommand{\rgood}{r_{\textrm{good}}}
\newcommand{\rbad}{r_{\textrm{bad}}}

\newcommand{\roff}{r_{\textrm{off}}}

\newcommand{\bit}{\ensuremath{\{0, 1\}}}
\newcommand{\pmo}{\ensuremath{\{-1, 1\}}}

\newcommand{\cc}[1]{\ensuremath{\mathsf{#1}}}

\newcommand{\NP}{\cc{\NP}}

\newcommand{\lampar}{\ensuremath{\lambda_{\mathrm{par}}}}
\newcommand{\vpar}{\ensuremath{\vec{v}_{\mathrm{par}}}}

\newcommand{\fh}{\hat{f}}

\DeclarePairedDelimiter\iprod{\langle}{\rangle}
\DeclarePairedDelimiter\abs{|}{|}
\DeclarePairedDelimiter\card{|}{|}
\DeclareMathOperator{\sign}{sign}

\title{Fine-grained hardness of CVP(P)---\\Everything that we can prove (and nothing else)}
\author{Divesh Aggarwal\thanks{National University of Singapore. Part of this work was funded by the Singapore Ministry of
Education under grant MOE2019-T2-1-145 and the National Research Foundation under grant R-710-000-012-135. }\\ \texttt{dcsdiva@nus.edu.sg} 
\and Huck Bennett\thanks{University of Michigan. Supported by the National Science Foundation under award CCF-2006857. The views expressed in this work are those
    of the authors and do not necessarily reflect the official policy
    or position of the National Science Foundation. Part of this work was performed while the author was at Northwestern University and supported by a Warren Postdoctoral Fellowship.}\\ \texttt{huckbennett@gmail.com}  
\and Alexander Golovnev\thanks{Georgetown University. Partially supported by a Rabin Postdoctoral Fellowship.}\\ \texttt{alexgolovnev@gmail.com}  
\and Noah Stephens-Davidowitz\thanks{Cornell University. Partially supported by an NSF-BSF grant number 1718161 and
NSF CAREER Award number 1350619.}\\ \texttt{noahsd@gmail.com} 
}
\date{}

\begin{document}

\maketitle
\begin{abstract}
We show a number of fine-grained hardness results for the Closest Vector Problem in the $\ell_p$ norm ($\CVP_p$), and its approximate and non-uniform variants. First, we show that $\CVP_p$ cannot be solved in $2^{(1-\eps)n}$ time for all $p \notin 2\Z$ and $\eps > 0$, assuming the Strong Exponential Time Hypothesis (SETH). Second, we extend this by showing that there is no $2^{(1-\eps)n}$-time algorithm for \emph{approximating} $\CVP_p$ to within a constant factor $\gamma$ for such $p$ assuming a ``gap'' version of SETH, with an explicit relationship between $\gamma$, $p$, and the arity $k = k(\eps)$ of the underlying hard CSP. Third, we show the same hardness result for (exact) $\CVP_p$ with preprocessing (assuming non-uniform SETH).

For exact ``plain'' $\CVP_p$, the same hardness result was shown in [Bennett, Golovnev, and Stephens-Davidowitz FOCS 2017] for all but finitely many $p \notin 2\Z$, where the set of exceptions depended on $\eps$ and was not explicit. For the approximate and preprocessing problems, only very weak bounds were known prior to this work.

 We also show that the restriction to $p \notin 2\Z$ is in some sense inherent. In particular, we show that no ``natural'' reduction can rule out even a $2^{3n/4}$-time algorithm for $\CVP_2$ under SETH. For this, we prove that the possible sets of closest lattice vectors to a target in the $\ell_2$ norm have quite rigid structure, which essentially prevents them from being as expressive as $3$-CNFs. 
 
We prove these results using techniques from many different fields, including complex analysis, functional analysis, additive combinatorics, and discrete Fourier analysis. E.g., along the way, we give a new (and tighter) proof of Szemer\'{e}di's cube lemma for the boolean cube.
\end{abstract}

\vfill
\thispagestyle{empty}
\newpage
\tableofcontents
\thispagestyle{empty}
\newpage
\pagenumbering{arabic}

\section{Introduction}
\label{sec:introduction}

A lattice $\lat$ is the set of all integer linear combinations of
linearly independent basis vectors $\vec{b}_1,\dots,\vec{b}_n \in \R^d$,
\[
\lat = \lat(\vec{b}_1,\ldots, \vec{b}_n) := \big\{ z_1 \vec{b}_1 + \cdots + z_n \vec{b}_n  \ : \ z_i \in \Z \big\}
\; .
\] We call $n$ the \emph{rank} of the lattice $\lat$ and $d$ the \emph{dimension} or the \emph{ambient dimension} of the lattice.

The two most important computational problems on lattices are the Shortest Vector Problem ($\SVP$) and the Closest Vector Problem ($\CVP$). Given a basis for a lattice $\lat \subset \R^d$,
$\SVP$ asks us to compute the minimal length of a non-zero vector in $\lat$, and $\CVP$ asks us to compute the distance from some target point $\vec{t} \in \R^d$ to the lattice. Typically, we define length and distance in terms of the $\ell_p$ norm for some $1 \leq p \leq \infty$, given by
\[
\|\vec{x}\|_p := (|x_1|^p + |x_2|^p + \cdots + |x_d|^p)^{1/p}
\]
for finite $p$ and
\[
\|\vec{x}\|_{\infty} := \max_{1 \leq i \leq d} |x_i|
\; .
\]
In particular, the case where $p = 2$ corresponds to the Euclidean norm, which is the most important and best-studied norm in this context. 
We write $\SVP_p$ and $\CVP_p$ for the respective problems in the $\ell_p$ norm. $\CVP$ is known to be at least as hard as $\SVP$ (in any norm, under an efficient reduction that preserves the rank, ambient dimension, and approximation factor)~\cite{GMSS99} and appears to be significantly harder.

In the past decade, these problems have taken on still more importance, as their hardness underlies the security of most post-quantum public-key cryptography schemes, while the schemes that are currently used for most practical applications are not secure against quantum computers. Recent rapid progress in quantum computing (e.g.,~\cite{AAB+QuantumSupremacy19}) has therefore created a rush to switch to lattice-based cryptography in many applications. Indeed, for this reason, lattice-based cryptography is in the process of standardization for widespread use~\cite{NIST_quantum}.

Given the obvious importance of these problems, they have been studied quite extensively. However, in spite of much effort, algorithmic progress has stalled for $\CVP$. The fastest algorithm for $\CVP_2$ runs in $2^{n+o(n)}$ time~\cite{ADS15}---even for arbitrarily large constant approximation factors---and there are fundamental reasons that our current techniques cannot do better.\footnote{There are only two known algorithms that solve $\CVP_2$ in its exact form in time $2^{O(n)}$~\cite{MV13, ADS15}, and both of them involve enumeration over all $2^n$ cosets of $\lat$ modulo $2\lat$. (These cosets arise naturally in this context, and they play a large role in Section~\ref{sec:limitations}.) There are other approaches that achieve constant-factor approximation in time $2^{O(n)}$, but the constant in the exponent is significantly larger. The situation for $\SVP$ is far more dynamic. See, e.g.,~\cite{BDGL16,ASJustTake18}.} For arbitrary $p$, the fastest known exact algorithm is still Kannan's $n^{O(n)}$-time algorithm from over thirty years ago~\cite{Kannan87}. For constant-factor approximation and arbitrary $p$, Bl{\"o}mer and Naewe~\cite{BN09} gave a $2^{O(d)}$-time algorithm, which was later improved to $2^{O(n)}$ time by Dadush~\cite{Dadush2012}, and a $4^{(1+\eps)d}$-time algorithm for $p = \infty$ by Aggarwal and Mukhopadhyay~\cite{AMFasterAlgorithms18}.

While we have known for decades that $\CVP_p$ is NP-hard~\cite{Boas81}, even to approximate up to superconstant approximation factors~\cite{DKRS03}, such coarse hardness results are insufficient to rule out, e.g., a $2^{n/20}$-time algorithm or even a $2^{\sqrt{n}}$-time algorithm. If such algorithms were found, they would have innumerable positive applications, but they would also render current lattice-based cryptographic constructions broken in practice. Even a relatively small improvement beyond $2^n$ time would have major consequences.

In~\cite{conf/focs/BennettGS17}, we therefore initiated the study of the \emph{fine-grained hardness} of CVP in an effort to explain this lack of algorithmic progress and to give evidence for the \emph{quantitative} security of lattice-based cryptography. We showed that there is no $2^{(1-\eps)n}$-time algorithm for $\CVP_p$ assuming the Strong Exponential Time Hypothesis (SETH, a common hypothesis in complexity theory, defined in Section~\ref{sec:prelims}), but we were only able to prove this lower bound explicitly for odd integers $p$ (and $p = \infty$). For other values of $p$, our result was much weaker. For every $\eps > 0$, we showed that there are at most finitely many $p \notin 2\Z$ with a $2^{(1-\eps)n}$-time algorithm for $\CVP_p$ (assuming SETH). 
In particular, for any \emph{specific} value of $p \notin (2\Z+1) \cup \{\infty\}$, we could not rule out such an algorithm. (We did, however, rule out $2^{o(n)}$-time algorithms for all $p$.)

We showed that the restriction $p \notin 2\Z$, though quite unfortunate, is in some sense inherent. Specifically, the main gadget that we used in our reduction does not exist for $p \in 2\Z$. However, the fact that our result had a non-explicit finite list of \emph{additional} exceptions seems to be an artifact of the proof techniques. And, we could not rule out some more general class of reductions that would work, e.g., for the most interesting case when $p = 2$.

Perhaps even more importantly, our results were far weaker for the \emph{approximate} variant of $\CVP_p$, in which the goal is to approximate the distance to the lattice up to some constant factor. In particular, like nearly all reductions to \emph{exact} $\CVP_p$, our reductions in~\cite{conf/focs/BennettGS17} produced rather unnatural $\CVP_p$ instances. In such instances, there are $2^n$ lattice points (corresponding to the $2^n$ possible assignments to a SAT formula) that are all essentially the same distance from the target, and the difficulty of the problem boils down entirely to determining whether any of these points is just \emph{slightly} closer than the others. One could argue that such artificial instances do not capture the \emph{geometric} spirit of $\CVP_p$. Certainly an algorithm that achieved a small constant-factor approximation would be essentially just as good as an exact algorithm for nearly all applications (including, e.g., for cryptanalysis). However, we were only able to rule out $2^{o(n)}$-time algorithms for the approximate version of the problem (under a conjecture known as Gap-ETH). So, one might worry that the problem becomes far easier for even relatively small approximation factors.

Finally, our lower bounds were quite weak for the problem of $\CVP_p$ with preprocessing ($\CVPP_p$), an offline-online variant of $\CVP_p$ where an unbounded-time preprocessing algorithm may perform arbitrary preprocessing on the lattice $\lat$ in a way that helps an online query algorithm to find a closest lattice vector to a given target $\vec{t} \in \R^d$. In~\cite{conf/focs/BennettGS17}, we were only able to rule out a $2^{o(\sqrt{n})}$-time algorithm for this problem. It therefore remained plausible that much faster algorithms could exist for $\CVPP_p$ than for $\CVP_p$ or for constant-factor approximate $\CVP_p$. Such algorithms would, for example, lead to very strong preprocessing attacks on certain lattice-based cryptographic schemes.

In follow-up work, we used the main result of~\cite{conf/focs/BennettGS17} to prove strong lower bounds for SVP~\cite{ASGapETH18}, for SIVP~\cite{ACNoteConcrete19} with Chung, and for BDD~\cite{BPHardnessBounded20} with Peikert. However, these works inherited some of the deficiencies described above. Specifically, the strongest hardness results in the first two works only applied to odd integers $p \in (2\Z+1)$ (and $p = \infty$) and some non-explicit set of additional $p$. (\cite{BPHardnessBounded20} was written after a preliminary version of this work was published, and therefore was able to take advantage of the stronger results that we describe below.) 

\subsection{Our results}

\paragraph{Hardness results in a nutshell. } We improve on the hardness results of~\cite{conf/focs/BennettGS17} in a number of ways. We extend the main hardness result in~\cite{conf/focs/BennettGS17} to \emph{all} $p$ except for the even integers, to approximate $\CVP_p$, \emph{and} to $\CVPP_p$. (See Table~\ref{tbl:fine-grained-results}. In the introduction, we sometimes informally refer to an ``approximate variant of SETH'' as ``Gap-SETH.'' There is no consensus definition for what the ``right'' version of this hypothesis is. See Definition~\ref{def:gap-seth} for one possible definition due to Manurangsi~\cite{Manurangsi2019}, which is in some sense the most conservative possible definition of ``Gap-SETH.'')

\begin{theorem}[Informal, see Corollary~\ref{cor:cvp-hardness} and Theorems~\ref{thm:stronger_gap_SETH_or_something} and~\ref{thm:cvpp-hardness-main}]
	\label{thm:CVP_hard_intro}
	For every $1 \leq p \leq \infty$ with $p \notin 2\Z$, there is no $2^{(1-\eps)n}$-time algorithm for $\CVP_p$ for any constant $\eps > 0$ unless SETH is false. The same conclusion holds for $\CVPP_p$ unless non-uniform SETH is false.
	
	Furthermore, for every $1 \leq p \leq \infty$ with $p \notin 2\Z$ and constant $\eps > 0$, there is no $2^{(1-\eps)n}$-time algorithm for $\gamma_{p, \eps}$-approximate $\CVP_p$ for some $\gamma_{p,\eps} > 1$ unless Gap-SETH is false.	
\end{theorem}

As in~\cite{conf/focs/BennettGS17}, our result is actually a bit stronger than the above. SETH-based hardness only requires a reduction from $k$-SAT to $\CVP_p$, but we show a reduction from Max-$k$-SAT, and even from weighted Max-$k$-SAT. In fact, we also rule out $2^{o(n)}$-time algorithms for $\CVPP_p$ under a weaker complexity-theoretic assumption: the (non-uniform) Exponential Time Hypothesis. This weaker lower bound under a weaker assumption holds for all $p \neq 2$---including even integers $p \geq 4$.\footnote{Theorem~\ref{thm:CVP_hard_intro} also yields immediate similar improvements to the hardness of $\SVP_p$ and $\mathsf{SIVP}_p$, i.e., to the results of~\cite{ASGapETH18,ACNoteConcrete19}. In particular, by the main results in~\cite{ACNoteConcrete19}, the $2^n$ hardness for $\CVP_p$ and its approximate variant immediately extends to $\mathsf{SIVP}_p$. The results for $\SVP_p$ are rather complicated, as they vary with $p$ in complex ways~\cite{ASGapETH18}, but our results imply extensions of~\cite{ASGapETH18} to more values of $p$ than were known previously. See Appendix~\ref{app:svp} for a complete statement of the result.}

\paragraph{Concrete(-ish) approximation factors. }
Perhaps the most important result in Theorem~\ref{thm:CVP_hard_intro} is the hardness of approximation. As we mentioned above, approximate $\CVP_p$ is a far more natural problem than \emph{exact} $\CVP_p$, and prior to this work, one might have worried that the approximate variant could be solved in much less than $2^n$ time, even for approximation factors $\gamma = 1+\eps$. Theorem~\ref{thm:CVP_hard_intro} shows that the approximate variant is hard too (under an appropriate conjecture), in the sense that one cannot solve $\gamma$-approximate $\CVP_p$ in time $2^{(1-\eps)n}$ for \emph{some} constant $\gamma > 1$. 

For the simplest form of our reduction, however, the resulting constant $\gamma$ is not very satisfying for two reasons.
First, our simplest proof is itself non-constructive in the sense that we show how to reduce $(1 + \eps)$-approximate Max-$k$-SAT to $\gamma(p, k, \eps)$-approximate $\CVP_p$, but the dependence of $\gamma$ on $p$ and $k$ is not explicit.
Second, if we reduce from approximate Max-$k$-SAT, then we must have $\gamma < 1/(1-2^{-k}) \approx 1 + 2^{-k}$ because Max-$k$-SAT can be trivially approximated up to a factor of $1/(1-2^{-k})$, and, not surprisingly, our techniques cannot give a fine-grained reduction from approximate Max-$k$-SAT to approximate $\CVP_p$ with a larger approximation factor. %
Putting these two issues together, we see that while $\gamma$ is certainly a constant for fixed $\eps > 0$, our result might only really kick in for, e.g., $\gamma < 1+2^{-1000}$.

We show how to get around both of these issues by instead reducing from approximate Max-$k$-\emph{Parity}, the problem of determining how many constraints of the form $x_{i_1} \oplus \cdots \oplus x_{i_k} = b$ can be satisfied simultaneously.

\begin{theorem}[Informal, see Theorem~\ref{thm:parity_reduction}]
	\label{thm:intro_hard_approx}
	For every $1 \leq p \leq \infty$ with $p \notin 2\Z$, integer $k > \max\{2,p\}$, and $0 < \eps < 1$, there is an efficient (Karp) reduction from $(1 + \eps)$-approximate Max-$k$-Parity on $n$ variables to $\gamma$-approximate $\CVP_p$ on rank $n$ lattices, where
		\[
			\gamma \approx 1 + \eps |\sin(\pi p/2)|/k^{(p+3)/2}
			\; .
		\]
\end{theorem}

Behind this result is a new identity concerning certain weighted sums over binomial coefficients (Theorem~\ref{thm:crazy_binomial}), which generalizes~\cite{stack_exchange_2827591}.
We note that a similar identity was previously used in~\cite{conf/soda/LiWW20} to show lower bounds on sketching problems.

Theorem~\ref{thm:intro_hard_approx} (1) gives explicit (quite reasonable) dependence of $\gamma$ on $p$, $k$, and $\eps$, and (2) this reduction is meaningful (see below) for, say, $\eps = 1/\poly(k)$ as opposed to $\eps \approx 1/2^k$ when reducing from Max-$k$-SAT, resulting in a corresponding approximation factor of $\gamma = 1 + 1/\poly(k)$ as opposed to $\gamma \approx 1 + 2^{-k}$ for fixed $p \notin 2\Z$. (In fact, by a known fine-grained reduction from approximate Max-$k$-SAT to approximate Max-$k$-Parity~\cite{SV19}, we can also leverage Theorem~\ref{thm:intro_hard_approx} to resolve the first issue when reducing from approximate Max-$k$-SAT. I.e., we can get an explicit bound on $\gamma$ in terms of $p$, $k$, and $\eps$ in that case as well; see Theorem~\ref{thm:stronger_gap_SETH_or_something}.)

We next discuss for which values of $\eps$ the above reduction is ``meaningful.''
For $k \geq 3$, Max-$k$-Parity is known to be NP-hard to approximate up to any constant approximation factor strictly less than two~\cite{journals/jacm/Hastad01}. On the other hand, the fastest known algorithm for $(1 + \eps)$-approximate Max-$k$-Parity for arbitrary $\eps > 0$ (and also Max-$k$-SAT, when $\eps = \eps(k)$ is very small) runs in time roughly $(2 - \sqrt{\eps})^n$~\cite{conf/soda/AlmanCW20}. If one were to hypothesize that the fastest possible algorithm for Max-$k$-Parity has a similar runtime --- i.e. of the form $(2 - 1/\poly(\eps))^n$ --- and choose $\eps = \eps(k) = 1/\poly(k)$ then Theorem~\ref{thm:intro_hard_approx} shows that $1 + 1/\poly(k)$-approximate $\CVP_p$ has no $(2 - 1/\poly(k))^n$-time algorithm. We emphasize again that, in contrast, $(1 + \eps)$-approximate Max-$k$-SAT is trivial for such $\eps = 1/\poly(k)$.

\paragraph{Hardness of proving better hardness.} The restriction that $p$ is not an even integer is unfortunate, especially because we are most interested in the case when $p = 2$. But, this seems inherent. (In fact, it is known that $\ell_2$ is ``the easiest norm'' in a certain precise sense~\cite{RR06}.) Indeed, in~\cite{conf/focs/BennettGS17}, we already showed that our specific techniques are insufficient to prove hardness for $p \in 2\Z$. 

Here, we also rule out a far more general class of techniques for $p = 2$, which we call ``natural reductions.'' These are reductions with a fixed mapping between witnesses. Specifically, a reduction from a $k$-SAT formula $\phi$ to $\CVP_p$ over a lattice with basis $\basis$ is natural if there is a fixed (not necessarily efficient) mapping $f : \{0,1\}^n \to \Z^{n'}$ such that $\basis \vec{z}$ is a closest lattice vector if and only if $\vec{z} = f(\vec{x})$, where $\vec{x} \in \{0,1\}^n$ is a satisfying assignment (assuming that $\phi$ is satisfiable). We also mention here the fact that natural reductions cannot prove better than $2^n$ hardness for $1 < p < \infty$. We include a simple proof of this fact in Section~\ref{sec:natural_intro}.

\begin{theorem}[Informal]
	\label{thm:limitations_intro}
	There is no natural reduction from $3$-SAT on $n$ variables to $\CVP_2$ on a lattice with rank $n' \leq 4(n-2)/3$. In particular, no natural reduction can rule out even a $2^{3n/4}$-time algorithm for $\CVP_2$ under SETH.
	
	Furthermore, for any $1 < p < \infty$, there is no natural reduction from $3$-SAT on $n$ variables to $\CVP_p$ on a lattice with rank $n' < n$. In particular, no natural reduction can rule out a $2^n$-time algorithm for $\CVP_p$ under SETH for $1 < p < \infty$.
\end{theorem}

Notice that we even rule out reductions from $3$-SAT to CVP. To prove SETH-hardness, we would need to show a reduction from $k$-SAT for all constant $k \geq 3$. Furthermore, we stress that this result also rules out such reductions from any problem that is provably at least as hard as $3$-SAT (under fine-grained natural reductions). This includes most ``reasonable'' Max-$3$-CSPs, such as Max-$3$-Parity (and of course Max-$k$-Parity for $k > 3$ as well). The essential obstruction is that the possible sets of closest vectors do not form an expressive enough class to capture $3$-SAT formulas.

Behind (the non-trivial $p=2$ part of) Theorem~\ref{thm:limitations_intro} are two new techniques. First is a new result concerning the structure of the closest lattice vectors to a target point in the $\ell_2$ norm. Specifically, we show that the structure of the closest vectors is quite rigid modulo $2\lat$. (See Lemma~\ref{lem:4to8}.) Second is a new and tighter proof of Szemer\'edi's cube lemma (Lemma~\ref{lem:additive_combinatorics}) for the boolean hypercube. We expect both of these results to be of independent interest.

\renewcommand{\arraystretch}{1.2}
\begin{table}[t]
	\begin{center}
		\begin{tabular}{|cl|cc|ccc|}
			\hline
			\multicolumn{2}{|c|}{Problem} & \multicolumn{2}{c|}{Upper bounds} & \multicolumn{3}{c|}{Lower bounds} \\ \hline
			&&Exact&Approximate& Exact & Approximate &Preprocessing \\
			\hline
			\multirow{3}{*}{$\CVP_p$}  & $p \notin 2\Z$ &$n^{O(n)}$&$2^{O(n)}$ & \blue{$2^{(1-\eps)n}$*} & \blue{$2^{(1-\eps)n}$} & \blue{$2^{(1-\eps)n}$}\\
			& $p \neq 2$ & $n^{O(n)}$&$2^{O(n)}$ & $2^{\Omega(n)}$ & $2^{\Omega(n)}$  & \blue{$2^{\Omega(n)}$}\\ 
			& $p = 2$ & $2^{n + o(n)}$&$2^{n + o(n)}$ & $2^{\Omega(n)}$ & $2^{\Omega(n)}$ &$2^{\Omega(\sqrt{n})}$ \\
			\hline
		\end{tabular}
	\end{center}
	\caption{\label{tbl:fine-grained-results} A summary of known quantitative upper and lower bounds under various assumptions on the complexity of $\CVP_p$ and $\CVPP_p$ for $p \in [1, \infty]$. New results appear in \blue{blue} (with a star next to the one result that is only novel for some $p$). Upper bounds for the approximate problems are for any constant approximation factor $\gamma > 1$, while lower bounds are for some small, explicit approximation factor $\gamma > 1$ depending on $p$ (and, in the case of $\CVP_p$ for $p \notin 2\Z$, also on $\eps >0$). The $2^{(1-\eps)n}$-time lower bounds are based on SETH (or Gap-SETH or non-uniform SETH), while the $2^{\Omega(\sqrt{n})}$-time and $2^{\Omega(n)}$-time lower bounds are based on ETH (or Gap-ETH or non-uniform ETH).}
\end{table}

\subsection{Our reductions}

The high-level idea behind our reductions (and those of~\cite{conf/focs/BennettGS17}) is as follows. 
The reduction is given as input a list $\phi_1,\ldots, \phi_m$ of $k$-clauses on $n$ boolean variables $x_1,\ldots, x_n$, where $k \geq 2$ is some constant. We wish to construct some basis $\basis \in \R^{d \times n}$ and target $\vec{t} \in \R^d$ such that for any $\vec{z}\in \Z^n$, $\|\basis \vec{z} - \vec{t}\|_p^p$ for $\vec{z} \in \Z^n$ is small if and only if $\vec{z}\in \{0,1\}^n$ represents an assignment that satisfies all of the $\phi_i$.

To that end, for each $\phi_i$, we wish to find a matrix $\Phi_i \in \R^{d' \times n}$ and target $\vec{t}_i\in \R^{d'}$ such that $\|\Phi_i \vec{z} - \vec{t}_i\|_p^p$ is small if and only if $z_{j_1},\ldots, z_{j_k} \in \{0,1\}$ represents an assignment that satisfies $\phi_i$. If we could find such matrices, we could take
\begin{equation}
\label{eq:Phi_t}
\basis := \begin{pmatrix}
	\Phi_1\\
	\Phi_2\\
	\vdots\\
	\Phi_m\\
	2\alpha I_n
\end{pmatrix} \in \R^{md' \times n} \qquad \qquad 
\vec{t} := \begin{pmatrix}
\vec{t}_1\\
\vec{t}_2\\
\vdots\\
\vec{t}_m\\
\alpha \vec{1}
\end{pmatrix}
\; ,
\end{equation}
where $\alpha \vec{1} \in \R^n$ is the vector whose coordinates are all $\alpha$.
Then, $\|\basis \vec{z}- \vec{t}\|_p^p = \sum_i \|\Phi_i \vec{z} - \vec{t}_i\|_p^p$ will be small if and only if $\vec{z} \in \{0,1\}^n$ corresponds to a satisfying assignment. (By taking $\alpha$ to be sufficiently large, we can guarantee that  any closest vectors must be of the form $\basis\vec{z}$ for $\vec{z} \in \{0,1\}^n$.)\

Since $\Phi_i \{0,1\}^n - \vec{t}_i = \{ \Phi_i \vec{z}- \vec{t}_i\ : \ \vec{z} \in \{0,1\}^n\}$ is a parallelepiped, and since the most important case (corresponding to $k$-SAT) is when all but one point in this set is long and all others are short, we call such objects \emph{isolating parallelepipeds}, as we explain below. The difficult step in these reductions is therefore to find isolating parallelepipeds $\Phi_i, \vec{t}_i$. We also naturally think of $\Phi_i \in \R^{d' \times k}$ by implicitly setting all entries in columns that do not correspond to the variables in $\phi_i$ to zero.

\paragraph{Finding isolating parallelepipeds. } We say that a parallelepiped $\Phi \{0,1\}^k - \vec{t}$ is a \emph{$(p, k)$-isolating parallelepiped} if all $\|\Phi \vec{z} - \vec{t}\|_p = 1$ for non-zero $\vec{z}\in \{0,1\}^k$ and $\|\Phi\vec0 - \vec{t}\|_p = \|\vec{t}\|_p > 1$. (We think of the vertex $-\vec{t}$ as ``isolated'' from the others. See Figure~\ref{fig:par}.) To find isolating parallelepipeds, we construct a family of parallelepipeds $\Phi, \vec{t}$ parameterized by $\alpha_1,\ldots, \alpha_{2^k} \geq 0$ and $t^* \in \R$ (see Figure~\ref{fig:param-ip-example}).
This family has the useful property that the norms $\|\Phi \vec{z} - \vec{t}\|_p^p$ are linear in the $\alpha_i$ for fixed $t^*$. (In~\cite{conf/focs/BennettGS17}, we used a less general family of parallelepipeds.)

\begin{figure}
	\begin{center}
		\begin{tabular}{m{4cm} m{0cm} m{4cm}}
			\hspace{-20ex} \includegraphics{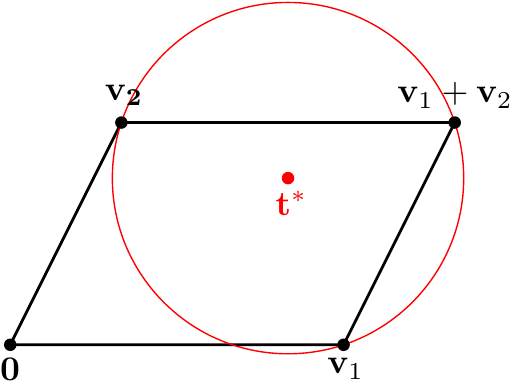} &  &  \includegraphics[scale=0.8]{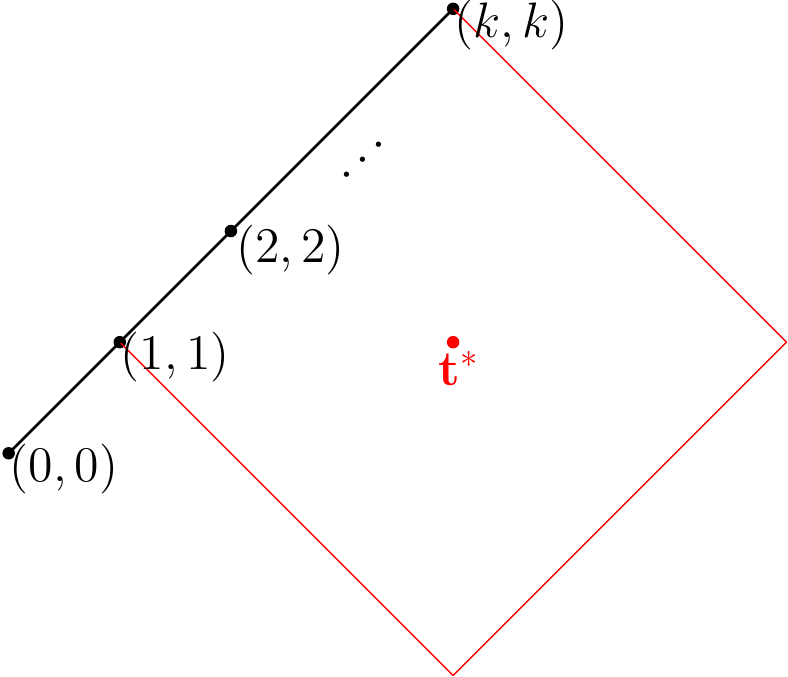}
		\end{tabular}
	\end{center}
	\caption{$(p, k)$-isolating parallelepipeds for $p = 2, k = 2$ (left) and $p = 1, k \geq 1$ (right). On the left, the vectors $\vec{v}_1$, $\vec{v}_2$, and $\vec{v}_1 + \vec{v}_2$ are all at the same distance from $\vec{t}^*$, while $\vec{0}$ is strictly farther away. 
		On the right is the degenerate parallelepiped generated by $k$ copies of the vector $(1,1)$. The vectors $(i, i)$ are all at the same $\ell_1$ distance from $\vec{t}^*$ for $1 \leq i \leq k$, while $(0, 0)$ is strictly farther away. 
		The (scaled) unit balls centered at $\vec{t}^*$ are shown in red, while the parallelepipeds are shown in black. (Figure taken from~\cite{conf/focs/BennettGS17}.)}
	\label{fig:par}
\end{figure}

So, finding isolating parallelepipeds essentially reduces to showing that a certain system of linear equations has a solution. (We actually need a non-negative solution, which is a major issue, but we ignore this for now.) To that end, we study the matrix $H_{k, p}(t^*) \in \R^{2^k \times 2^k}$ corresponding to this system of linear equations and try to show that its determinant is non-zero for some computable choice of $t^*$.
To do this, we observe that $H_{k, p}(t^*)$ satisfies the recurrence
\[
H_{k, p}(t^*) = \begin{pmatrix}
H_{k-1, p}(t^*-1) & H_{k-1, p}(t^*+1)\\
H_{k-1, p}(t^* + 1) & H_{k-1, p}(t^*-1)
\end{pmatrix}
\; .
\]
(It is this recurrence that makes this family more useful than the less general family in~\cite{conf/focs/BennettGS17}.)
This makes showing that $\det(H_{k, p}(t^*))$ is non-zero susceptible to a proof by induction on $k$. 

 To that end, we give formulas for the eigenvalues of $H_{k,p}(t^*)$ as functions of $t^*$. These functions are in turn each non-zero $\Z$-linear combinations of functions of the form $(t^* + \beta)^p$ for distinct $\beta \in \R$. (They are actually piecewise combinations of such functions, but we ignore this here.) We prove that such functions are $\R$-linearly independent if (and only if) either $p \geq k$ or $p \notin \Z$. Therefore, the eigenvalues cannot be identically zero as functions of $t^*$ for such $p$, which in turn implies that $\det(H_{k, p}(t^*))$ is not identically zero as a function of $t^*$, as needed. We finish the proof by noting that $\det(H_{k, p}(t^*))$ is (piecewise) analytic so that its zeros must be isolated, and it therefore has a computable non-zero point.

By combining this construction with our previous work, we completely characterize the values of $p$ and $k$ for which $(p, k)$-isolating parallelepipeds exist. Namely, the only case not handled by the construction above is the case where $p \in  \{1,\ldots, k-1\}$. In this case,~\cite{conf/focs/BennettGS17} showed that such parallelepipeds exist for odd $p$ but cannot exist for even $p < k$. (We provide a full proof of this latter claim in Lemma~\ref{lem:incl-excl}.) So, $(p,k)$-isolating parallelepipeds exist if and only if $p \notin \{2i \ :\ i < k/2\}$. 

As a corollary, we show a reduction from (weighted Max-)$k$-SAT on $n$ variables to a $\CVP_p$ instance with rank $n$ for all $p \notin \{2i \ :\ i < k/2\}$. In particular, we prove that $\CVP_p$ is SETH-hard for all $p \notin 2\Z$.

\paragraph{Hardness of approximation. } To prove hardness of approximation, we must show how to reduce an \emph{approximate} Max-$k$-SAT instance with $n$ variables to an approximate $\CVP_p$ instance with rank $n$. The $2^n$-hardness of approximate $\CVP_p$ described in Theorem~\ref{thm:CVP_hard_intro} then follows from the recent Gap-SETH conjecture of Manurangsi~\cite{Manurangsi2019}.

The construction shown in Eq.~\eqref{eq:Phi_t} is insufficient to prove hardness of approximation because the presence of the ``identity matrix gadget'' $2\alpha I_n$ forces the closest vector to be within distance roughly $\alpha n^{1/p}$ to the target. As a result, all SAT instances yield a $\CVP_p$ instance with $\dist_p(\vec{t}, \lat) \in (r, (1+O(1/n))r)$ for some radius $r \approx \alpha n^{1/p}$. 

To reduce to approximate $\CVP_p$, we therefore need to somehow remove this gadget, which we do by extending isolating parallelepipeds to ``isolating lattices.'' Specifically, we show how to construct a basis $\Phi \in \R^{d^* \times k}$ and target vector $\vec{t}^* \in \R^{d^*}$ such that $\Phi \vec{z}$ is a closest lattice vector to $\vec{t}^*$ if and only if $\vec{z} \in \{0,1\}^k$ and $\vec{z}$ corresponds to a satisfying assignment of the $k$-CNF $\phi$. I.e., while previously the satisfying assignments corresponded exactly to the closest vectors to $\vec{t}^*$ \emph{in the parallelepiped} $\Phi\{0,1\}^k$, now the satisfying assignments must correspond exactly to the closest vectors to $\vec{t}^*$ \emph{in the entire lattice} $\Phi \Z^k$. This eliminates the need for the identity matrix gadget.

We show a relatively straightforward reduction from isolating parallelepipeds to isolating lattices, which is enough to show a relatively weak hardness of approximation result, i.e., this allows us to reduce $(s,c)$-Gap-$k$-SAT (i.e., the problem of distinguishing between a $k$-SAT instance in which at least a $c$ fraction of the clauses are simultaneously satisfiable and one in which no assignment satisfies an $s$ fraction of the clauses) reduces to $\gamma(s,c, k)$-approximate CVP. However, applying this reduction to the above construction of isolating parallelepipeds is unsatisfying for two reasons: (1) the dependence of $\gamma(s,c,k)$ on $k$ is quite bad (the techniques described above do not even allow us to compute it explicitly, but it is relatively straightforward to see that $\gamma(s,c,k) \lesssim 1 + (c-s)/2^k$); and (2) because any $k$-SAT formula has an assignment satisfying at least a $(1-2^{-k})$-fraction of its clauses, $s > 	1-2^{-k}$ must rapidly approach one as $k$ increases.

We solve both of these problems by switching from Gap-$k$-SAT to another Gap-CSP: Gap-$k$-Parity, in which the input is $m$ constraints of the form $x_{i_1} \oplus \cdots \oplus x_{i_k} = b$ for $b \in \{0,1\}$, and the goal is to approximate the maximal number of simultaneously satisfiable clauses. This CSP is quite natural in this context because (1) H\r{a}stad showed that it is NP-hard to approximate up to any constant strictly less than $2$~\cite{journals/jacm/Hastad01}; and (2)~\cite{SV19} showed a fine-grained reduction from $(s,c)$-Gap-$k$-SAT to $(s',c')$-Gap-$k$-Parity with $s' = (1-2^{-k})s$ and $c' = (1-2^{-k})c$.

Furthermore, we are able to show that very good approximate parallelepipeds exist for Parity, i.e., for all positive integers $k$, there exists $\Phi \in \R^{d \times k}$ and $\vec{t}\in \R^d$ such that 
\begin{equation}
\label{eq:intro_parity_ip}
\| \Phi\vec{z} - \vec{t}\|_p^p = \begin{cases}
1 & \text{$\vec{z} \in \{0,1\}^k$ has odd Hamming weight}\\
1+\eps & \text{$\vec{z} \in \{0,1\}^k$ has even Hamming weight,}
\end{cases}
\end{equation}
(and vice-versa), where $\eps > 0$ is not too small. In particular, we can achieve $\eps \approx |\sin(\pi p/2)|/k^{(p+3)2}$. (This $\sin(\pi p/2)$ term is quite remarkable, as it elegantly accounts for the fact that our construction cannot possibly work for even $p$.)

To prove this, we study the eigenvalues of the matrix $H_{k, p}(t^*)$. We show that the symmetries of $H_{k, p}(t^*)$ imply that its eigenvectors correspond exactly to the output tables of the parity functions 
$\chi_S : \pmo^k \to \pmo$, $\chi_S(\vec{x}) := \prod_{i \in S} x_i$,
 and the eigenvalues are exactly the corresponding Fourier coefficients of a relatively simple function: $h : \pmo^k \to \R$, $h(\vec{x}) := \abs{\sum_{i=1}^k x_i - t^*}^p$. In particular, the parity function itself is equal to $\chi_{[k]}$, so that its output table is an eigenvector of $H_{k, p}$ with corresponding eigenvalue
\[
\lampar = 2\sum_{j = 0}^k (-1)^j \binom{k}{j} \abs{j-\tau}^p \ ,
\]
for $\tau := t^* - k/2$.
We would like to lower bound the absolute value of this sum, but notice that doing so seems non-trivial. E.g., it is not even clear whether it is positive, negative, or zero. (In fact, the sign of the sum for integer $\tau$ is $(-1)^{\tau + \floor{p/2} + 1}$, which is certainly not obvious.)

Using a contour integral, we give an explicit formula for this sum for integer values of $\tau$ (and large enough $k > p$), and in particular show that
\begin{equation}
\label{eq:intro_crazy_formula}
\sum_{i=0}^k (-1)^i \binom{k}{i} |i-\floor{k/2}|^p =  (-1)^{\floor{k/2}+1} \sin(\pi p/2)\binom{k}{\floor{k/2}} \cdot \beta_{p,k}
\; ,
\end{equation}
where $\beta_{p,k} > C_p$ converges to
\[
	C_p  := 2(2-2^{-p}) \frac{\Gamma(p+1)}{\pi^{p+1}} \zeta(p+1) \approx  (p/(\pi e))^p
\]
as $k \to \infty$.
Eq.~\eqref{eq:intro_crazy_formula} therefore allows us to understand the behavior of this sum quite precisely.
 (This formula is a generalization of the one appearing in~\cite{stack_exchange_2827591,conf/soda/LiWW20}. See Theorem~\ref{thm:crazy_binomial} and Corollary~\ref{cor:crazy_binomial}.)
 
 This allows us to explicitly describe a relatively simple parallelepiped satisfying Eq.~\eqref{eq:intro_parity_ip}, whereas previously we were only able to prove that such an object exists. (In terms of the construction described above, we set half the values of $\alpha_i$ to be zero and half to be one.) We can then directly compute $\eps$, which is given by the ratio of the eigenvalue computed in Eq.~\eqref{eq:intro_crazy_formula} to the largest eigenvalue. The largest eigenvalue is equal to the same sum without the alternating $(-1)^i$ term, which we show is equal to roughly $\binom{k}{\floor{k/2}} \cdot k^{(p+1)/2}$. So, the ratio is $
 \eps \approx |\sin(\pi p/2)|/k^{(p+1)/2}$. We therefore get an approximation factor of essentially $1+(c-s)\eps = 1+(c-s)|\sin(\pi p/2)|/k^{(p+3)/2}$ (losing an extra factor of $k$ in the conversion from an isolating parallelepiped to an isolating lattice), where in this context an approximation factor of $1$ means that the reduction fails.
  
A particularly striking feature of Eq.~\eqref{eq:intro_crazy_formula} is the term $\sin(\pi p/2)$. This term is of course zero if and only if $p\in 2\Z$, so that this quite neatly captures the fact that this construction does not (and cannot) work for $p \in 2\Z$. Furthermore, $\beta_{p,k}$ is monotonically increasing in $p$. So, in some sense this formula ``factors out'' the strange restriction that forces this sum to be zero when $p$ is an even integer.

\paragraph{Hardness of $\CVPP_p$. }
We next show how to extend the hardness result above from $\CVP_p$ to the Closest Vector Problems with Preprocessing in the $\ell_p$ norm ($\CVPP_p$). Namely, we show that $\CVPP_p$ is $2^n$-hard assuming (non-uniform) SETH for all $p \notin 2\Z$. To do this, we define an enhanced notion of an isolating parallelepiped, that we call an \emph{on-off-isolating parallelepiped} (this is analogous to what~\cite{SV19} does for codes). An on-off-isolating parallelepiped is an isolating parallelepiped $\Phi, \vec{t}^*$ together with a target $\toff$ such that $\|\Phi \vec{z} - \toff\|_p$ is constant for all $\vec{z} \in \{0,1\}^k$. 

To use these objects to reduce (Max-)$k$-SAT on $n$ variables to a $\CVPP_p$ instance with rank $n$, we must reduce $k$-SAT to $\CVP_p$ with a \emph{fixed} basis matrix $\basis_{n,k} \in \R^{d \times n}$. We use the matrix
\[
\basis_{n,k} := \begin{pmatrix} 
\Phi_1\\
\vdots\\
\Phi_{M}
\end{pmatrix}
\]
consisting of the on-off-isolating parallelepipeds for each possible $k$-clause on $n$ variables, stacked on top of each other, where $M := 2^k \binom{n}{k}$. Given a $k$-SAT formula $\{\phi_{i_1},\ldots, \phi_{i_m}\}$, we create the target 
\[
	\vec{t} := \begin{pmatrix}
		\vec{t}_1\\
		\vdots\\
		\vec{t}_{M}
	\end{pmatrix}
\]
such that $\vec{t}_i = \toff$ if $\phi_{i} \notin \{\phi_{i_1},\ldots, \phi_{i_m}\}$ and otherwise $\vec{t}_i = \vec{t}^*$. (We are oversimplifying a bit here. In our actual construction, we must shift $\toff$ in a way depending on which literals in the clause are negated. See Section~\ref{sec:cvpp-hardness}.) I.e., we use $\toff$ to ``turn off'' the clauses that do not appear in our SAT instance.

Finally, we show that $(p,k)$-on-off-isolating parallelepipeds exist if and only if $(p,k+1)$-isolating parallelepipeds exist. To transform a $(p,k+1)$-isolating parallelepiped $\Phi := (\Phi',\vec{\phi}_{k+1}), \vec{t}^*$ into a $(p,k)$-on-off-isolating parallelepiped, we simply take $\Phi'$, $\vec{t}^*$, and $\toff := \vec{t} - \vec{\phi}_{k+1}$. A simple calculation shows that $\|\Phi' \vec{z} - \toff\|_p = 1$ for all $\vec{z} \in \{0,1\}^k$ and $\|\Phi' \vec{z} - \vec{t}^*\|_p = 1$ for all non-zero $\vec{z} \in \{0,1\}^k$, as needed.

\subsection{Impossibility of natural reductions for \texorpdfstring{$p = 2$}{p=2}}

\label{sec:natural_intro}

In~\cite{conf/focs/BennettGS17}, we showed that the technique described above cannot work for even integers $p < k$. Specifically, we showed that isolating parallelepipeds do not exist in this case. However, this still left open the possibility of some other (potentially even simple) reduction from $k$-SAT to $\CVP_p$ for even integers $p$---perhaps even for $p = 2$. Here, we show that a very large class of reductions cannot work for $p = 2$. Behind these limitations is a new result concerning the structure of the closest lattice vectors to a target in the Euclidean norm.

Before we define natural reductions and show their limitations, we motivate the definition (and our techniques) by showing a simple limitation that applies for all $1 < p < \infty$. Specifically, we recall the well-known fact that for such $p$, the number of closest lattice vectors to a target is at most $2^{n'}$, where $n'$ is the rank of the lattice. (We show the simple proof of this fact below. Notice that $2^{n'}$ closest vectors are actually achieved by the integer lattice $\lat = \Z^{n'}$ and the all-halves target vector $\vec{t} = (1/2,\ldots, 1/2)$.) Therefore, if a reduction maps each satisfying assignment of some $3$-SAT formula to a distinct closest lattice vector, the rank $n'$ of the resulting lattice must be at least $\log_2 S$, where $S$ is the number of satisfying assignments. (Here, and below, we only consider the YES case, when there exists at least one satisfying assignment.) Since the number of satisfying assignments can be as large as $2^{n}$, where $n$ is the number of variables in the input instance, we immediately see that we must have $n' \geq n$.

Our specific reductions described above actually map each assignment $\vec{z} \in \{0,1\}^n$ to a very simple lattice vector: $\basis \vec{z}$. I.e., $\vec{z}$ is a satisfying assignment if and only if $\| \basis \vec{z} - \vec{t}\|_2 = r$. This suggests the following generalization of this type of reduction.

We call a reduction \emph{natural} if there exists a map $f$ from assignments $\vec{x} \in \{0,1\}^{n}$ to coordinate vectors $\vec{z}\in \Z^{n'}$ such that whenever the input $3$-SAT formula is satisfiable, $\|\basis \vec{z} - \vec{t} \|_2 = \dist_2(\vec{t}, \lat)$ if and only if $\vec{z} = f(\vec{x})$ for some satisfying assignment $\vec{x} \in \{0,1\}^n$. (We do not require $f$, or even the reduction itself, to be efficiently computable.)
Our reductions described above then correspond to the special case when $n = n'$ and $f$ is the identity map. 

Natural reductions are similar to \emph{parsimonious} reductions, which are efficient reductions that are required to preserve the number of witnesses between problems. However, natural reductions are more restrictive in the sense that $f$ must be instance independent.

\paragraph{Closest vectors mod two.} To rule out such reductions for $n' < 4n/3$, we study the algebraic and combinatorial properties of the set $S_{\basis, \vec{t}}$ of coordinates $\vec{z}\in \Z^{n'}$ of closest lattice vectors $\basis \vec{z}$ to some target vector $\vec{t}$. To motivate our techniques, let us first recall the well-known simple proof of the fact (mentioned above) that the number of closest vectors $|S_{\basis, \vec{t}}|$ is at most $2^{n'}$ for $1 < p < \infty$. Consider two distinct coordinates of closest vectors $\vec{z}_1, \vec{z}_2 \in \Z^{n'}$ to some target $\vec{t}$. Suppose that $\vec{z}_1 + \vec{z}_2 = 2\vec{z}$ for some integer vector $\vec{v} \in \Z^{n'}$. Then, $\|\basis \vec{v} - \vec{t}\|_p = \|(\basis \vec{z}_1 - \vec{t})/2 + (\basis \vec{z}_2 - \vec{t})/2\|_p < \|\basis\vec{z}_1 - \vec{t}\|_p/2 + \|\basis \vec{z}_2 - \vec{t}\|_p/2$, where we have used the \emph{strict convexity} of the $\ell_p$ norms for $1 < p < \infty$. (I.e., the triangle inequality $\|\vec{x} + \vec{y}\|_p \leq \|\vec{x} \|_p + \|\vec{y}\|_p$ is tight for $1 < p < \infty$ if and only if $\vec{y}$ is a scalar multiple of $\vec{x}$. Notice that this is false for $p = 1$ and $p =\infty$, and in each of these cases it is easy to show that there can be arbitrarily many closest lattice vectors to a target, even in two dimensions.)

The above proof does not \emph{only} show that the number of closest vectors is at most $2^{n'}$; it also shows that the set $S_{\basis, \vec{t}} \subset \Z^{n'}$ of coordinates of closest vectors in some basis $\basis$ has some algebraic structure. Specifically, there can be at most one element in $S_{\basis, \vec{t}}$ in each \emph{coset} of $\Z^{n'}/(2\Z^{n'})$. Here, a coset is the set $2\Z^{n'} + \vec{z}$ of all integer vectors with fixed coordinate parities. Notice that two cosets can be added together to obtain a new coset, $(2\Z^{n'} + \vec{z}_1) + (2\Z^{n'} + \vec{z}_2) = 2\Z^{n'} + (\vec{z}_1 + \vec{z}_2)$, and the above proof relied crucially on this structure (and specifically the fact that a coset summed with itself equals the zero coset). Of course, under addition, the cosets are isomorphic to $\F_2^{n'}$.
It is then natural to ask about the structure of $T_{\basis, \vec{t}} := S_{\basis, \vec{t}} \bmod 2$, viewed as a subset of the hypercube $\F_2^{n'}$.

Indeed, in Section~\ref{sec:limitations} we show the following curious property of $S_{\basis, \vec{t}}$ for $p = 2$. Let $C_2 \subset \F_2^{n'}$ be an affine square mod two (i.e., a two-dimensional affine subspace), and suppose that $C_2 \subseteq T_{\basis, \vec{t}}$. Let $C \subseteq S_{\basis, \vec{t}}$ be the set such that $C \bmod 2 = C_2$. (The above discussion shows that each element in $C_2$ has a unique preimage, so that $C$ is unique and $|C| = |C_2| = 4$.) Then, we show that either (1) the points in $C$ form a parallelogram over the reals (i.e., they must have the form $\vec{z}_1, \vec{z}_1 + \vec{z}_2, \vec{z}_1 + \vec{z}_3, \vec{z}_1 + \vec{z}_2 + \vec{z}_3$ over the reals, not just modulo $2$), or (2) there is a set of four other elements $C'$, uniquely determined by $C$, that must also lie in $S_{\basis, \vec{t}}$.

\paragraph{Studying the image of $f$.} To see how this can be used to rule out natural reductions, consider the image $A := f(\{0,1\}^n)$ of $f$ and $A_2 := A \bmod 2$. Suppose that $A_2$ contains an affine square $C_2 \subset A_2$, with $C \subset A$ such that $C = C_2 \bmod 2$. The fact that the set of closest vectors contains at most one element in each subset immediately implies that $|C| = 4$. Suppose that $C$ is not a parallelogram over the reals, and let $C'$ be the other four elements guaranteed by the above discussion. Then, let $E := f^{-1}(C) \subset \{0,1\}^n$ and $E' := f^{-1}(C') \subset \{0,1\}^n$ be the corresponding set of assignments. We observe that there exist $3$-SAT instances that are satisfied by all elements in $E$ but not all elements in $E'$. (This can be accomplished with a single clause.)  But, our reduction must map any such instance to a basis $\basis$ and a target $\vec{t}$ such that $C', C\subset S_{\basis, \vec{t}}$. This contradicts the assumption that $f$ only maps satisfying assignments to closest vectors.

Therefore, whenever $A_2$ contains an affine square $C_2$, the corresponding set $C$ in $A$ must be a parallelogram. It follows that any affine \emph{3-cube} in $A_2$ must correspond to a $3$-dimensional parallelepiped $P$ in $A$. Finally, we find a $3$-SAT instance satisfied by exactly seven of the eight elements in $f^{-1}(P)$. It follows that the reduction must produce a parallelepiped with exactly seven out of eight points closest to some target. In~\cite{conf/focs/BennettGS17}, we already showed that this is impossible. (We provide a simpler proof in Section~\ref{sec:limitations} as well.)

From this, we conclude that $A_2$ cannot contain any affine $3$-cube.

\paragraph{Using additive combinatorics to finish the proof. } Above, we observed that the image $A_2$ of $f$ modulo $2$ cannot contain any $3$-cube. But, we have already observed that $|A_2| = 2^n$ (i.e., the closest vectors must be distinct modulo $2$). So, $A_2 \subseteq \F_2^{n'}$ is some subset of $2^n$ points in $\F_2^{n'}$ that contains no affine hypercube. By Szemer\'edi's cube lemma, we must have $n' \geq 4n/3$, which is what we wished to prove. 

In fact, we only need a special case of Szemer\'edi's cube lemma. We provide a simpler proof of this special case based on the pigeon-hole principle. Though the proof is quite simple, to the authors' knowledge it is novel.

\subsection{Related work}

The most closely related work to this paper is of course~\cite{conf/focs/BennettGS17}. There are three additional papers showing fine-grained hardness of lattice problems:~\cite{ASGapETH18}, which showed such results for SVP;~\cite{ACNoteConcrete19}, which showed such results for SIVP; and~\cite{BPHardnessBounded20} which did the same for BDD. The first two of these works relied on the results in~\cite{conf/focs/BennettGS17}, and our improvements therefore immediately imply better hardness results for both SVP and SIVP. The third work was written after a preliminary version of this work appeared, and uses the results of this paper.

An additional line of work has shown different kinds of hardness for $\CVP$, $\SVP$, and related problems. In particular, Bhattacharyya, Ghoshal, Karthik, and Manurangsi showed the parameterized hardness of $\CVP$ and $\SVP$, as well as the analogous coding problems~\cite{BGMParameterizedIntractability18}.~\cite{SV19} showed tight hardness results for coding problems, using many ideas from~\cite{conf/focs/BennettGS17}. We in turn use some ideas from~\cite{SV19}, and in particular the idea of on-off-isolating parallelepipeds.

The work of Eisenbrand and Venzin~\cite{conf/esa/EisenbrandV20} gives a $2^{(0.802 + \eps)n}$-time algorithm for $\gamma$-$\CVP_p$ for constant $\gamma = \gamma(\eps)$ depending on $\eps > 0$. Their work combined with our work implies that (assuming Gap-SETH) there must be a time-approximation tradeoff for $\gamma$-$\CVP_p$ for $p \notin 2\Z$. In particular, their result shows that we cannot hope to get $2^{(1-\eps)n}$-hardness of $\gamma$-$\CVP_p$ for arbitrarily large constant $\gamma > 0$ and $p \not \in 2\Z$.

Finally, as mentioned earlier,~\cite{conf/soda/LiWW20} uses the bound in Equation~\eqref{eq:intro_crazy_formula} to show lower bounds on a natural sketching problem in $\ell_p$ norms. Interestingly, because this quantity vanishes for $p \in 2\Z$, both the present work and~\cite{conf/soda/LiWW20} are unable to show certain lower bounds for such $p$.

\subsection{Open questions}

The most obvious question that we leave open is, of course, to prove similar $2^n$ hardness results for $\CVP_2$, and more generally, for $\CVP_p$ for even integers $p$. In the $p = 2$ case, we show that any such proof (via SETH) would have to use an ``unnatural reduction.''
So, a fundamentally different approach is needed.\footnote{We note that the main reduction in~\cite{conf/focs/BennettGS17} works as a (natural) reduction from weighted Max-$2$-SAT formulas on $n$ variables with arbitrary (possibly exponential) weights to $\CVP_p$ instances of rank $n$ for all $p \in [1, \infty)$, including $p = 2$. So, a $2^{(1-\eps)n}$-time algorithm for $\CVP_2$ would imply a $2^{(1-\eps)n}$-time algorithm for weighted Max-$2$-SAT with arbitrary weights, for which no such algorithm is known. (Ryan Williams' algorithm for Max-$2$-SAT~\cite{journals/tcs/Williams05} runs in $W \cdot 2^{\omega n/3 + o(n)}$-time, where $W$ is the largest weight of a clause and $\omega < 2.374$ is the matrix multiplication constant.) So, there is already (rather weak) evidence that there is no $2^{(1-\eps)n}$-time algorithm for $\CVP_2$.}
One potentially promising direction would be to find a Cook reduction, as our limitations only apply to Karp reductions.
Another direction would be to show somewhat weaker hardness (say, $2^{n/2}$-hardness) of $\CVP_2$ assuming SETH using natural reductions. (Our limitations only apply to showing $2^{3n/4}$ or better hardness.)
Yet another potential direction would be to reduce directly to approximate $\CVP_2$ (presumably from a GapCSP). Our limitations show that the set of exact closest vectors cannot be as expressive as $3$-SAT formulas, but it says nothing about sets of ``nearly closest'' vectors.

Another potentially easier problem would be to show hardness of $\CVP_p$ in terms of the ambient dimension $d$, rather than $n$. Indeed, though there do exist $2^{O(n)}$-time constant-factor approximation algorithms for $\CVP_p$, the parameter $d$ is in some sense more natural. (E.g., the original algorithm of~\cite{BN09} runs in time $2^{O(d)}$, and the algorithm of~\cite{AMFasterAlgorithms18} also has its running time in terms of $d$.) This problem is potentially easier than the above because for $p = 2$ we may assume without loss of generality that $n = d$.

Of course, another open question is to prove stronger quantitative lower bounds for $\SVP_p$, and in particular for $\SVP_2$. While~\cite{ASGapETH18} did prove quite strong lower bounds for sufficiently large $p$, their bounds for small $p$ and in particular for $p = 2$ are quite weak.

We also note that $\CVP_p$ for $p \neq 2$ has received relatively little attention from an algorithmic perspective. In particular, there has not been much work trying to optimize the hidden constants in the exponent in the running times of $2^{O(n)}$ or $2^{O(d)}$ of the best known algorithms for constant-factor approximate $\CVP_p$. Our lower bounds provide new motivation for work on this subject. In particular, we ask whether our lower bounds are tight.

In fact, we do not expect our lower bound to be tight in the case when $p = \infty$.  (Recall that our limitation in Theorem~\ref{thm:limitations_intro} does not apply to $p = 1$ or $p = \infty$.) Indeed, because the kissing number in the  $\ell_\infty$ norm is $3^n-1$, one might guess that the fastest algorithms for $\CVP_\infty$ and $\SVP_\infty$ actually run in time $3^{n + o(n)}$ or perhaps $3^{d + o(d)}$. (See~\cite{AMFasterAlgorithms18}, which more-or-less achieves this.) We therefore ask whether stronger lower bounds can be proven in this special case.

We also note that our results only apply for \emph{exact} $\CVP_p$ or $\CVP_p$ with a rather small constant approximation factor. For cryptographic applications, one is interested in much larger approximation factors, typically approximation factors polynomial in $n$ (though the fastest known algorithms for these approximate problems work by solving smaller exact or near-exact instances of $\SVP_p$). While there are strong complexity-theoretic barriers to proving hardness in that regime, one might still hope to prove fine-grained hardness results for larger approximation factors---such as large constants or even superconstant. Indeed, we know NP-hardness up to an approximation factor of $n^{c/\log \log n}$, but this result is not fine-grained~\cite{DKRS03}.

Our work further motivates the emerging study of \emph{fine-grained hardness of approximation}. In particular, we wish to draw attention to the question of finding the ``right'' notion of Gap-SETH. Manurangsi's version~\cite{Manurangsi2019}, presented here in Definition~\ref{def:gap-seth}, is quite beautiful and natural, and we suspect that it will have many additional applications in the study of fine-grained hardness of approximation. However, what makes it so natural is that it is in some sense the weakest possible form of such a hypothesis (e.g., any hypothesis of the same form for any Gap-$k$-CSP implies Manurangsi's hypothesis). In particular, the order of quantifiers makes it difficult to use this hypothesis to prove hardness of approximation for specific constant approximation factors. So, perhaps a stronger hypotheses (or families of hypotheses) should be explored. The results of this work and those of~\cite{SV19} show that a hypothesis about Gap-$k$-Parity could prove useful, but we do not attempt to formalize this or claim that this is the ``right'' notion.

A final open question is to show $2^{\Omega(n)}$-hardness of $\CVPP_2$ assuming non-uniform ETH. 
The proof techniques in Section~\ref{sec:cvpp-hardness} show such hardness for $\CVPP_p$ for all $p \neq 2$ (including even integers $p$ greater than $2$), but for $p = 2$ the $2^{\Omega(\sqrt{n})}$-hardness shown in~\cite{conf/focs/BennettGS17} remains the best known. (For~$\CVP_p$, such $2^{\Omega(n)}$-hardness for all $p$, including $p = 2$, assuming ETH is known.)
\subsection*{Acknowledgments}

We would like to thank the Bertinoro program on \href{http://www.wisdom.weizmann.ac.il/~robi/Bertinoro2019_FineGrained/}{Fine Grained Approximation Algorithms and Complexity} at which some of this work was completed. We are also
grateful to the anonymous reviewers for their helpful comments.
\section{Preliminaries}
\label{sec:prelims}

Throughout this paper, we work with lattice problems over $\R^d$ for convenience. As usual, to be formal we must pick a suitable representation of real numbers and consider both the size of the representation and the efficiency of arithmetic operations in the given representation. But, we omit such details throughout to ease readability. We write $\Re(x)$ and $\Im(x)$ for the real part and imaginary part of $x \in \C$ respectively.
We will use boldfaced variables to denote column vectors, but will occasionally abuse notation by writing things like $\vec{v} = (\vec{u}, \vec{w})$ instead of $\vec{v} = (\vec{u}^T, \vec{w}^T)^T$.

\subsection{Lattice problems}

Let $\dist_p(\lat, \vec{t}) := \min_{\vec{x} \in \lat} \norm{\vec{x} - \vec{t}}_p$ denote the $\ell_p$ distance of $\vec{t}$ to $\lat$. 
We next formally define the lattice problems that we consider.

\begin{definition}
For any $\gamma \geq 1$ and $1 \leq p \leq \infty$, the $\gamma$-approximate Shortest Vector Problem with respect to the $\ell_p$ norm ($\gamma$-$\SVP_p$) is the promise problem defined as follows. Given a lattice $\lat$ (specified by a basis $B \in \R^{d \times n}$) and a number $r > 0$, distinguish between a `YES' instance where there exists a non-zero vector $\vec{v} \in \lat$ such that $\norm{\vec{v}}_p \leq r$, and a `NO' instance where $\norm{\vec{v}}_p > \gamma r$ for all non-zero $v \in \lat$.
\label{def:svp}
\end{definition}

\begin{definition}
For any $\gamma \geq 1$ and $1 \leq p \leq \infty$, the $\gamma$-approximate Closest Vector Problem with respect to the $\ell_p$ norm ($\gamma$-$\CVP_p$) is the promise problem defined as follows. Given a lattice $\lat$ (specified by a basis $B \in \R^{d \times n}$), a target vector $\vec{t} \in \R^d$, and a number $r > 0$, distinguish between a `YES' instance where $\dist_p(\lat, \vec{t}) \leq r$, and a `NO' instance where $\dist_p(\lat, \vec{t}) > \gamma r$.
\label{def:cvp}
\end{definition}

\noindent When $\gamma = 1$, we simply refer to the problems as $\SVP_p$ and $\CVP_p$.

\begin{definition}
The Closest Vector Problem with Preprocessing with respect to the $\ell_p$ norm ($\CVPP_p)$ is the problem of finding a preprocessing function $P$ and an algorithm $Q$ which work as follows. Given a lattice $\lat$ (specified by a basis $B \in \R^{d \times n}$), $P$ outputs a new description of $\lat$. Given $P(\lat)$, a target vector $\vec{t} \in \R^d$, and a number $r > 0$, $Q$ decides whether $\dist_p(\lat, \vec{t}) \leq r$.
\label{def:cvpp}
\end{definition}

When we measure the runtime of a $\CVPP$ algorithm, we only count the runtime of $Q$, and not of the preprocessing algorithm $P$. We will assume that the runtime of $Q$ is at least the size of the preprocessing, $|P(L)|$.

\subsection{Isolating parallelepipeds}

We recall the definition of an \emph{isolating parallelepiped} from~\cite{conf/focs/BennettGS17}. See Figure~\ref{fig:par}.
\begin{definition}
For any $1 \leq p \leq \infty$ and integer $k \geq 1$, we say that $V \in \R^{d^* \times k}$ and $\vec{t}^* \in \R^{d^*}$ define a \emph{$(p, k)$-isolating parallelepiped} if:
\begin{enumerate}
\item $\norm{V\vec{x} - \vec{t}^*}_p = 1$ for all $\vec{x} \in \bit^k \setminus \set{\vec{0}}$,
\item $\norm{\vec{t}^*}_p > 1$.
\end{enumerate}
\label{def:ip}
\end{definition}

We will more generally refer to the set $V \cdot \bit^k - \vec{t}^*$ for $V \in \R^{d^* \times k}$ and $\vec{t}^* \in \R^{d^*}$ as a \emph{$k$-parallelepiped}. We call a $2$-parallelepiped a \emph{parallelogram}.

\subsection{Constraint Satisfaction Problems}
A \emph{$k$-constraint} is a boolean function $C: \bit^k \to \bit$. 
A $k$-\emph{Constraint Satisfaction Problem} ($k$-CSP) $\mathcal{C}$ is specified by a set of $k$-constraints $\mathcal{C} = \set{C_1, \ldots, C_r}$. An instance $\Phi$ of a $k$-CSP $\mathcal{C}$ on $n$ variables $x_1, \ldots, x_n$ consists of $m$ $k$-constraints $C_1, \ldots, C_m \in \mathcal{C}$, where each constraints $C_i$ has $k$ (not necessarily distinct) variables $x_{i,1}, \ldots, x_{i,k}$ of $\Phi$ as its input variables. An assignment $\vec{y} \in \bit^n$ to the variables of $\Phi$ \emph{satisfies} constraint $C_i$ if $C_i(y_{i,1}, \ldots, y_{i,k}) = 1$, and satisfies $\Phi$ if it satisfies all of the constraints $C_1, \ldots, C_m$ of $\Phi$.
Let $\val(\Phi)$ denote the maximum fraction of constraints of $\Phi$ satisfiable by some assignment $\vec{y}$.

\begin{definition}
Let $\mathcal{C}$ be a $k$-CSP. The $(s, c)$-Gap-$\mathcal{C}$-CSP problem for $0 \leq s \leq c \leq 1$ is the promise problem defined as follows. On input an instance of $\Phi$, the goal is to distinguish between a YES instance in which $\val(\Phi) \geq c$, and a NO instance in which $\val(\Phi) < s$.
\end{definition}

We will primarily consider two CSPs in this work: (1) $k$-SAT, which consists of the $2^k$ functions $C : \bit^k \to \bit$ with exactly $2^k - 1$ satisfying assignments (equivalently, where each constraint is the disjunction of $k$ variables and negated variables), and (2) $k$-Parity, where $\mathcal{C}$ consists of the two constraints $C_0(x_1, \ldots, x_k) := x_1 \oplus x_2 \oplus \cdots \oplus x_k = 0$ and $C_1(x_1, \ldots, x_k) := x_1 \oplus x_2 \oplus \cdots \oplus x_k = 1$.

When a formula $\Phi$ is clear from context, we will write $m^+(\vec{y})$ to denote the number of constraints of $\Phi$ satisfied by the assignment $\vec{y}$.

Finally, we will need the following one of the main results of~\cite{SV19}. (This is actually a slight modification of the original theorem, but it is clear that the proof yields this modified version as well.)

\begin{theorem}[{\cite[Theorem 4.2]{SV19}}]
	\label{thm:SV19}
	For any integer $k \geq 2$ and $1-2^{-k} < s \leq c \leq 1$, there is a polynomial-time (Karp) reduction from $(s,c)$-Gap-$k$-SAT on $n$ variables to $(s',c')$-Gap-$k$-Parity on $n$ variables, where
\[
		s' := \frac{2^{k-1}}{2^k - 1} s
		\; ,
	\]
	and
	\[
		c' := \frac{2^{k-1}}{2^k - 1} c
		\; .
	\]
\end{theorem}

\paragraph{\texorpdfstring{$k$}{k}-SAT.} We next introduce some notation specific to $k$-SAT.
Let $\Phi$ be a $k$-SAT formula on $n$ variables $x_1, \ldots, x_n$ and $m$ clauses $C_1, \ldots, C_m$ (where each clause represents a constraint, when viewing $k$-SAT as a $k$-CSP). Let $\ind(\ell)$ denote the index of the variable underlying a literal $\ell$. I.e., $\ind(\ell) = j$ if $\ell = x_j$ or $\ell = \lnot x_j$. Call a literal $\ell$ \emph{positive} if $\ell = x_j$ and \emph{negative} if $\ell = \lnot x_j$ for some variable $x_j$.
Given a clause $C_i = \lor_{s=1}^k \ell_{i, s}$, let $P_i := \set{s \in [k] : \ell_{i, s} \textrm{ is positive}}$ and let $N_i := \set{s \in [k] : \ell_{i, s} \textrm{ is negative}}$ denote the indices of positive and negative literals in $C_i$ respectively. Given an assignment $\vec{y} \in \{0,1\}^n$ to the variables of $\Phi$, let $S_i(\vec{y})$ denote the indices of literals in $C_i$ satisfied by $\vec{y}$. I.e., $S_i(\vec{y}) := \set{s \in P_i : a_{\ind(\ell_{i,s})} = 1} \cup \set{s \in N_i : a_{\ind(\ell_{i,s})} = 0}$.

\subsection{Hardness assumptions}

\begin{definition}[SETH; {\cite{journals/jcss/ImpagliazzoPZ01}}]
	For every $\eps > 0$ there exists a $k = k(\eps) \in \Z^+$ such that no algorithm solves $k$-SAT on $n$ variables in $2^{(1-\eps)n}$ time.
\end{definition}

In his Ph.D. thesis, Manurangsi~\cite{Manurangsi2019} gave one possible definition of Gap-SETH.
\begin{definition}[Gap-SETH; {\cite[Conjecture 12.1]{Manurangsi2019}}]
For every $\eps > 0$ there exist $k = k(\eps) \in \Z^+$ and $\delta = \delta(\eps) > 0$ such that there is no algorithm that can distinguish between a $k$-SAT formula with $n$ variables that is satisfiable and one that has value less than $1 - \delta$ in $2^{(1-\eps)n}$ time.
\label{def:gap-seth}
\end{definition}
We will show that $\CVP_p$ cannot be approximated to within some factor $\gamma_{\eps} > 1$ in $2^{(1-\eps)n}$ time assuming Gap-SETH. Unfortunately, $\gamma_{\eps}$ decays as a function of $\eps$. However, our reduction from Gap-$k$-SAT to $\CVP_p$ can be adapted to a reduction from \emph{any} Gap-$k$-CSP to $\CVP_p$ with the same relevant parameters. (Namely, our reduction maps CSP instances on $n$ variables to CVP(P) instances of rank $n$.)

We will also use non-uniform variants of ETH and SETH to prove hardness results about $\CVPP_p$.

\begin{definition}[Non-uniform ETH]
There is no family of circuits of size $2^{o(n)}$ that solves $3$-SAT instances on $n$ variables.
\end{definition}

\begin{definition}[Non-uniform SETH]
For every $\eps > 0$ there exists a $k = k(\eps) \in \Z^+$ such that no family of circuits of size $2^{(1-\eps)n}$ solves $k$-SAT instances on $n$ variables.
\end{definition}

Our results are also quite robust to how we define non-uniform (S)ETH. For example, one of our main results about the complexity of $\CVPP_p$ roughly says that assuming non-uniform ETH (as stated above) there is no subexponential-sized family of circuits that decides $\CVPP_p$ for $p \neq 2$. However, if we were to change non-uniform ETH to say that there is no $2^{o(n)}$-time algorithm using $\poly(n)$ advice, then we would get a corresponding statement for $\CVPP_p$: that there is no $2^{o(n)}$-time algorithm for $\CVPP_p$ using $\poly(n)$ advice.

Interestingly, many of our results only depend on weaker versions of these hypotheses, where we replace an assumption about the hardness of $k$-SAT with an assumption about the hardness of Max-$k$-SAT or even weighted Max-$k$-SAT.

\subsection{Linear algebra}
We recall that an \emph{affine $k$-cube} in $\F_2^n$ is $\{ \vec{y}_0 + \sum_{j \in W} \vec{y}_j \ : \ W \subseteq \{1,\ldots, k\}\}$ for some $\vec{y}_0 \in \F_2^n$ and linearly independent $\vec{y}_1,\ldots, \vec{y}_k \in \F_2^n$.

We say that functions $f_0, \ldots, f_n : \R \to \R$ are linearly independent over the reals if given $a_0, \ldots, a_n \in \R$, the sum $\sum_{i=0}^n a_i f_i(x)$ is identically zero (is equal to $0$ for all $x \in \R$) only if $a_0 = \cdots = a_n = 0$. We say that $f \in C^k$ if the first $k$ derivatives of $f$ exist and are continuous, $f \in C^{\infty}$ if $f$ has derivatives of all orders, and that $f$ is \emph{analytic} if $f \in C^{\infty}$ and if the Taylor series of $f$ expanded around any point $x$ in the domain converges to $f$ in some neighborhood of $x$. We say that $f \in C^k(a, b)$ if the first $k$ derivatives of $f$ exist and are continuous on the (open) interval $(a, b)$ (we define $f \in C^{\infty}(a, b)$ and $f$ being analytic on $(a, b)$ analogously).

\begin{definition}
We define the \emph{Wronskian} of $f_0, \ldots, f_n \in C^n(a, b)$ to be $\det(M)$, where $M$ is the $(n + 1) \times (n + 1)$ matrix defined by
\[
M := \begin{pmatrix}
f_0(x) & f_1(x) & \cdots & f_n(x) \\
\frac{d}{dx} f_0(x) & \frac{d}{dx} f_1(x) & \cdots & \frac{d}{dx} f_n(x) \\
\vdots & \vdots & \ddots & \vdots \\
\frac{d^n}{dx^n} f_0(x) & \frac{d^n}{dx^n} f_1(x) & \cdots & \frac{d^n}{dx^n} f_n(x) \\
\end{pmatrix}
\ 
\]
for $x \in (a, b)$.
\label{def:wronskian}
\end{definition}

Because the derivative is a linear operator, we have the following.

\begin{fact}
Functions $f_0, \ldots, f_n$ are linearly independent over the reals if their Wronskian exists and is not identically zero on some interval $(a, b)$. 
\label{fct:wronskian-lin-ind}
\end{fact}

\subsection{Discrete Fourier analysis}
\label{subsec:discrete-fourier}

We will use several basic concepts from discrete Fourier analysis. We briefly review these concepts here; see~\cite{books/daglib/0033652} for a comprehensive survey.

The goal of discrete Fourier analysis is to analyze boolean functions $f: \pmo^k \to \R$ by representing them as multilinear polynomials. Every such function $f$ has such a representation, called its \emph{Fourier expansion}:
\[
f(\vec{x}) = \sum_{S \subseteq [k]} \fh(S) \cdot \chi_S(\vec{x}) \ .
\]
Here the functions $\chi_S(\vec{x}) := \prod_{i \in S} x_i$ are called \emph{Fourier characters}, and the values $\fh(S) \in \R$ are their corresponding \emph{Fourier coefficients}. The Fourier characters form an orthonormal basis of the function space $\mathcal{F} := \set{f : (f: \pmo^k \to \R)}$ equipped with the inner product $\iprod{f, g} := \E_{\vec{x} \sim \pmo^k}[f(\vec{x})g(\vec{x})]$.

The function space $\mathcal{F}$ is isomorphic as a vector space to $\R^{2^k}$ by the mapping $f \mapsto (f(\vec{x}))_{\vec{x} \in \pmo^k}$, where $(f(\vec{x}))_{\vec{x} \in \pmo^n}$ is the vector representing the output table of $f$. Applying this mapping to the Fourier characters $\chi_S$ results in vectors
\begin{equation}
\vec{v}_S := (\prod_{i \in S} x_i)_{\vec{x} \in \pmo^k} \ ,
\label{eq:vs}
\end{equation}
which (up to scaling) in turn form an orthonormal basis of $\R^{2^k}$ equipped with the standard inner product. Applying this mapping to a boolean function yields
\[
f(x) = \sum_{S \subseteq [k]} \fh(S) \cdot \chi_S(\vec{x}) \mapsto \sum_{S \subseteq [k]} \fh(S) \cdot \vec{v}_S \ .
\]

The key property of the Fourier characters that we will need is the recurrence relation
\begin{equation}
\label{eq:fourier-recurrence}
	\{ \vec{v}_{T} \ : \ T \subset [k] \} = \{(\vec{v}_{S}, \pm \vec{v}_{S})  \ : \ S \subset [k-1]\}
	\; ,
\end{equation}
which can be verified by inspection.

\subsection{The gamma function}

For $x,y  \in \C$, we adopt the convention that $x^{y} := \exp(\log(x) y)$ and $\sqrt{x} = \exp(\log(x)/2)$, where $\log(x)$ is the principal branch of the logarithm, satisfying $- \pi < \Im(\log(x)) \leq \pi$. 

The $\Gamma$ function is defined as 
\[
	\Gamma(x) := \int_0^\infty t^{x-1} \exp(-t) {\rm d} t
	\; .
\] 
This integral converges for $\Re(x) > 0$, and the function can be analytically continued to the entire complex plane except non-positive integer $s$, where the function has a simple pole. The inverse of the $\Gamma$ function, $1/\Gamma(x)$ is an entire function with zeros at all non-positive integers. The $\Gamma$ function satisfies the functional equation $\Gamma(x+1) = x\Gamma(x)$. In particular, for positive integers $n$, $\Gamma$ satisfies $\Gamma(n) = (n-1)!$.

We will need the following striking identity due to Ramanujan (see, e.g.,~\cite[Page 2]{MORFormulasTheorems66}).

\begin{theorem}
	\label{thm:ramanujan}
	For any positive integer $k$ and $x \in \C$ with $x \notin \{i,2i,\ldots, ki\}$,
	\[
		\frac{\Gamma(k+1)^2}{\Gamma(k+ix+1)\Gamma(k-ix+1)} =  \frac{\sinh(\pi x)}{\pi x} \prod_{j=1}^k (1+x^2/j^2)^{-1}
		\; .
	\]
	In particular, this quantity is positive and monotonically decreasing in $k$ for real $x \neq0$.
\end{theorem}

\section{Isolating parallelepipeds in \texorpdfstring{$\ell_p$}{l\_p} norms for all non-integer \texorpdfstring{$p$}{p}}
\label{sec:other-p}

Our first new result is a strengthening of a result in~\cite{conf/focs/BennettGS17}, which asserts that for every fixed $k \in \Z^+$ there exist $(p, k)$-isolating parallelepipeds for \emph{almost every} $p \in [1, \infty) \setminus 2\Z$, to a result showing that this is true for \emph{every} $p \in [1, \infty) \setminus 2\Z$. We also show that there exist $(p, k)$-isolating parallelepipeds when $k \leq p$. Moreover, we show that these are the only cases in which isolating parallelepipeds exist, and we therefore obtain a complete characterization of the values of $p$ and $k$ for which these objects exist. (Furthermore, our isolating parellelepipeds are computable if $p$ is computable.)

Our construction generalizes the approach from~\cite{conf/focs/BennettGS17}, and follows the same high-level structure. We start by showing that it suffices to ``define isolating parallelepipeds over $\pmo$ instead of $\bit$,'' i.e., that if there exist $V = (\vec{v}_1, \ldots, v_k) \in \R^{d \times k}$ and $\vec{t}^* \in \R^d$ that satisfy $\norm{V\vec{y} - \vec{t}^*}_p = 1$ for $\vec{y} \in \pmo^k \setminus \set{-\vec{1}}$ and $\norm{V(-\vec{1}) - \vec{t}^*}_p > 1$, then there exists a $(p, k)$-isolating parallelepiped.

We then define a family of $k$-parallelepipeds $V \in \R^{2^k \times k}, \vec{t}^* \in \R^{2^k}$ parameterized by $2^k$ numbers, $\alpha_{\vec{u}} \ge 0$ for $\vec{u} \in \pmo^k$, and a number $t^*$. Specifically, the row of $V$ indexed by $\vec{u} \in \pmo^k$ is equal to $\alpha_{\vec{u}}^{1/p} \cdot \vec{u}^T$ and the coordinate of $(\vec{t}^*)_{\vec{u}} = \alpha_{\vec{u}}^{1/p} \cdot t^*$.
(Throughout this section, we will adopt the convention that vectors $\vec{v} \in \R^{2^k}$ for some $k \in \Z^+$ are indexed by elements in $\pmo^k$ in lexicographic order. We adopt an analogous convention for rows (resp. columns) of matrices of the form $M \in \R^{2^k \times m}$ (resp. $M \in \R^{m \times 2^k}$) for some $m$.)
Figure~\ref{fig:param-ip-example} shows the form of such a $k$-parallelepiped when $k = 3$.

We observe that for such a family of $k$-parallelepipeds and $\vec{y} \in \pmo^k$, $\norm{V\vec{y} - \vec{t}^*}_p^p = \sum_{\vec{u}} \alpha_{\vec{u}} \abs{\iprod{\vec{u}, \vec{y}} - t^*}^p$. I.e., for fixed $\vec{y}$ and $\vec{t}^*$, $\norm{V\vec{y} - \vec{t}^*}_p^p$ is linear in the values $\alpha_{\vec{u}}$. This leads us to define the $2^k \times 2^k$ matrix $H_{k, p}(t^*)$ whose entry in row $\vec{u}$ and column $\vec{y}$ is equal to $\abs{\iprod{\vec{u}, \vec{y}} - t^*}^p$. Then, for non-negative $\vec{\alpha} = (\alpha_{\vec{u}})_{\vec{u} \in \pmo^k}$, the coordinate of $H_{k, p}(t^*) \cdot \vec{\alpha}$ indexed by $\vec{y}$ is equal to $\norm{V\vec{y} - \vec{t}^*}_p^p$.

In order to show that there exist choices of $\vec{\alpha}$ and $t^*$ such that $V$ and $\vec{t}^*$ form a ``$\pmo$ isolating parallelepiped,'' it therefore suffices to find non-negative $\vec{\alpha}$ such that $H_{k, p}(t^*) \cdot \vec{\alpha} = (1 + \eps, 1, 1, \ldots, 1)^T$ for some $\eps > 0$. We then use the following proof strategy for finding such $\vec{\alpha}$: (1) Show that for certain values of $k$ and $p$, $H_{k, p}(t^*)$ is non-singular so that we can compute $\vec{\alpha} = H_{k, p}(t^*)^{-1} \cdot (1 + \eps, 1, 1, \ldots, 1)^T$, and (2) show that if we pick $\eps > 0$ to be small enough then $\vec{\alpha}$ computed this way will be non-negative. In fact, there is nothing special about the vector $(1 + \eps, 1, 1, \ldots, 1)^T$, and we show a similar result for all vectors in some open neighborhood of $\vec{1}$, which will prove useful in proving Theorem~\ref{thm:intro_hard_approx}.

\subsection{A characterization of isolating parallelepipeds and SETH-hardness}

We now present the main result of this section and show how it implies both a full characterization of the existence of isolating parallelepipeds and the SETH-hardness of $\CVP_p$ for $p \notin 2\Z$.

\begin{theorem}
	For $k \in \Z^+$ and $p \in [1, \infty)$ if $p$ satisfies either (1) $p \notin \Z$ or (2) $p \geq k$, there exists a $(p, k)$-isolating parallelepiped $V \in \R^{2^k \times k}$, $\vec{t}^* \in \R^{2^k}$. 
	Moreover, if $p$ is computable then there is an algorithm that on input $k$ and $p$ outputs such an isolating parallelepiped.
	\label{thm:more-ips}
\end{theorem}

By combining Theorem~\ref{thm:more-ips}, the impossibility results in Corollary~\ref{cor:no-ips-kgeqp}, and the isolating parallelepiped construction in~\cite{conf/focs/BennettGS17} for odd integer $p$, we obtain a complete characterization of the values of $p$ and $k$ for which there exist $(p, k)$-isolating parallelepipeds.

\begin{theorem}
	There exists a $(p, k)$-isolating parallelepiped for $k \in \Z^+$ and $p \in [1, \infty)$ if and only if $p$ satisfies either (1) $p \notin 2\Z$ or (2) $p \geq k$. 
	Moreover, there is an algorithm that on input $k \in \Z^+$ and any computable $p \in [1, \infty)$ with either (1) $p \notin 2\Z$ or (2) $p \geq k$, outputs $V \in \R^{2^k \times k}$ and $\vec{t}^* \in \R^{2^k}$ that define a $(p, k)$-isolating parallelepiped.
	\label{thm:ip-characterization}
\end{theorem}

\begin{proof}
	By Proposition 4.4 and Corollary 4.7 in~\cite{conf/focs/BennettGS17}, such parallelepipeds and the corresponding algorithm exist for odd integers $p$.  Theorem~\ref{thm:more-ips} shows that such parallelepipeds exist for all $p \geq k$ and all $p \notin \Z$, with corresponding algorithms for computable $p$. Corollary~\ref{cor:no-ips-kgeqp} shows that these are the only cases in which isolating parallelepipeds exist.
\end{proof}

The (Karp) reduction from (weighted Max-)$k$-SAT to $\CVP_p$ assuming the existence of computable $(p, k)$-isolating parallelepipeds given in~\cite[Theorem 3.2]{conf/focs/BennettGS17} immediately implies the following. (We actually show a strictly stronger reduction in Section~\ref{sec:gap-seth-hardness}.)

\begin{corollary}
	For every $\eps > 0$ and every computable $p \in [1, \infty) \setminus 2\Z$, there is no $2^{(1-\eps)n}$-time algorithm for $\CVP_p$ assuming W-Max-SAT-SETH. In particular, there is no $2^{(1-\eps)n}$-time algorithm for $\CVP_p$ assuming SETH.
	\label{cor:cvp-hardness}
\end{corollary}

We also note that the ``in particular'' part of the above claim also holds for $p = \infty$ by~\cite[Theorem 6.5]{conf/focs/BennettGS17}, but that the reduction given in~\cite[Theorem 3.2]{conf/focs/BennettGS17} only works when $p$ is finite.

A natural question to ask is whether Corollary~\ref{cor:cvp-hardness} can be extended to $p \in 2\Z$ using a reduction that does not use isolating parallelepipeds. In Section~\ref{sec:limitations}, we give an impossibility result precluding a much larger class of reductions, which we call ``natural reductions.'' 

\subsection{A parameterized family of parallelepipeds}

We first recall the following simple observation from~\cite{conf/focs/BennettGS17}, which says that we can ``work over $\pmo$ instead of $\bit$'' when defining isolating parallelepipeds, which we will do in this section. 

\begin{lemma}
	There is an efficient algorithm that takes as input a matrix $V \in \R^{d^* \times k}$ and vector $\vec{t}^* \in \R^{d^*}$, and outputs a matrix $V' \in \R^{d^* \times k}$ and vector $\vec{t}' \in \R^{d^*}$ such that for all $\vec{z}\in \{0,1\}^n$ and all $p \geq 1$, $\|V' \vec{z} - \vec{t}' \|_p = \|V(\vec{1}_k - 2\vec{z}) - \vec{t}^*\|_p$, where $\vec{1}_k := (1,1,\ldots, 1)$.
\label{lem:ips-pmo-bit}
\end{lemma}

\begin{figure}
\begin{align*}
&V :=
\begin{pmatrix*}[c]
\alpha_{(-1, -1, -1)}^{1/p} & 0 & \cdots & 0 \\
0 & \alpha_{(-1, -1, 1)}^{1/p} & \cdots & 0 \\
\vdots & \vdots & \ddots & \vdots \\
0 & 0 & \cdots & \alpha_{(1, 1, 1)}^{1/p}
\end{pmatrix*}
\cdot
\begin{pmatrix*}[r]
-1 & -1 & -1 \\
-1 & -1 & 1 \\
-1 & 1 & -1 \\
1 & -1 & -1 \\
-1 & 1 & 1 \\
1 & -1 & 1 \\
1 & 1 & -1 \\
1 & 1 & 1 \\
\end{pmatrix*} \\ 
&\vec{t}^* := 
\begin{pmatrix*}[c]
\alpha_{(-1, -1, -1)}^{1/p} & 0 & \cdots & 0 \\
0 & \alpha_{(-1, -1, 1)}^{1/p} & \cdots & 0 \\
\vdots & \vdots & \ddots & \vdots \\
0 & 0 & \cdots & \alpha_{(1, 1, 1)}^{1/p}
\end{pmatrix*}
\cdot
\begin{pmatrix*}[r]
t^* \\
t^* \\
t^* \\
t^* \\
t^* \\
t^* \\
t^* \\
t^* \\
\end{pmatrix*}
\end{align*}
\caption{$V$ and $\vec{t}^*$ of the form defined in Definition~\ref{def:param-ip} for $k = 3$ and $p \geq 1$. Lemma~\ref{lem:ips-pmo-bit} and Proposition~\ref{prop:finding-alpha} together assert that, for $p \in [1, \infty)$ where $p$ satisfies either (1) $p \notin \Z$ or (2) $p \geq 3$, there exist $(\alpha_{\vec{u}})_{\vec{u} \in \pmo^3}$ and $t^*$ such that $V' := 2V$, $(\vec{t}^*)' := V\vec{1} + \vec{t}^*$ form an isolating parallelepiped.}
\label{fig:param-ip-example}
\end{figure}

We next define a family of $k$-parallelepipeds $V \in \R^{2^k \times k}$,  $\vec{t}^* \in \R^{2^k}$ parameterized by (1) $2^k$ non-negative numbers $(\alpha_{\vec{u}})_{\vec{u}\in \pmo^k}$, where, for some $p \geq 1$, $\alpha_{\vec{u}}^{1/p}$ scales the row of $V$ and coordinate of $\vec{t}^*$ corresponding to $\vec{u} \in \pmo^k$, and (2) another number $t^* \in \R$.

\begin{definition}
For $p \in [1, \infty)$, $k \in \Z^+$, $\vec{\alpha} \in (\R^{\geq 0})^{2^k}$, and $t^* \in \R$, define the matrix $V = V(\vec{\alpha}) \in \R^{2^k \times k}$ and vector $\vec{t}^* = \vec{t}^*(\vec{\alpha}, t^*) \in \R^{2^k}$ as follows.
Set the row of $V$ indexed by $\vec{u}$ to be $\alpha_{\vec{u}}^{1/p} \cdot \vec{u}^T$, and set $\vec{t}^* := t^* \cdot (\alpha_{\vec{u}}^{1/p})_{\vec{u} \in \pmo^k}$.
\label{def:param-ip}
\end{definition}

I.e., $V$ is the matrix whose rows consist of vectors $\vec{u} \in \pmo^k$ scaled by corresponding weights $\alpha_{\vec{u}}^{1/p}$, and the coordinate of $\vec{t}^*$ indexed by $\vec{u}$ is equal to $\alpha_{\vec{u}}^{1/p} \cdot t^*$. (See Figure~\ref{fig:param-ip-example}.)
We also define another matrix, $H$, which we will use to relate our choice of parameters $\vec{\alpha}$ and $t^*$ to the value of $\norm{V\vec{y} - \vec{t}^*}_p^p$ for $\vec{y} \in \pmo^k$.

\begin{definition}
For $p \geq 1$ and an integer $k \geq 0$, define the matrix $H_{k, p}(t^*) \in \R^{2^k \times 2^k}$ by
$(H_{k, p}(t^*))_{\vec{u}, \vec{v}} := \abs{\iprod{\vec{u}, \vec{v}} - t^*}^p$ for $k \geq 1$, and define $H_{0, p}(t^*) := \abs{t^*}^p$.
\label{def:H-matrix}
\end{definition}

We next show that for $\vec{y} \in \pmo^k$, $\norm{V\vec{y} - \vec{t}^*}_p^p$ is equal to the inner product of $\vec{\alpha}$ with row $\vec{y}$ of $H_{k,p}(t^*)$.

\begin{lemma}
For $\vec{\alpha} \in (\R^{\geq 0})^{2^k}$ and $t^* \in \R$, let $V = V(\vec{\alpha})$ and let $\vec{t}^* = \vec{t}^*(\vec{\alpha}, t^*)$ be as defined in Definition~\ref{def:param-ip}. Then
\[
(H_{k, p}(t^*) \cdot \vec{\alpha})_{\vec{y}} = \norm{V\vec{y} - \vec{t}^*}_p^p \ .
\]
\label{lem:alpha-to-weights}
\end{lemma}

\begin{proof}
For $\vec{y} \in \pmo^k$,
\begin{align*}
(H_{k, p}(t^*) \cdot \vec{\alpha})_{\vec{y}} &= \sum_{\vec{u} \in \pmo^k} H_{k, p}(t^*)_{\vec{y}, \vec{u}} \cdot \alpha_{\vec{u}} \\
		                             &= \sum_{\vec{u} \in \pmo^k} \abs{\iprod{\alpha_{\vec{u}}^{1/p} \cdot \vec{u}, \vec{y}} - \alpha_{\vec{u}}^{1/p} \cdot t^*}^p \\
											           &= \norm{V\vec{y} - \vec{t}^*}_p^p \ ,
\end{align*}
as needed.
\end{proof}

We will show that for every $k \in \Z^+$ and every $p \in [1, \infty)$ that satisfies either (1) $p \notin \Z$ or (2) $p \geq k$, there exists $t^* \in \R$ such that $H_{k, p}(t^*)$ is non-singular.
To show this, we will start by analyzing the eigenvectors and eigenvalues of $H_{p, k}(t^*)$. 

\subsection{Eigenvectors and eigenvalues of \texorpdfstring{$H_{k,p}(t^*)$}{H}}
\label{subsec:eig-vec-val}

We start by showing that the eigenvectors of $H_{p, k}(t^*)$ have a very nice form, and importantly that they do not depend on either $p$ or $t^*$. Namely, the vectors $\vec{v}_S \in \pmo^{2^k}$ corresponding to the output table of the Fourier characters $\chi_S(\vec{x}) := \prod_{i \in S} x_i$ (as in Eq.~\eqref{eq:vs}) are eigenvectors of $H_{p, k}(t^*)$. Accordingly, the eigenvalues $\lambda_S$ corresponding to $\vec{v}_S$ are the Fourier coefficients $\fh(S)$ where $f(\vec{u}) := |\sum u_i - t^*|$ for $\vec{u}\in \{-1,1\}^k$.

\begin{lemma}
For all $1 \leq p < \infty$ and $t^* \in \R$, the $2^k$ vectors $\vec{v}_S \in \pmo^{2^k}$ for $S \subseteq [k]$ of the form in Eq.~\eqref{eq:vs} are eigenvectors of $H_{p, k}(t^*)$. In particular, there are $2^k$ such vectors, and they form an eigenbasis of $H_{p, k}(t^*)$.
\label{cor:H-eigenvectors}
\end{lemma}

\begin{proof}
We prove the lemma by induction on $k$. In the base case where $k = 0$, the scalar $1$ is an ``eigenvector'' of $H_{k, p}(t^*)$.
We next consider the inductive case where $k \geq 1$. Let $\vec{u} = (u_1, \vec{u}'), \vec{y} = (y_1, \vec{y}') \in \pmo^k$. If $u_1 = v_1$ we then have that $\iprod{\vec{u}, \vec{y}} = \iprod{\vec{u}', \vec{y}'} + 1$, and if $u_1 \neq y_1$ then $\iprod{\vec{u}, \vec{y}} = \iprod{\vec{u}', \vec{y}'} - 1$ (with $\iprod{\vec{u}', \vec{y}'} = 0$ if $k-1 = 0$).
Therefore, we can write $H_{k, p}(t^*)$ in block form as
\[
H_{k, p}(t^*) = \begin{pmatrix} H_{k-1, p}(t^* - 1) & H_{k-1, p}(t^* + 1) \\ H_{k-1, p}(t^* + 1) & H_{k-1, p}(t^* - 1) \end{pmatrix} \ .
\]
By the induction hypothesis, the eigenvectors of $H_{k-1, p}(t^* - 1)$ and $H_{k-1, p}(t^* + 1)$ are the same. Suppose that $\vec{v}$ is such an eigenvector. Then one can check that $(\vec{v}, \vec{v})$ and $(-\vec{v}, \vec{v})$ are eigenvectors of $H_{k, p}(t^*)$.

Furthermore, by the induction hypothesis, each vector $\vec{v}_S$ for $S \subseteq [k - 1]$ is an eigenvector of $H_{k-1, p}(t^* - 1)$ and $H_{k-1, p}(t^* + 1)$. Therefore, $\vec{v}_{S'} := (\vec{v}_S, \vec{v}_S)$ and $\vec{v}_{S''} := (-\vec{v}_S, \vec{v}_S)$ are (distinct) eigenvectors of $H_{k, p}(t^*)$. By Eq.~\eqref{eq:fourier-recurrence}, we see that for all $T \subseteq [k]$, $\vec{v}_T$ has this form, and is thus an eigenvector of $H_{k, p}(t^*)$.
\end{proof}

\begin{corollary}
Each eigenvalue $\lambda_S$ corresponding to the eigenvector $\vec{v}_S$ of $H_{k, p}(t^*)$ has the value
\begin{equation}
\lambda_S = \sum_{\vec{x} \in \pmo^k} \chi_S(\vec{x}) \cdot \Big|\sum_{i=1}^k x_i - t^*\Big|^p \ .
\label{eq:H-eigenvalues}
\end{equation}
In particular, each $\lambda_S$ satisfies
\begin{equation}
\lambda_S = \sum_{j=0}^k a_j \cdot \abs{k - 2j - t^*}^p
\label{eq:H-eigenvalues-form}
\end{equation}
for some $a_0, \ldots, a_k \in \Z$ with $a_0 = 1$.
\label{cor:H-eigenvalues}
\end{corollary}

\begin{proof}
Fix an eigenvector $\vec{v}_S$ of $H_{k, p}(t^*)$. It holds that for $\vec{u} \in \pmo^k$,
\[
(H_{k,p}(t^*) \cdot \vec{v}_S)_{\vec{u}} := \sum_{\vec{x} \in \pmo^k} \chi_S(\vec{x}) \cdot \abs{\iprod{\vec{u}, \vec{x}} - t^*}^p \ .
\]
Moreover, setting $\vec{u} = \vec{1}$ in the above expression and noting that $(\vec{v}_S)_{\vec{1}} = 1$ for all $S \subseteq [k]$, we get that $\lambda_S = \sum_{\vec{x} \in \pmo^k} \chi_S(\vec{x}) \cdot \abs{\sum_{i=1}^k x_i - t^*}^p$, as claimed.

The ``in particular'' part of the claim follows by noting that each term in the sum in Eq.~\eqref{eq:H-eigenvalues} is equal to $\pm \abs{k - 2j - t^*}^p$ for some $j \in \set{0, 1, \ldots, k}$. The fact that $a_0 = 1$ in Eq.~\eqref{eq:H-eigenvalues-form} for every $S$ follows by noting that the term corresponding to $\vec{x} = 1$ is the unique term equal to $\pm \abs{k - t^*}^p$ and that because $\chi_S(\vec{1}) = 1$ it is equal to $\abs{k - t^*}^p$.
\end{proof}

Fix $k \geq 0$. Using Corollary~\ref{cor:H-eigenvalues}, we can compute relatively simple expressions for the eigenvalues $\lambda = \lambda_{\emptyset}$ of $\vec{1} = \vec{v}_{\emptyset}$ and $\lampar = \lambda_{[k]}$ of $\vpar = \vec{v}_{[k]}$ by noting that for $S = \emptyset$ and $S = [k]$ the value of each term $\chi_S(\vec{x}) \cdot |\sum_{i = 1}^k x_i - t^*|^p$ in Eq.~\eqref{eq:H-eigenvalues} only depends on the number of coordinates $j$ of $\vec{x}$ equal to $-1$. Namely,
\begin{equation}
\lambda = \sum_{j = 0}^k \binom{k}{j} \abs{k - 2j - t^*}^p \ ,
\label{eq:lambda}
\end{equation}
and
\begin{equation}
\lampar = \sum_{j = 0}^k (-1)^j \binom{k}{j} \abs{k - 2j - t^*}^p \ .
\label{eq:lambda-par}
\end{equation}
We note that $\lambda > 0$ for $k \geq 1$ regardless of $p$ and $t^*$. Indeed, this follows by observing that each term in Eq.~\eqref{eq:lambda} sum is non-negative, and at most one term is equal to zero.

\subsection{Non-singularity of \texorpdfstring{$H$}{H} with certain parameters}

We next show that the function $t^* \mapsto \det(H_{k,p}(t^*))$ is analytic and not identically zero for certain $k$ and $p$. Using the general fact that such functions have isolated roots, this leads to a simple algorithm for finding $t^*$ such that $\det(H_{k,p}(t^*))$ is non-singular for such $k$ and $p$.

\begin{proposition}
Let $k \in \Z^+$, and let $p \in [1, \infty)$ be a value that satisfies either (1) $p \notin \Z$ or (2) $p \geq k$. Then $\det(H_{k,p}(t^*))$ is analytic and not identically zero as a function of $t^*$ for $t^* > k$.
\label{prop:wronskian}
\end{proposition}

\begin{proof}
Fix $k \in \Z^+$ and $p$ satisfying either (1) $p \notin \Z$ or (2) $p \geq k$.
We have that $\det(H_{k,p}(t^*)) = \prod_{S \subseteq [k]} \lambda_S$, and by Corollary~\ref{cor:H-eigenvalues} that each eigenvalue $\lambda_S$ is a linear combination \begin{equation}
\lambda_S = \lambda_S(t^*) = \sum_{j=0}^k a_j \cdot \abs{k - 2j - t^*}^p = \sum_{j=0}^k a_j \cdot \abs{t^* - k + 2j}^p
\label{eq:eigenval-form-wronskian}
\end{equation}
of functions $|t^* - k + 2j|^p$ for $j \in \set{0, 1, \ldots, k}$ with $a_0, \ldots, a_k \in \Z$ and $a_0 = 1$.
Moreover, functions of the form $\abs{t^* - k + 2j}^p$ for $j \in \set{0, 1, \ldots, k}$ satisfy $\abs{t^* - k + 2j}^p = (t^* - k + 2j)^p$ and are analytic for $t^* > k$.
So, $\det(H_{k, p}(t^*))$ is also analytic for $t^* > k$, and in order to show that $\det(H_{k, p}(t^*))$ is not identically zero it suffices to show that each eigenvalue $\lambda_S = \lambda_S(t^*)$ is not identically zero as a function of $t^*$. (Here, we are using the fact that the product $\prod_{i=1}^{m} f_i(t)$ of finitely many analytic functions $f_i : U \to \R$ over an open set $U \subset \R$ is identically zero if and only if one of the $f_i(t)$ is identically zero. This follows, e.g., from the fact that analytic functions have at most countably many roots.)

Because $a_0 = 1 \neq 0$ for each $\lambda_S = \lambda_S(t^*) = \sum_{j = 0}^k a_j \cdot |t^* - k + 2j|^p$, to show that $\lambda_S(t^*)$ is not identically zero it suffices to show that the functions $|t^* - k + 2j|^p$ for $j \in \set{0, 1, \ldots, k}$ are linearly independent over the reals. Moreover, it suffices to show that these functions are linearly independent for $t^* > k$, and therefore to show that the functions $(t^* - k + 2j)^p$ for each $j \in \set{0, 1, \ldots, k}$ are linearly independent, since $|t^* - k + 2j|^p = (t^* - k + 2j)^p$ for $t^* > k$.

By Fact~\ref{fct:wronskian-lin-ind}, to show that these functions are linearly independent it suffices to show that their Wronskian $W := \det(M)$ with
\[
M = M_{p, k}(t^*) :=
\begin{pmatrix}
(p)_0 \cdot (t^*-k)^{p\phantom{-0}} & (p)_0 \cdot (t^*-k+2)^{p\phantom{-0}} & \cdots & (p)_0 \cdot (t^*+k)^{p\phantom{-0}} \\
(p)_1 \cdot (t^*-k)^{p-1} & (p)_1 \cdot (t^*-k+2)^{p-1} & \cdots & (p)_1 \cdot (t^*+k)^{p-1} \\
 \vdots & \vdots & \ddots & \vdots \\
(p)_k \cdot (t^*-k)^{p-k} & (p)_k \cdot (t^*-k+2)^{p-k} & \cdots & (p)_k \cdot (t^*+k)^{p-k}
\end{pmatrix}
\]
is not identically zero for $t^* > k$.
Here the notation $(p)_i$ denotes the falling factorial function, which is defined by $(p)_i := p (p-1) \cdots (p-(i-1))$ for $i \geq 1$ and $(p)_0 := 1$.

In fact, we show that 
\begin{equation}
\label{eq:Wronskian_value}
	W = - 2^{k(k+1)/2} \cdot \Big(\prod_{0 \leq i < j \leq k} (j - i) \Big) \cdot \Big(\prod_{i=0}^k (p)_i \Big) \cdot \Big( \prod_{j=0}^k (t^* - k + 2j)^{p-k} \Big)
\end{equation}
is non-zero for all $t^* > k$ and $p$ satisfying the conditions of the theorem. This immediately implies the result.

To that end, dividing the $i$th row of $M$ (which has rows indexed by $i \in \set{0, 1, \ldots, k}$) by $(p)_i$ (which is non-zero because of our assumptions about $p$) we obtain
\[
M' = M_{p, k}'(t^*) := \begin{pmatrix}
(t^*-k)^{p\phantom{-0}} & (t^*-k+2)^{p\phantom{-0}} & \cdots & (t^*+k)^{p\phantom{-0}} \\
(t^*-k)^{p-1} & (t^*-k+2)^{p-1} & \cdots & (t^*+k)^{p-1} \\
\vdots & \vdots & \ddots & \vdots \\
(t^*-k)^{p-k} & (t^*-k+2)^{p-k} & \cdots & (t^*+k)^{p-k}
\end{pmatrix} \ .
\]
Similarly, dividing the $j$th column of $M'$ (which has columns indexed by $j \in \set{0, 1, \ldots, k}$) by $(t^* - k + 2j)^{p-k}$ (which is well-defined and non-zero for $t^* > k$) we obtain
\[
M'' = M_{p, k}''(t^*) := \begin{pmatrix}
(t^*-k)^{k\phantom{-0}} & (t^*-k+2)^{k\phantom{-0}} & \cdots & (t^*+k)^{k\phantom{-0}} \\
(t^*-k)^{k-1} & (t^*-k+2)^{k-1} & \cdots & (t^*+k)^{k-1} \\
\vdots & \vdots & \ddots & \vdots \\
1 & 1 & \cdots & 1
\end{pmatrix} \ ,
\]
which is a Vandermonde matrix up to transposition and reordering of the rows. We can therefore use the formula for the determinant of a Vandermonde matrix to compute
\[
\det(M'') = -\prod_{0 \leq i < j \leq k} \big((t^* - k + 2j) - (t^* - k + 2i)\big) = - \prod_{0 \leq i < j \leq k} 2(j - i) \neq 0
\]
The result follows by noting that 
\[
	W = \det(M'') \cdot \Big(\prod_{i=0}^k (p)_i \Big) \cdot \Big( \prod_{j=0}^k (t^* - k + 2j)^{p-k} \Big)
	\; ,
\]
as claimed in Eq.~\eqref{eq:Wronskian_value}.
\end{proof}

\begin{corollary}
For every $k \in \Z^+$ and every real $p \in [1, \infty)$ that satisfies either (1) $p \notin \Z$ or (2) $p \geq k$, there exists $t^*$ such that $\det(H_{k, p}(t^*)) \neq 0$. Moreover, if $p$ is computable then there is an algorithm that on input $k$ and $p$ outputs such a $t^*$.
\label{cor:finding-x}
\end{corollary}

\begin{proof}
The corollary is an immediate consequence of Proposition~\ref{prop:wronskian} and the fact that an analytic function that is not identically zero has isolated roots. Indeed, the fact that such a function has isolated roots implies that the following algorithm must halt (when $p$ is computable). Compute $\det(H_{k, p}(t_i^*))$ where $t_i^* = k + 2^{-i}$ for $i = 1, 2, \ldots$, and output the first $t_i^*$ for which ${\det(H_{k, p}(t_i^*)) \neq 0}$.
\end{proof}

\subsection{Finishing the proof}

If $H_{k, p}(t^*)$ is non-singular then for any vector $\vec{w} \in \R^{2^k}$ we can solve the linear system
$H_{k, p}(t^*) \cdot \vec{\alpha} = \vec{w}$
to obtain some solution $\vec{\alpha}$.
In particular, if $\vec{w} = (1 + \eps, 1, 1, \ldots, 1)$ for some $\eps > 0$ and the solution $\vec{\alpha}$ to this equation is non-negative, then by Lemma~\ref{lem:alpha-to-weights} we can use $\vec{\alpha}$ as the weights in an isolating parallelepiped of the form in Definition~\ref{def:param-ip}.
The issue with this is that we critically require that our solution $\vec{\alpha}$ be non-negative, and a priori there is no guarantee that it will be.
However, we next show that by setting $\eps > 0$ appropriately we can ensure that the solution $\vec{\alpha}$ will in fact be non-negative. (In fact, we note that such solutions exist for $\vec{w}$ in an open neighborhood around $\vec{1}$.  Note that the theorem is only interesting if $\vec{\alpha}' \notin (\R^{\geq 0})^{2^k}$ has at least one negative coordinate. Otherwise, we can clearly take $\eps > 0$ to be as large as we like.)

\begin{proposition}
Fix $k \in \Z^+$, $p \in [1, \infty)$, and $t^* \in \R$ such that $H_{k,p}(t^*)$ is non-singular.
Let $\vec{b} \in \R^{2^k}$ and let $\vec{\alpha}' := H_{k,p}(t^*)^{-1} \cdot \vec{b}$.
Then there exists $\vec{\alpha} \in (\R^{\geq 0})^{2^k}$ that satisfies
$H_{k,p}(t^*) \cdot \vec{\alpha} = \vec{1} + \eps \cdot \vec{b}$ for
\begin{equation}
\eps := \frac{1}{\lambda \cdot \abs{\min_{\vec{u}\in \{-1,1\}^k} \alpha_{\vec{u}}'}} \geq \frac{1}{\lambda \cdot \norm{\vec{\alpha}'}_{\infty}} > 0 \ ,
\label{eq:parallelepiped-eps}
\end{equation}
where $\lambda = \sum_{j = 0}^k \binom{k}{j} \cdot \abs{t^* - k + 2j}^p > 0$ is the eigenvalue of $H_{k,p}(t^*)$ corresponding to the eigenvector $\vec{1}$ as in Eq.~\eqref{eq:lambda}.
\label{prop:finding-alpha}
\end{proposition}

\begin{proof}
Let
\[
\vec{\alpha} := \frac{1}{\lambda} \cdot \vec{1} + \eps \cdot \vec{\alpha}' \ .
\]
Then $\vec{\alpha}$ is non-negative and $H_{k,p}(t^*) \cdot \vec{\alpha} = \vec{1} + \eps \cdot \vec{b}$, as needed.
\end{proof}

Theorem~\ref{thm:more-ips} then follows immediately by combining Lemma~\ref{lem:ips-pmo-bit}, Lemma~\ref{lem:alpha-to-weights}, Corollary~\ref{cor:finding-x}, and Proposition~\ref{prop:finding-alpha} (applied with $\vec{b} := \vec{e}_1$).

\section{Gap-SETH hardness of CVP}
\label{sec:gap-seth-hardness}

In this section, we prove fine-grained hardness of approximation of $\CVP_p$ for all $p \in [1, \infty) \setminus 2\Z$. To that end, we first show in Section~\ref{subsec:iso-lats} how to modify isolating parallelepipeds to what we call \emph{isolating lattices}, which are entire lattices with the property that the closest vectors to some target correspond exactly to the satisfying assignments of some CSP, and all other lattice vectors are at least a $(1 + \eps)$ factor farther away from the target. (We also need the unsatisfying assignments to be ``second-closest'' vectors, and all exactly a $1 + \eps$ factor farther away.)  We show in Theorem~\ref{thm:cvp-hardness-cvips} that such gadgets imply a reduction from constant-factor approximate Gap-$k$-CSPs to constant-factor approximate $\CVP_p$, where the approximation factor depends on $\eps$.

In fact, we consider general $k$-CSPs, and not just $k$-SAT.
This is motivated for two reasons: (1) the fact that general $k$-CSPs are known to be $\mathsf{NP}$-hard to approximate to within much better approximation factors than $k$-SAT~\cite{journals/jacm/Chan16,journals/cc/AustrinM09,conf/dagstuhl/MakarychevM17} (since $(s,c)$-Gap-$k$-SAT is trivial for $s \leq 1-2^{-k}$), and so it is natural to hypothesize some corresponding quantitative hardness of approximation for them; and (2) we are able to analyze our parallelepiped construction better for different CSPs, and therefore to get an explicit lower bound on $\eps$. In particular, $(s,c)$-Gap-$k$-Parity is known to be NP-hard to approximate for any constants $s > 1/2$ and $c < 1$, and and at least as hard (in a fine-grained sense) to approximate as $k$-SAT, as in Theorem~\ref{thm:SV19}. Furthermore, we show how to build isolating parallelepipeds (and thus isolating lattices) for $k$-\emph{Parity} that have a relatively large gap $\eps \approx 1/k^{(p+3)/2}$ between the distances for satisfying and unsatisfying assignments.

In fact, $k$-Parity arises particularly naturally in this context because the parity function corresponds to the eigenvector $\vec{v}_{[k]}$ of $H_{p,k}(t^*)$ as described in Lemma~\ref{cor:H-eigenvectors}, which makes it much more amenable to the techniques in Section~\ref{sec:other-p}. Indeed, Lemma~\ref{cor:H-eigenvectors} shows an eigenbasis corresponding exactly to the parity functions applied to subsets of their input variables (i.e., the Fourier basis). So, $k$-Parity is the only non-degenerate $k$-CSP (i.e., the only $k$-CSP with constraints that depend on all $k$ of their input variables) corresponding to a vector in this eigenbasis.

 Together, these two properties allow us not only to show reductions from $(s,c)$-Gap-$k$-Parity to approximate $\CVP_p$ with a relatively large approximation factor but even to show reductions from $(s,c)$-Gap-$k$-SAT to $\CVP_p$ with a larger approximation factor than we know how to achieve directly.

This leads to the following fine-grained hardness of approximation results for $\CVP_p$.

\begin{theorem}
		\label{thm:parity_reduction}
		For all $p \in [1,\infty) \setminus 2\Z$, all integers $k > \max\{2,p\}$, and all $1/2 <s \leq c < 1$, there exists a polynomial time (Karp) reduction from $(s, c)$-Gap-$k$-Parity instances on $n$ variables to $\gamma$-$\CVP_p$ instances of rank $n$ with $\gamma = \gamma(p, k, s,c)$ satisfying 
		\[
		\gamma \geq 1+ (c-s)\cdot \frac{|\sin(\pi p/2)|}{4p^3k} \cdot \Big(\frac{2p}{ e^2 \pi^2 k}\Big)^{(p+1)/2} 
		\; .
		\]
\end{theorem}

Theorem~\ref{thm:parity_reduction} implies $2^{(1-\eps)n}$-hardness of approximation of $\CVP_p$ with an explicit constant approximation factor $\gamma(p, \eps)$ under a sufficiently strong complexity-theoretic assumption. 
As mentioned in the introduction, the fastest known algorithms for $(s,c)$-Gap-$k$-Parity run in time roughly $(2 - \sqrt{c/s - 1})^n$~\cite{conf/soda/AlmanCW20}. So, it is consistent with current knowledge to hypothesize that the fastest possible algorithms for $(s,c)$-Gap-$k$-Parity require time $(2 - \poly(c/s - 1))^n$ for $c/s < 2$.
Assuming this hypothesis and taking, e.g., $c - s = 1/\poly(k)$ (in which case $c/s \leq 1 + 1/\poly(k)$) Theorem~\ref{thm:parity_reduction} then asserts that $(1 + 1/\poly(k))$-approximate $\CVP_p$ requires $(2 - 1/\poly(k))^n$ time for fixed $p \notin 2\Z$

We also get fine-grained hardness of approximation for $\CVP_p$ based on Gap-$k$-SAT as an immediate corollary of Theorem~\ref{thm:parity_reduction} combined with the reduction from Gap-$k$-SAT to Gap-$k$-Parity in Theorem~\ref{thm:SV19}.

\begin{theorem}
	\label{thm:stronger_gap_SETH_or_something}
	For all $p \in [1,\infty) \setminus 2\Z$, all integers $k > \max\{2,p\}$, and all $1-2^{-k} <s \leq c \leq 1$, there exists a polynomial time (Karp) reduction from $(s, c)$-Gap-$k$-SAT instances on $n$ variables to $\gamma$-$\CVP_p$ instances of rank $n$ with $\gamma = \gamma(p, k, s,c)$ satisfying 
	\[
	\gamma \geq 1+ \frac{2^{k-1}}{2^k - 1} \cdot (c-s)\cdot \frac{|\sin(\pi p/2)|}{4p^3 k} \cdot \Big(\frac{2p}{ e^2 \pi^2 k}\Big)^{(p+1)/2} 
	\; .
	\]
In particular, for all $p \in [1, \infty) \setminus 2\Z$ and every $\eps > 0$ there exists $\gamma = \gamma(p, \eps, s, c) > 1$ such that there is no $2^{(1 - \eps)n}$-time algorithm for $\gamma$-$\CVP_p$ assuming Gap-SETH.
\end{theorem}

We note in passing that the Gap-SETH result (with a non-explicit approximation factor $\gamma(p, \eps, s, c) > 1$) in the above theorem can also be shown directly (i.e., without going through parity) from the results of Section~\ref{sec:other-p}, Proposition~\ref{prop:ip_to_lattice}, and Theorem~\ref{thm:cvp-hardness-cvips}. (In particular, the gap $\eps(k,p) > 0$ that we obtain in Section~\ref{sec:other-p} is necessarily constant for constant $k$ and $p$.)

\subsection{Isolating lattices}
\label{subsec:iso-lats}

The following definition strengthens the notion of an ``isolating parallelepiped'' to an ``isolating lattice.'' It also generalizes to arbitrary CSPs, rather than just $k$-SAT, and explicitly considers the ``gap'' $\eps$ between satisfying and unsatisfying assignments.
(We allow for the possibility that $\eps = 0$ in order to capture the case when $p \in 2\Z$ more naturally.)

\begin{definition}
For any $1 \leq p \leq \infty$, integer $k \geq 1$, constraint $C : \bit^k \to \bit$, and $\eps \geq 0$, we say that $V \in \R^{d^* \times k}$ with full column rank and $\vec{t}^* \in \R^{d^*}$ define a \emph{$(p, k, C, \eps)$-isolating parallelepiped} (respectively, $(p, k, C, \eps)$-isolating lattice) if conditions~\ref{en:sat-vecs} and~\ref{en:unsat-vecs} (resp., if conditions~\ref{en:sat-vecs},~\ref{en:unsat-vecs}, and~\ref{en:other-vecs}) below hold:
\begin{enumerate}
\item (Satisfying assignments are close.) For all $\vec{x} \in C^{-1}(1)$, $\norm{V\vec{x} - \vec{t}^*}_p = 1$. \label{en:sat-vecs}
\item (Unsatisfying assignments are far.) For all $\vec{x} \in C^{-1}(0)$, $\norm{V\vec{x} - \vec{t}^*}_p = 1 + \eps$. \label{en:unsat-vecs}
\item (Non-boolean assignments are far.) For all $\vec{x} \in \Z^k \setminus \bit^k$, $\norm{V\vec{x} - \vec{t}^*}_p \geq 1 + \eps$. \label{en:other-vecs}
\end{enumerate}
\label{def:cvip}
\end{definition}

The following proposition shows how to construct a $(p, k, C, \eps')$-isolating lattice from any $(p, k, C, \eps)$-isolating parallelepiped $V, \vec{t}^*$.
The idea is simply to append a scaled identity matrix to the bottom of $V$ and a vector whose entries are all the same to the bottom of $\vec{t}^*$.
We note that, up to the values of $\eps, \eps'$, the converse to the proposition is trivial since any isolating lattice is also an isolating parallelepiped.
\begin{proposition}
	\label{prop:ip_to_lattice}
For any $1 \leq p \leq \infty$, integer $k \geq 1$, constraint $C : \bit^k \to \bit$, and $\eps' > 0$, if there exists a (computable) $(p, k, C, \eps)$-isolating parallelepiped then there exists a (computable) $(p, k, C, \eps')$-isolating lattice, where
\[
	\eps' = \Big(\frac{(1 + \eps)^p + k \mu}{1 + k \mu}\Big)^{1/p} - 1 \geq \frac{\eps}{1 + k \mu}
\] 
and $\mu := (1+\eps)^p/(3^p - 1)$.
\label{prop:ip-to-cvip}
\end{proposition}

\begin{proof}
Suppose that $V$ and $\vec{t}^*$ define a $(p, k, C, \eps)$-isolating parallelepiped. Define
\[
V' := \frac{1}{(1 + k \mu)^{1/p}} \cdot \begin{pmatrix} V \\ 2\mu^{1/p} \cdot I_k  \end{pmatrix} \ , \qquad \vec{t}' := \frac{1}{(1 + k \mu)^{1/p}} \cdot \begin{pmatrix} \vec{t}^* \\ \mu^{1/p} \cdot \vec{1} \end{pmatrix} \ .
\]
One can check that $V'$ and $\vec{t}'$ define a $(p, k, C, \eps')$-isolating lattice with $\eps'$ as specified above.

Finally, to see that $\eps' \geq \eps/(1+k\mu)$, we define
\[
	f(s,x) := \Big(\frac{x^s + \beta}{1+\beta}\Big)^{1/s}
	\; .
\]
It suffices to show that $f(s,x)$ is non-decreasing in $s$ for $s \geq 1$, $x \geq 1$, and $\beta > 0$.
Notice that
\[
	g(s,x) := \frac{\partial}{\partial s} \log f(s,x) = \frac{1}{s^2} \cdot \Big( \frac{s x^s \log(x)}{x^s + \beta} - (\log(x^s + k\mu) - \log(1+k\mu)) \Big)
\]
In particular, this is zero when $x = 1$, so it suffices to show that this expression is increasing in $x$ for $x \geq 1$. Indeed, a simple computation shows that
\[
	\frac{\partial}{\partial x} g(s,x) = \beta \cdot \frac{x^{s-1} \log x}{(x^s + \beta)^2} \geq 0
	\; ,
\]
as needed. Therefore, $f(s,x)$ is increasing in $s$, so that $f(s,x) \geq f(1,x)$, and the result follows by plugging in $s = p$, $x = 1+\eps$, and $\beta = \mu k$.
\end{proof}

\subsection{Constructing isolating parallelepipeds for parity with large \texorpdfstring{$\eps$}{gap}}
\label{subsec:parity-parallelepipeds}

We now show a construction of $(p, k, C_b, \eps)$-isolating parallelepipeds for the parity constraints defined by 
$C_1(x_1, \ldots, x_k) :=  x_1 \oplus \cdots \oplus x_k$ and $C_0(x_1,\ldots, x_k) = \neg C_1(x_1, \ldots, x_k)$, i.e., $C_b$ constrains the parity of the number of non-zero inputs. (It is trivial to convert an isolating parallelepiped for $C_b$ into one for $C_{1-b}$, but our construction happens to naturally yield both.) 

The proof relies on bounds on sums of binomial coefficients corresponding to eigenvalues of $H_{p,k}(t^*)$. We defer the proof of these bounds to Section~\ref{sec:crazy_binomial}.
 
\begin{theorem}
	\label{thm:ips_parity}
	For any $k \geq 3$, $1 \leq p < k$, and $b \in \{0,1\}$, there exists a computable $(p, k, C_b, \eps)$-isolating parallelepiped for some 
	\[
		\eps \geq \frac{|\sin(\pi p/2)|}{p^2} \cdot \Big(\frac{2p}{ e^2 \pi^2 k}\Big)^{(p+1)/2}
		\; .
	\]
	
	Up to scaling, this is achieved (in $\{-1,1\}$ coordinates) by the construction given in Definition~\ref{def:param-ip} with $\alpha_{\vec{u}} := 1 + (-1)^{\eta + b} \prod_i u_i$ for $\vec{u}\in \{-1,1\}^k$ and $t^* :=  (1+(-1)^{k+1})/2$, where $\eta :=  \floor{k/2}  + \floor{p/2}$.
\end{theorem}
\begin{proof}
 	Let $\vec{1} \in \{-1,1\}^{2^k}$ be the all-ones vector, and let $\vpar \in \{-1,1\}^{2^k}$ be the vector whose $\vec{u}$ coordinate is $\prod_i u_i$ for $\vec{u} \in \{-1,1\}^k$, as in Section~\ref{subsec:eig-vec-val}. In particular, the vector $\vec{\alpha} \in \{-1,1\}^{2^k}$ whose $\vec{u}$ coordinate is $\alpha_{\vec{u}}$ satisfies $\vec{\alpha} = \vec{1}+ (-1)^{\eta + b} \vpar$. 
 	
 	By Lemma~\ref{cor:H-eigenvectors} and Corollary~\ref{cor:H-eigenvalues}, $\vec{1}$ and $\vpar$ are eigenvectors of $H_{p,k}(t^*)$ with respective eigenvalues
	 \[
 		\lambda = 2^p \sum_{j=0}^k \binom{k}{j}|k/2 - j-t^*/2|^p = 2^p \sum_{j=0}^k \binom{k}{j} |j-\floor{k/2}|^p
 		\; ,
 	\]
 	and
 	\[
 		\lampar = 2^p \sum_{j=0}^k (-1)^j \binom{k}{j}|k/2 - j-t^*/2|^p = 2^p \sum_{j=0}^k  (-1)^j \binom{k}{j} |j - \floor{k/2}|^p
 		\; .
 	\]
 	(See also Eqs.~\eqref{eq:lambda} and~\eqref{eq:lambda-par}.)
 	Therefore,
 	\[
 		\vec{\beta} := H_{p,k}(t^*) \vec{\alpha} = \lambda \vec{1}+ (-1)^{\eta + b} \lampar \vpar
 		\; .
 	\]
 	In other words, $\beta_{\vec{u}} = \lambda + (-1)^{\eta + b} \lampar \prod_i u_i \geq 0$, so that the coordinates in $\vec{\beta}$ take just two values, depending only on $\prod_i u_i$. 
 	By Corollary~\ref{cor:crazy_binomial}, $\mathsf{sign}(\lampar) = (-1)^{\eta + 1}$ (where we take this statement to be true by convention if $\lampar = 0$), so that $\beta_{\vec{u}}$ is smaller when $\prod_i u_i = (-1)^{b}$. 
 	By Lemma~\ref{lem:alpha-to-weights}, $\beta_{\vec{u}} = \|V \vec{u} - \vec{t}^*\|_p^p$ in the corresponding parallelepiped. So, (up to scaling and change of coordinates) this gives a $(p, k, C_b, \eps)$-isolating parallelepiped with 
 	\begin{equation}
 	\label{eq:eps_lambda_lambda_par}
 		(1+\eps)^p = \frac{\lambda + |\lampar|}{\lambda - |\lampar|} \geq 1 + 2|\lampar|/\lambda
 		\; .
 	\end{equation}
 	
 	It remains to bound $\eps$. Indeed, by Corollary~\ref{cor:simple_binomial}
 	\[
 	\lambda \leq 44 2^p \cdot \binom{k}{\floor{k/2}} (pk/2)^{(p+1)/2}
 	\; .
 	\]
 	Similarly, by Corollary~\ref{cor:crazy_binomial}, 
 	\[
 		|\lampar| \geq 4 2^p \cdot |\sin(\pi p/2)| \binom{k}{\floor{k/2}} (p/(e\pi))^p
 		\; .
 	\]
 	Plugging this into Eq.~\eqref{eq:eps_lambda_lambda_par}, we see that
 	\begin{align*}
 		(1+\eps)^p 
 			&\geq 1+ \frac{|\sin(\pi p/2)| (p/(e\pi))^p}{6 (pk/2)^{(p+1)/2}} \\
 			&= 1 + |\sin(\pi p/2)| \cdot \frac{e\pi }{6p} \cdot \Big(\frac{2p}{ e^2 \pi^2 k}\Big)^{(p+1)/2} \\
 			&\geq 1 + 1.4 \cdot \frac{|\sin(\pi p/2)|}{p} \cdot \Big(\frac{2p}{ e^2 \pi^2 k}\Big)^{(p+1)/2}
 			\; .
 	\end{align*}
 	Finally, using the fact that $(1+\eps)^p \leq 1+1.4p\eps$ (for, e.g., $\eps \leq 1/(4p)$),
 	we have that
 	\[
 		\eps \geq \frac{|\sin(\pi p/2)|}{p^2} \cdot \Big(\frac{2p}{ e^2 \pi^2 k}\Big)^{(p+1)/2}
 		\; ,
 	\]
 	as claimed.
\end{proof}

\begin{corollary}
	\label{cor:il_parity}
	For any $k \geq 3$, $1 \leq p < k$, and $b \in \{0,1\}$, there exists a computable $(p, k, C_b, \eps)$-isolating lattice for some 
	\[
	\eps \geq \frac{|\sin(\pi p/2)|}{p^2} \cdot \frac{1}{1+2k/(3^p-1)} \cdot \Big(\frac{2p}{ e^2 \pi^2 k}\Big)^{(p+1)/2}  \geq \frac{|\sin(\pi p/2)|}{2p^2 k} \Big(\frac{2p}{ e^2 \pi^2 k}\Big)^{(p+1)/2} 
	\; .
	\]
\end{corollary}
\begin{proof}
	Combining Theorem~\ref{thm:ips_parity} with Proposition~\ref{prop:ip_to_lattice} yields an isolating lattice with 
	\[
		\eps \geq \frac{\eps'}{1 + k \mu}
		\; ,
	\]
	where 
	\[
		\eps' := \frac{|\sin(\pi p/2)|}{p^2} \cdot \Big(\frac{2p}{ e^2 \pi^2 k}\Big)^{(p+1)/2} 
		\; ,
	\]
	and $\mu := (1+\eps')^p/(3^p-1) \leq 2/(3^p-1)$. The result follows.
\end{proof}

\subsection{Gap-SETH hardness of CVP from isolating lattices}
\label{subsec:gap-seth-hardness-cvp}

\begin{theorem}
Let $\mathcal{C}$ be a $k$-CSP for some $k \in \Z^+$ and suppose that for some $p \in [1, \infty)$ and $\eps > 0$, there exists a computable $(p, k, C, \eps)$-isolating lattice for every $C \in \mathcal{C}$. 
Then, for every $0 < s \leq c \leq 1$ there exists a polynomial time (Karp) reduction from $(s, c)$-Gap-$\mathcal{C}$ instances on $n$ variables to $\gamma$-$\CVP_p$ instances of rank $n$ with $\gamma = \gamma(p,\eps, s,c)$ satisfying 
\[
	\gamma^p = \frac{1 - s(1-1/(1 + \eps)^p)}{1 - c(1-1/(1 + \eps)^p)}
	\; .
\]
\label{thm:cvp-hardness-cvips}
\end{theorem}

\begin{proof}
Let $\Phi$ be an $(s,c)$-Gap-$\mathcal{C}$ instance with $n$ variables and $m$ constraints $C_1, \ldots, C_m$.
Let $(V_1, \vec{t}_1^*), \ldots, (V_m, \vec{t}_m^*)$ be $(p, k, C, \eps)$-isolating lattices corresponding to the constraints $C = C_1, \ldots, C_m$, respectively.
We define the output $\gamma$-$\CVP_p$ instance $(B, \vec{t}, r)$ as follows.
 We set
\[
B := \begin{pmatrix} B_1 \\ \vdots \\ B_m \end{pmatrix}, \qquad \vec{t} := \begin{pmatrix} \vec{t}_1^* \\ \vdots \\ \vec{t}_m^* \end{pmatrix}\ ,
\]
with blocks $B_i \in \R^{d^* \times n}$ defined by
\[
(B_i)_j := \begin{cases}
(V_i)_s & \textrm{if $x_j$ is the $s$th variable of $C_i$} \ , \\
\vec{0}  & \textrm{otherwise} \ ,
\end{cases}
\]
for $1 \leq i \leq m$ and $1 \leq j \leq n$, where $(V_i)_s$ denotes the $s$th column of $V_i$. We set 
\[
r := ((1+\eps)^p - c((1 + \eps)^p-1))^{1/p} m^{1/p}
\; .
\]
Clearly, the reduction runs in polynomial time. The fact that $B$ is full-rank (and hence a lattice basis) follows from the fact that the $V_i$ are full-rank, assuming without loss of generality that all $n$ variables appear in $\Phi$.

For $\vec{y} \in \bit^n$,
\[
\norm{B\vec{y} - \vec{t}}_p^p = \sum_{i=1}^m \norm{B_i\vec{y} - \vec{t}_i}_p^p = m^+(\vec{y}) + (m - m^+(\vec{y})) \cdot (1 + \eps)^p \ ,
\]
where $m^+(\vec{y})$ denotes the number of constraints satisfied by $\vec{y}$.
It follows that if $\val(\Phi) \geq c$, then there exists $\vec{y} \in \bit^n$ such that 
\[
\norm{B\vec{y} - \vec{t}}_p^p \leq cm + (1-c) (1 + \eps)^p m = r^p
\; .
\] 
Alternatively, if $\val(\Phi) < 1 - \delta$ then for every $\vec{y} \in \Z^n$,
\[
\norm{B\vec{y} - \vec{t}}_p^p \geq \norm{B \vec{y}' - \vec{t}}_p^p > ((1-s) \cdot (1 + \eps)^p + s)m 
= \frac{(1+\eps)^p - s((1 + \eps)^p-1)}{(1+\eps)^p - c((1 + \eps)^p-1)} \cdot r^p \ ,
\]
where $\vec{y}'$ is an (arbitrary) vector satisfying $\vec{y}' \in \bit^n$ and $y_i' = y_i$ for coordinates $i$ such that $y_i \in \bit$.
Therefore, the output is an instance of $\gamma$-$\CVP$ with
\[
\gamma = \gamma(p, s, c) =\frac{(1 - s(1-1/(1 + \eps)^p))^{1/p}}{(1 - c(1-1/(1 + \eps)^p))^{1/p}} \ ,
\]
which is a `YES' instance if $\Phi$ is a `YES' instance and a `NO' instance if $\Phi$ is a `NO' instance, as needed.
\end{proof}

We are now ready to prove the main hardness result in this section, Theorem~\ref{thm:parity_reduction}.

\begin{proof}[Proof of Theorem~\ref{thm:parity_reduction}]
	Combining Theorem~\ref{thm:cvp-hardness-cvips} with Corollary~\ref{cor:il_parity} gives
	\[
		\gamma \geq \Big( \frac{1 - s(1-1/(1 + \eps)^p)}{1 - c(1-1/(1 + \eps)^p)} \Big)^{1/p}
		\; ,
	\]
	where 
	\[
		\eps := \frac{|\sin(\pi p/2)|}{2p^2 k} \Big(\frac{2p}{ e^2 \pi^2 k}\Big)^{(p+1)/2} 
		\; .
	\]
	We need to show that $\gamma \geq 1+ \eps (c-s)/(2p)$. 
	
	To that end, let
	\[
		f(\delta,\alpha) := \frac{1 - (c-\delta)(1-1/\alpha)}{1 - c(1-1/\alpha)} \cdot \frac{1}{(1+ \delta (\alpha^{1/p}-1)/(2p))^p}
		\; .
	\]
	Notice that $\gamma^p/(1+\eps (c-s)/p)^p \geq f(c-s, (1+\eps)^p)$, and $f(0,\alpha) = 1$. So, it suffices to show that $f(\delta,\alpha)$ is increasing in $\delta$ for $1/2 \leq \delta \leq c \leq 1$ and $1 \leq \alpha \leq 2$.
	
	We have
	\[
		\frac{\partial }{\partial \delta} \log(f(\delta, \alpha)) = \frac{\alpha-1}{\alpha - (c-\delta)(\alpha-1)}-\frac{p(\alpha^{1/p} -1)}{2p + \delta (\alpha^{1/p}-1)} \geq  \frac{\alpha-1}{\alpha}-\frac{\alpha -1}{2} \geq 0
		\; ,
	\]
	as needed.
\end{proof}

\section{Hardness of CVPP from on-off isolating parallelepipeds}
\label{sec:cvpp-hardness}

In this section, we substantially improve the quantitative hardness results from~\cite{conf/focs/BennettGS17} for $\CVPP_p$.~\cite{conf/focs/BennettGS17} showed $2^{\Omega(\sqrt{n})}$-hardness of $\CVPP_p$ for all $p \in [1, \infty)$ assuming non-uniform ETH, and did not show any additional hardness assuming non-uniform SETH. Here we show $2^{\Omega(n)}$-hardness of $\CVPP_p$ for all $p \neq 2$ (including even integers other than $2$) assuming non-uniform ETH, and $2^{(1-\eps)n}$-hardness of $\CVPP_p$ for all $p \notin 2\Z$ assuming non-uniform SETH. We also show both of these results for $p = \infty$. We do not show any improved hardness for the case where $p = 2$, which remains a tantalizing open question.

\begin{theorem}
The following hardness results hold for $\CVPP_p$:
\begin{enumerate}
\item For every $p \in [1, \infty) \setminus 2\Z$ and $\eps > 0$, there is no $2^{(1 - \eps)n}$-time algorithm for $\CVPP_p$ assuming non-uniform Max-SAT-SETH. In particular, there is no $2^{(1 - \eps)n}$-time algorithm for $\CVPP_p$ assuming non-uniform SETH. \label{en:cvpp-seth-hardness-finite-p}
\item For every $p \geq 1$, $p \neq 2$, there is no $2^{o(n)}$-time algorithm for $\CVPP_p$ assuming non-uniform Max-SAT-ETH. In particular, there is no $2^{o(n)}$-time algorithm for $\CVPP_p$ assuming non-uniform ETH. \label{en:cvpp-eth-hardness-finite-p}
\item For every $\eps > 0$, there exists a $\gamma(\eps) > 1$ such that there is no $2^{(1 - \eps)n}$-time algorithm that approximates $\CVPP_{\infty}$ to within a factor of $\gamma(\eps)$ assuming non-uniform SETH.
\label{en:cvpp-seth-hardness-infinity}
\end{enumerate}
\label{thm:cvpp-hardness-main}
\end{theorem}

Items~\ref{en:cvpp-seth-hardness-finite-p} and~\ref{en:cvpp-seth-hardness-infinity} together assert that we get the same $2^{(1 - \eps)n}$ hardness of $\CVPP_p$ for $p \notin 2\Z$ that we get for $\CVP_p$ (assuming non-uniform SETH). Furthermore, Item~\ref{en:cvpp-seth-hardness-infinity} gives hardness of approximation for $\CVPP_{\infty}$ (with a reasonably large $\gamma$), which is similar to the case for $\CVP_{\infty}$~\cite[Theorem 6.5]{conf/focs/BennettGS17}.
Item~\ref{en:cvpp-eth-hardness-finite-p} asserts that for \emph{every} $p \neq 2$, $\CVPP_p$ takes $2^{\Omega(n)}$-time assuming non-uniform ETH. We emphasize that, interestingly, this lower bound holds for even integers $p = 4, 6, \ldots$ greater than $2$, therefore yielding a stronger hardness result for $\CVPP_p$ for \emph{all} values of $p \neq 2$ than what is known for $p = 2$.

\subsection{On-off isolating parallelepipeds}

We show these results by defining a family of geometric gadgets called ``$(p, k)$-on-off isolating parallepeipeds'' that are defined by vectors $\vec{v}_1, \ldots, \vec{v}_k$ and \emph{two} targets $\ton$ and $\toff$, and then showing that such gadgets exist if and only if ``normal'' $(p, k + 1)$-isolating parallepipeds exist.
As the name suggests, $(p, k)$-on-off isolating parallelepipeds will allow us to ``turn clauses on and off.'' More precisely, for a given $n$ and $k$, we will output a single basis $B = (\vec{b}_1, \ldots, \vec{b}_n)$ as preprocessing. Then, given a $k$-SAT instance $\Phi$ on $n$ variables, we will output a target vector $\vec{t}$ that uses copies of $\ton$ to ``turn on'' row blocks in $B$ corresponding to all clauses in $\Phi$, and copies of $\toff$ to ``turn off'' row blocks in $B$ corresponding to clauses not in $\Phi$. 

The high-level strategy of outputting a basis $B$ that ``represents all clauses possible in an $n$-variable $k$-SAT instance'' as preprocessing, and then, given a $k$-SAT instance $\Phi$ on $n$ variables, of ``turning on and off clauses'' according to whether they appear in $\Phi$ using the query target $\vec{t}$ is the same as was used in~\cite[Lemma 6.1]{conf/focs/BennettGS17}. However, here we use a different framework for turning on and off clauses, and use it to output bases $B$ of lower rank, leading to improved hardness results.

\begin{definition}[On-off isolating parallelepiped]
For $1 \leq p \leq \infty$ and $k \in \Z+$, we say that $V \in \R^{d^* \times k}$, $\ton \in \R^{d^*}$, and $\toff \in \R^{d^*}$ define a \emph{$(p, k)$-on-off isolating parallelepiped} if:
\begin{enumerate}
\item For all $\vec{x} \in \bit^k \setminus \set{\vec{0}}$, $\norm{V\vec{x} - \ton}_p = 1$.
\item $\norm{V\vec{\vec{0}} - \ton}_p = \norm{\ton}_p > 1$.
\item For all $\vec{x} \in \bit^k$, $\norm{V\vec{x} - \toff}_p = 1$.\footnote{It is natural to ask whether the given definition of an on-off isolating parallelepiped is sufficiently general.
Indeed, one could define three different radii $\rgood := \norm{V\vec{x} - \ton}_p$ for $\vec{x} \in \bit^k \setminus \set{\vec{0}}$, $\rbad := \norm{\vec{\ton}}_p$, and $\roff := \norm{V\vec{x} - \toff}_p$ for $\vec{x} \in \bit^k$ corresponding to the three cases in the definition (with the requirement that $\rgood < \rbad$). However, given $V, \ton, \toff$ satisfying these conditions for some $\rgood, \rbad, \roff$, we can output another $(p, k)$-on-off isolating parallelepiped that achieves $\roff = \rgood = 1$ simply by appending a coordinate of value $\abs{\rgood^p - \roff^p}^{1/p}$ to $\toff$ if $\rgood > \roff$ and to $\ton$ if $\roff > \rgood$, and then normalizing. So, the definition given is essentially without loss of generality.}
\end{enumerate}
\label{def:on-off-ips}
\end{definition}

We note that the first two conditions are the same as in the definition of ``normal'' isolating parallelepipeds (Definition~\ref{def:ip}) with $\ton$ taking the role of $\vec{t}^*$. 
As in the case of isolating parallelepipeds, the $2^k - 1$ close vectors $V\vec{x}$ for $\vec{x} \in \bit^k \setminus \set{\vec{0}}$ to $\ton$ correspond to the $2^k - 1$ possible satisfying assignments to the variables of a $k$-clause, and the more distant vector $\vec{0}$ corresponds to the single falsifying assignment to the variables of a $k$-clause.
The new third condition asserts that all $2^k$ vectors $V\vec{x}$ for $\vec{x} \in \bit^k$ are equally close to $\toff$, which says that the distance between $V\vec{x}$ and $\toff$ will be the same regardless of whether the corresponding clause is satisfied or not. In other words, by using $\toff$ in place of $\ton$ (or $\vec{t}^*$), we will be able to ``turn off'' a clause so that its satisfiability is irrelevant.

The following proposition gives a construction of a $(p, k)$-on-off isolating parallelepiped from a $(p, k + 1)$-isolating parallelepiped and vice-versa, therefore showing that one of these objects exists if and only if the other one does.
\begin{proposition}
For every $p \in [1, \infty)$ and integer $k \geq 1$, there exists a computable $(p, k)$-on-off isolating parallelepiped if and only if there exists a computable $(p, k + 1)$-isolating parallelepiped.
\label{prop:equiv-ips-ooips}
\end{proposition}

\begin{proof}
Suppose that $V = (\vec{v}_1, \ldots, \vec{v}_{k+1})$, $\vec{t}^*$ define a $(p, k + 1)$-isolating parallelepiped. Set $V' := (\vec{v}_1, \ldots, \vec{v}_k)$, set $\ton := \vec{t}^*$, and set $\toff := \vec{t}^* - \vec{v}_{k+1}$. It is straightforward to check that $V', \ton, \toff$ define a $(p, k)$-on-off isolating parallelepiped.

Suppose that $V = (\vec{v}_1, \ldots, \vec{v}_k)$, $\ton$, $\toff$ define a $(p, k)$-on-off isolating parallelepiped. Set $v_i' := v_i$ for $i = 1, \ldots, k$, set $v_{k+1}' := \ton - \toff$, and set $\vec{t}^* := \ton$. It is straightforward to check that $V' := (v_1', \ldots, v_{k+1}')$, $\vec{t}^*$ define a $(p, k + 1)$-isolating parallelepiped.
\end{proof}

\subsection{Hardness of CVPP from on-off isolating parallelepipeds}

The following theorem gives a non-uniform reduction from Max-$k$-SAT formulas on $n$ variables to $\CVPP_p$ instances of rank $n$, assuming that $(p, k)$-on-off-isolating parallelepipeds exist.\footnote{However, as a technical difference, the reduction below works as a reduction from MAX-$k$-SAT (or weighted MAX-$k$-SAT with polynomial integer weights), but not as a reduction from weighted MAX-$k$-SAT with arbitrary weights as in~\cite[Theorem 3.2]{conf/focs/BennettGS17}. This is because the reduction in~\cite[Theorem 3.2]{conf/focs/BennettGS17} requires scaling rows of both the basis matrix and target vector, and now we must output the basis matrix before we know the weights of the input weighted MAX-$k$-SAT instance.}

\begin{theorem}%
If there exists a computable $(p, k)$-on-off isolating parallelepiped defined by $V = (\vec{v}_1, \ldots, \vec{v}_k) \in \R^{d^* \times k}$, $\ton \in \R^{d^*}$, $\toff \in \R^{d^*}$ for some $p \in [1, \infty)$ and $k \in \Z^+$, then there exist a pair of polynomial-time algorithms $(P, Q)$ (in analogy to the definition of CVPP) that behave as follows.
\begin{enumerate}
\item On input $n \in \Z^+$, $P$ outputs a basis $B \in \R^{d \times n}$ of a rank $n$ lattice $\lat$, where $d = 2^k \binom{n}{k} d^* + n$.
\item On input a Max-$k$-SAT instance with $n$ variables, $Q$ outputs a target vector $\vec{t} \in \R^d$ and a distance bound $r \geq 0$ such that $\dist_p(\vec{t}, \lat) \leq r$ if and only if the input is a `YES' instance.
\end{enumerate}
\label{thm:cvpp-hardness-ooips}
\end{theorem}

\begin{proof}
Let $M := 2^k \cdot \binom{n}{k} = O(n^k)$ be the total possible number of $k$-clauses on $n$ variables, and let $C_1, \ldots, C_M$ denote those clauses.
By assumption, there exists a $(p, k)$-isolating parallelepiped $V, \ton, \toff$ with $\norm{\ton}_p = 1 + \eps$ for some $\eps > 0$.

The algorithm $P$ constructs the basis $B \in \R^{d \times n}$ as
\[
B := \begin{pmatrix} B_1 \\ \vdots \\ B_M \\ 2\alpha \cdot I_n \end{pmatrix} \ ,
\]
for $\alpha := M^{1/p} \cdot (1 + \eps)$ and with blocks $B_i \in \R^{d^* \times n}$ defined by
\[
(B_i)_j := \begin{cases}
\vec{v}_s & \textrm{if $x_j$ is the $s$th literal of $C_i$}\ , \\
- \vec{v}_s & \textrm{if $\lnot x_j$ is the $s$th literal of $C_i$}\ , \\
\vec{0}  & \textrm{otherwise} \ ,
\end{cases}
\]
for $1 \leq i \leq M$ and $1 \leq j \leq n$.

Given an instance $(\Phi, W)$ of Max-$k$-SAT with $m$ clauses, 
the algorithm $Q$ outputs
$\vec{t} \in \R^d$ defined by
\[
\vec{t} := \begin{pmatrix} \vec{t}_1 \\ \vdots \\ \vec{t}_M \\ \alpha \cdot \vec{1} \end{pmatrix}\ ,
\]
where $\vec{t}_i := \ton - \sum_{s \in N_i} \vec{v}_s$ if $C_i$ is in $\Phi$ and $\vec{t}_i := \toff  - \sum_{s \in N_i} \vec{v}_s$ if $C_i$ is not in $\Phi$ for $1 \leq i \leq M$, and
\[
 r := ((M - (m - W)) + (m - W) \cdot (1 + \eps)^p + n \cdot \alpha^p)^{1/p} \ .
\]

Clearly, both $P$ and $Q$ run in polynomial time. We next analyze for which $\vec{y} \in \Z^n$ it holds that $\norm{B\vec{y} - \vec{t}}_p \leq r$.
Note that by the definition of $\alpha$ above, $\alpha^p = M \cdot (1 + \eps)^p \geq (M - (m - W)) + (m - W) \cdot (1 + \eps)^p$ for all $m$ and $W$.
Therefore, for $y \notin \bit^n$, $\norm{B\vec{y} - \vec{t}}_p^p \geq \alpha^p \sum_{i=1}^n \abs{2 y_i - 1}^p \geq (n + 2) \cdot \alpha^p > r^p$. So, we only need to analyze the case where $\vec{y} \in \bit^n$.

Consider an assignment $\vec{y} \in \bit^n$ to the variables of $\Phi$. Then for $1 \leq i \leq M$ such that $C_i$ is in $\Phi$,
\begin{align*}
\norm{B_i \vec{y} - \vec{t}_i}_p &= \Big\|\sum_{s \in P_i} y_{\ind(\ell_{i, s})} \cdot \vec{v}_s - \sum_{s \in N_i} y_{\ind(\ell_{i, s})} \cdot \vec{v}_s - \Big(\ton - \sum_{s \in N_i} \vec{v}_s \Big)\Big\|_p \\
          &= \Big\|\sum_{s \in P_i} y_{\ind(\ell_{i, s})} \cdot \vec{v}_s + \sum_{s \in N_i} \big(1 - y_{\ind(\ell_{i, s})} \big) \cdot \vec{v}_s - \ton \Big\|_p \\
          &= \Big\|\sum_{s \in S_i(\vec{y})} \vec{v}_s - \ton \Big\|_p \ .
\end{align*}
By assumption, the last quantity is equal to $1$ if $|S_i(\vec{y})| \geq 1$ and is equal to $1 + \eps$ otherwise. A similar argument shows that for $1 \leq i \leq M$ such that $C_i$ is not in $\Phi$,
\[
\norm{B_i \vec{y} - \vec{t}_i}_p = \Big\|\sum_{s \in S_i(\vec{y})} \vec{v}_s - \toff \Big\|_p = 1
\]
regardless of $\vec{y}$.

Because $|S_i(\vec{y})| \geq 1$ if and only if $C_i$ is satisfied, it follows that
\[
\norm{B \vec{y} - \vec{t}}_p^p = \Big(\sum_{i=1}^M \norm{B_i\vec{y} - \vec{t}_i}_p^p\Big) + n \cdot \alpha^p = M - (m - m^+(\vec{y})) + (m - m^+(\vec{y})) \cdot (1 + \eps)^p + n \cdot \alpha^p \ .
\]
Therefore, $\norm{B\vec{y} - \vec{t}}_p \leq r$ if and only if $m^+(\vec{y}) \geq W$, and therefore there exists $\vec{y}$ such that $\norm{B\vec{y} - \vec{t}}_p \leq r$ if and only if $(\Phi, W)$ is a `YES' instance of MAX-$k$-SAT, as needed.
\end{proof}

We then get Theorem~\ref{thm:cvpp-hardness-main} Items~\ref{en:cvpp-seth-hardness-finite-p} and~\ref{en:cvpp-eth-hardness-finite-p} about the hardness of $\CVPP_p$ assuming (non-uniform, Max-SAT versions of) SETH and ETH, respectively. 

\begin{proof}[Proof of Theorem~\ref{thm:cvpp-hardness-main}, Items~\ref{en:cvpp-seth-hardness-finite-p} and~\ref{en:cvpp-eth-hardness-finite-p}]
Combine Theorem~\ref{thm:ip-characterization}, Proposition~\ref{prop:equiv-ips-ooips}, and Theorem~\ref{thm:cvpp-hardness-ooips}.
\end{proof}

\subsection{SETH Hardness of \texorpdfstring{$\CVPP_{\infty}$}{CVPP\_infinity}}
Finally, we give a non-uniform reduction from Max-$k$-SAT formulas on $n$ variables to $\CVPP_{\infty}$ instances of rank $n$.

\begin{theorem}%
For every $k \in \Z^+$, there exists a pair of polynomial-time algorithms $(P, Q)$ (in analogy to the definition of CVPP) that behave as follows.
\begin{enumerate}
\item On input $n \in \Z^+$, $P$ outputs a basis $B \in \R^{d \times n}$ of a rank $n$ lattice $\lat$, where $d = 2^k \binom{n}{k} + n$.
\item On input a $k$-SAT instance with $n$ variables, $Q$ outputs a target vector $\vec{t} \in \R^d$ such that $\dist_{\infty}(\vec{t}, \lat) \leq k/2$ if and only if the input is a `YES' instance.
\end{enumerate}
\label{thm:cvpp-inf}
\end{theorem}

\begin{proof}
Let $M := 2^k \cdot \binom{n}{k} = O(n^k)$ be the total possible number of $k$-clauses on $n$ variables, and let $C_1, \ldots, C_M$ denote those clauses.

The algorithm $P$ constructs the basis $B \in \R^{d \times n}$ as
\[
B := \begin{pmatrix} \vec{b}_1^T \\ \vdots \\ \vec{b}_M^T \\ k \cdot I_n \end{pmatrix} \ ,
\]
and with rows $\vec{b}_i^T$ defined by
\[
(B_i)_j := \begin{cases}
1 & \textrm{if $x_j$ is the $s$th literal of $C_i$}\ , \\
- 1 & \textrm{if $\lnot x_j$ is the $s$th literal of $C_i$}\ , \\
\vec{0}  & \textrm{otherwise} \ ,
\end{cases}
\]
for $1 \leq i \leq M$ and $1 \leq j \leq n$.

Given an instance $\Phi$ of $k$-SAT with $m$ clauses, 
the algorithm $Q$ outputs
$\vec{t} \in \R^d$ defined by
\[
\vec{t} := \begin{pmatrix} t_1 \\ \vdots \\ t_M \\ \frac{k}{2} \cdot \vec{1} \end{pmatrix}\ ,
\]
where $t_i := (k + 1)/2 - \card{N_i}$ if $C_i$ is in $\Phi$ and $t_i := k/2  - \card{N_i}$ if $C_i$ is not in $\Phi$ for $1 \leq i \leq M$, and where $r := k/2$.

Clearly, both $P$ and $Q$ run in polynomial time. We next analyze for which $\vec{y} \in \Z^n$ it holds that $\norm{B\vec{y} - \vec{t}}_{\infty} \leq r = k/2$.
If $y \notin \bit^n$, $\norm{B\vec{y} - \vec{t}}_{\infty} \geq \max_{i \in [n]} \abs{y_i \cdot k - k/2} \geq 3k/2$. So, we only need to analyze the case where $\vec{y} \in \bit^n$.

Consider an assignment $\vec{y} \in \bit^n$ to the variables of $\Phi$. Then for $1 \leq i \leq M$ such that $C_i$ is in $\Phi$,
\begin{align*}
\Big|\langle\vec{b}_i, \vec{y}\rangle - t_i\Big| &= \Big|\sum_{s \in P_i} y_{\ind(\ell_{i, s})} - \sum_{s \in N_i} y_{\ind(\ell_{i, s})} - ((k+1)/2 - |N_i|) \Big| \\
                                                        &= \Big|\sum_{s \in P_i} y_{\ind(\ell_{i, s})} - \sum_{s \in N_i} (1 - y_{\ind(\ell_{i, s})}) - (k+1)/2 \Big| \\
                                                        &= \Big||S_i(\vec{y})| - (k+1)/2 \Big|.
\end{align*}
It follows that if $|S_i(\vec{y})| = 0$ then $|\langle\vec{b}_i, \vec{y}\rangle - t_i| = (k+1)/2$, and otherwise $|\langle\vec{b}_i, \vec{y}\rangle - t_i| \leq (k - 1)/2$. 
Because $|S_i(\vec{y})| \geq 1$ if and only if clause $C_i$ is satisfied, it follows that $|\langle\vec{b}_i, \vec{y}\rangle - t_i| \leq (k - 1)/2$ if and only if clause $C_i$ is satisfied.

A similar argument shows that for $1 \leq i \leq M$ such that $C_i$ is not in $\Phi$,
\[
|\langle\vec{b}_i, \vec{y}\rangle - t_i| = ||S_i(\vec{y})| - k/2|
\]
regardless of $\vec{y}$.

Therefore for $\vec{y} \in \bit^n$, $\max_{1 \leq i \leq M} |\langle\vec{b}_i, \vec{y}\rangle - t_i|$ is less than or equal to $k/2$ if every clause in $\Phi$ is satisfied, and is greater than $(k+1)/2$ if there exists a clause in $\Phi$ that is not satisfied.
It follows that
\[
\norm{B\vec{y} - \vec{t}}_{\infty} = \max \set{|\langle\vec{b}_1, \vec{y}\rangle - t_1|, \ldots, |\langle\vec{b}_m, \vec{y}\rangle - t_m|, k/2} = k/2 = r
\]
if $\vec{y}$ satisfies $\Phi$, and $\norm{B\vec{y} - \vec{t}}_{\infty} \geq (k + 1)/2 > r$ if not.
Therefore, there exists $\vec{y} \in \bit^n$ that satisfies $\Phi$ if and only if there exists $\vec{y} \in \bit^n$ that satisfies $\norm{B\vec{y} - \vec{t}}_{\infty}$, as needed.
\end{proof}

Theorem~\ref{thm:cvpp-hardness-main}, Item~\ref{en:cvpp-seth-hardness-infinity} follows as a corollary.
\section{Limitations}
\label{sec:limitations}

\subsection{Impossibility of \texorpdfstring{$(p,k)$}{(p,k)}-isolating parallelepipeds for even integer \texorpdfstring{$p < k$}{p<k}}
In~\cite{conf/focs/BennettGS17}, we proved that there do not exist $(2, 3)$-isolating parallelepipeds, and noted that there are no $(p, p+1)$-isolating parallelepipeds for $p\in2\Z$. Here, we give a simple geometric proof of the non-existence of $(2, 3)$-isolating parallelepipeds, and we also prove that  there are no $(p, p+1)$-isolating parallelepipeds for $p\in2\Z$. This finishes the complete characterization of values of $p$ and $k$ such that $(p,k)$-isolating parallelepipeds exist, as presented in Theorem~\ref{thm:ip-characterization}.

\begin{lemma}
Suppose that $V = (\vec{v}_1, \vec{v}_2, \vec{v}_3) \in \R^{d \times 3}$, $\vec{t} \in \R^d$, and $\norm{V\vec{x} - \vec{t}} = 1$ for all $\vec{x} \in \bit^3 \setminus \set{\vec{0}}$. Then $\norm{\vec{t}} = \norm{V\vec{x} - \vec{t}}$ for $\vec{x} \in \bit^3 \setminus \set{\vec{0}}$, and hence $V, \vec{t}$ do not form an isolating parallelepiped.
\label{lem:same-distances-imply-cube}
\end{lemma}

\begin{proof}
For $\vec{p}=\vec{t}-\vec{v}_3$, by assumption we have that
\[
\norm{\vec{p}}=\norm{\vec{v}_1-\vec{p}}=\norm{\vec{v}_2-\vec{p}}=\norm{\vec{v}_1+\vec{v}_2-\vec{p}}=r \ .
\]

Let us consider a plane $P$ passing through the points $\vec{0}, \vec{v}_1$ and $\vec{v}_2$, and let $\vec{p}^*$ be the projection of $\vec{p}$ onto $P$. Consider the parallelogram $D$ formed by the points $\vec{0}, \vec{v}_1, \vec{v}_2$ and $\vec{v}_1+\vec{v}_2$. These points lie on a circle around the point $\vec{p}^*$. Therefore, $D$ is a cyclic parallelogram, i.e., a rectangle. %

Let $\vec{t}^*$ be the projection of $\vec{t}$ onto $P$. Let $\norm{\vec{v}_1-\vec{t}^*}=\norm{\vec{v}_2-\vec{t}^*}=\norm{\vec{v}_1+\vec{v}_2-\vec{t}^*}=r'$. Since the three points of the rectangle formed by the points $\vec{0}, \vec{v}_1, \vec{v}_2$ and $\vec{v}_1+\vec{v}_2$ lie on the circle of radius $r'/2$ around the point $\vec{t}^*$, the fourth point of this rectangle also lies on this circle. Thus, $\norm{\vec{t}}=\norm{\vec{v}_1-\vec{t}}$. 
\end{proof}

\begin{corollary}
There do not exist $(2, k)$-isolating parallelepipeds for $k \geq 3$.
\end{corollary}

\begin{lemma}
For every $p\in2\Z$, integers $d$ and $k>p$, and vectors $\vec{v}_1,\ldots,\vec{v}_k, \vec{t} \in \R^{d}$,
we have
\[
\label{eq:incl-excl}
\sum_{S\subseteq[k]}(-1)^{|S|} \, \Big\|\vec{t} - \sum_{i\in S}\vec{v}_i\Big\|_p^p=0 \; .
\]
\end{lemma}
\label{lem:incl-excl}
\begin{proof}
We will use the Multinomial theorem which states that 
\[
(x_1+\ldots+x_m)^n=\sum_{a_1+\ldots+a_m=n}\binom{n}{a_1,\ldots,a_m} \prod_{t=1}^m 	x_t^{a_t}\ ,
\]
where 
\[
\binom{n}{a_1,\ldots,a_m}=\frac{n!}{a_1!\cdots a_m!} \ .
\]
Let $\vec{t}=(t_1,\ldots, t_d)$ and for an $i\in[k], \vec{v}_i=(v_{i,1},\ldots,v_{i,d})$. For a set $S\in[k]$, and an integer $1\leq i\leq |S|$, let $S_i$ be the $i$th element of the set $S$.
Then we have that for $p\in2\Z$,
\begin{align*}
\sum_{S\subseteq[k]}(-1)^{|S|} \, \Big\|\vec{t} - \sum_{i\in S}\vec{v}_i\Big\|_p^p
&= \sum_{S\subseteq[k]}(-1)^{|S|} \, \sum_{j=1}^d (t_j-\sum_{i=1}^{|S|}v_{S_i,j})^p \\
&= \sum_{S\subseteq[k]}(-1)^{|S|} \, \sum_{j=1}^d \sum_{a_0+\ldots+a_{|S|}=p} \binom{p}{a_0,\ldots,a_{|S|}} t_j^{a_0}\prod_{i=1}^{|S|}v_{S_i,j}^{a_i} \\
&= \sum_{j=1}^d \sum_{a_0+\ldots+a_{k}=p} \binom{p}{a_0,\ldots,a_{k}} t_j^{a_0}\prod_{i=1}^{k}v_{i,j}^{a_i} \cdot  \sum_{S\supseteq\{i\colon a_i\neq0\}}(-1)^{|S|}\\
&= \sum_{j=1}^d \sum_{a_0+\ldots+a_{k}=p} \binom{p}{a_0,\ldots,a_{k}} t_j^{a_0}\prod_{i=1}^{k}v_{i,j}^{a_i} \cdot  (1-1)^{k-|\{i\colon a_i\neq0\}|}\\
&=0 \ ,
\end{align*}
where the last equality follows from $|\{i\colon a_i\neq0\}|\leq p < k$.
\end{proof}

\begin{corollary}
Let $p\in2\Z$. There do not exist $(p, k)$-isolating parallelepipeds for $k > p$.
\label{cor:no-ips-kgeqp}
\end{corollary}
\begin{proof}
Suppose towards a contradiction that $V = (\vec{v}_1, \ldots, \vec{v}_k) \in \R^{d \times k}$ and $\vec{t} \in \R^d$ form an isolating parallelepiped. Then for all $\vec{x} \in \bit^k \setminus \set{\vec{0}}$,  $\norm{V\vec{x} - \vec{t}} = 1$. By Lemma~\ref{lem:incl-excl}, 
\[
\norm{\vec{t}}_p^p = \sum_{\emptyset\neq S\subseteq [k]}(-1)^{|S|+1}\,\norm{t-\sum_{i\in S}v_i}_p^p=\sum_{\emptyset\neq S\subseteq [k]}(-1)^{|S|+1}=1
\ . \qedhere
\]
\end{proof}

\subsection{Impossibility of natural reductions for \texorpdfstring{$p = 2$}{p=2}}

For a lattice $\lat \subset \R^d$ with basis $\basis \in \R^{d \times n}$ and target vector $\vec{t} \in \R^d$, let 
\[
\CVP(\vec{t}, \basis) := \{ \vec{z}\in \Z^n \ : \ \| \basis \vec{z} - \vec{t}\|_2 = \dist_2(\vec{t},\lat) \}
\] 
be the set of the coordinates of closest lattice vectors to $\vec{t}$.

\begin{definition}
	A \emph{natural reduction} from $k$-SAT to $\CVP_2$ is a (not necessarily efficient) mapping from $k$-SAT instances on $n$ variables to $\CVP_2$ instances $\basis \in \R^{d \times n'}, \vec{t}\in \R^{d}$ such that there exists a (not necessarily efficiently computable) fixed function $f : \{0,1\}^n \to \Z^{n'}$ with the following property. 
If the input $k$-SAT instance is satisfiable, then $\vec{x} \in \bit^n$ is a satisfying assignment if and only if $f(x) \in \CVP(\vec{t}, \basis)$.
\end{definition}

In other words, a natural reduction is one in which there exists a fixed function $f$ such that, if the input to the reduction is a satisfiable formula, $f$ forms a bijection between satisfying assignments and (coordinates of) closest vectors.
The following theorem shows that no natural reduction can rule out a $2^{3n/4}$-time algorithm for $\CVP_2$ under SETH.

\begin{theorem}
	\label{thm:no_natural}
	Every natural reduction from $3$-SAT on $n$ variables to $\CVP_2$ on rank $n'$ lattices must have $n' > 4(n-2)/3$.
\end{theorem}

We remark that Theorem~\ref{thm:no_natural} trivially also applies to natural reductions from $k$-SAT for $k > 3$, but that it remains an interesting open question to show a quantitatively stronger result for some such $k$. Doing so would require improving or generalizing several of the following lemmas.
Additionally, we note that ``natural reductions'' are a special type of Karp reduction. It is an interesting question whether we can extend Theorem~\ref{thm:no_natural} to rule out a broader class of reductions such as some natural class of Cook reductions.

To prove Theorem~\ref{thm:no_natural}, we study the structure of $A := f(\{0,1\}^n)$ modulo two. In particular, we will show that $A$ cannot contain any affine $3$-cube modulo two. The next lemma is a version of Szemer\'{e}di's cube lemma for the boolean cube, which shows that any such set must be small (relative to $n'$). To the authors' knowledge, our proof is novel and significantly simpler than that of prior work (e.g.,~\cite[Lemma 3.1]{cs16}). We also obtain a tighter bound.

\begin{lemma}
	\label{lem:additive_combinatorics}
	Let $d\geq1$ be an integer. Every set $S\subseteq\F_2^n$ of size $|S|\ge 2^{n(1-2^{-(d-1)})+2}$ contains an affine subspace of dimension $d$.
\end{lemma}
\begin{proof}
	We prove the result by induction on $d$. For $d=1$, we have $|S| \ge 4$, and so the statement is trivially true since any set with $2$ elements contains an affine subspace of dimension $1$.
	
	Now we assume the result is true for $d=k$, and show that it is true for $d=k+1\ge 2$. Let $S := \{\vec{x}_1, \ldots, \vec{x}_N\}$, where $N = |S| \ge 2^{n(1-2^{-k})+2}$. Consider all ${N \choose 2}$ distinct pairs of elements in $S$. By the pigeon-hole principle, at least \[M=\frac{N(N-1)}{2 \cdot 2^n} \ge \frac{N^2}{4 \cdot 2^n} = 2^{n(1-2^{-(k-1)})+2}\] distinct pairs have the same sum, say $\vec{z}_0 \in \F_2^n$. Without loss of generality, let these pairs be $(\vec{x}_1, \vec{x}_1 +\vec{z}_0), (\vec{x}_2, \vec{x}_2 + \vec{z}_0), \ldots, (\vec{x}_M, \vec{x}_M + \vec{z}_0)$.
	
	By the induction hypothesis, there exist $\vec{z}^*$, and linearly independent vectors $\vec{z}_1, \ldots, \vec{z}_k$ such that the set  $\{\vec{x}_1, \ldots, \vec{x}_M\}$ contains every element of the form $\vec{z}^* + \sum_{i=1}^k \sigma_i \vec{z}_i $ where $\sigma_i \in \{0,1\}$ for $1\le i \le k$. 
	
	This implies that $S$ contains every element of the form $\vec{z}^* + \sum_{i=0}^k \sigma_i \vec{z}_i $ where $\sigma_i \in \{0,1\}$ for $0 \le i \le k$. To complete the proof, we need to show that $\vec{z}_0$ is not in the span of $\vec{z}_1, \ldots, \vec{z}_k$. But this is immediate %
	from the fact that each of the $M$ pairs above contains distinct elements.
\end{proof}

This next lemma shows that the coordinates of closest vectors have some additional structure modulo two. In particular, if $\vec{z}_1, \vec{z}_2, \vec{z}_3, \vec{z}_4 \in \CVP(\vec{t}, \basis)$ form a square modulo two (i.e., a two-dimensional affine subspace), then either they form a parallelogram over the reals or there must be some specific set of four other vectors $\vec{z}_1',\vec{z}_2',\vec{z}_3',\vec{z}_4' \in \CVP(\vec{t},\basis)$. We will then use this to argue that $A := f(\{0,1\}^n)$ cannot contain any affine $3$-cubes modulo two.

\begin{lemma}
	\label{lem:4to8}
	For any lattice $\lat \subset \R^d$ with rank $n \geq 2$ and basis $\basis \in \R^{d \times n}$ and any target $\vec{t} \in \R^d$, suppose that $\vec{z}_1, \vec{z}_2, \vec{z}_3, \vec{z}_4 := \vec{z}_1 + \vec{z}_2 + \vec{z}_3 - 2\vec{v} \in \CVP(\vec{t}, \basis)$ are coordinates of distinct closest lattice vectors with $\vec{v} \in \Z^n$. Then $\vec{z}_1', \vec{z}_2', \vec{z}_3', \vec{z}_4' \in \CVP(\vec{t}, \basis)$ where
	\[
	\vec{z}_1' := \vec{z}_2 + \vec{z}_3 - \vec{v}, \quad 
        \vec{z}_2' := \vec{z}_1 + \vec{z}_3 - \vec{v}, \quad 
        \vec{z}_3' := \vec{z}_1 + \vec{z}_2 - \vec{v}, \quad 
        \vec{z}_4' := \vec{v}
	\; .
	\]	
	In particular, $C := \set{\vec{z}_1, \vec{z}_2, \vec{z}_3, \vec{z}_4} \cup \set{\vec{z}_1', \vec{z}_2', \vec{z}_3', \vec{z}_4'}$ has size either four or eight, and $|C| = 4$ if and only if $C = \{\vec{y}_0, \vec{y}_0 + \vec{y}_1, \vec{y}_0 + \vec{y}_2, \vec{y}_0 + \vec{y}_1 + \vec{y}_2\}$ for some $\vec{y}_i \in \Z^n$, i.e., $C$ is a parallelogram.
\end{lemma}
\begin{proof}
	By shifting $\vec{t}$ appropriately, we may assume without loss of generality that $\vec{z}_3 = \vec0$. Let $\vec{x}:= \basis \vec{z}_1$, $\vec{y} := \basis \vec{z}_2$, and $\vec{w} := \basis \vec{v}$. Since $\vec{0}, \vec{x}, \vec{y}, \vec{x} + \vec{y} - 2\vec{w}$ are all the same distance from $\vec{t}$, we have
	\[
		\|\vec{x} - \vec{t}\|_2^2 = \|\vec{t}\|_2^2\; , \qquad \|\vec{y} - \vec{t}\|_2^2 = \|\vec{t}\|_2^2\; , \qquad \|\vec{x} + \vec{y} - 2\vec{w}- \vec{t}\|_2^2 = \|\vec{t}\|_2^2
		\; .
	\]
	Recalling the identity $\| \vec{u}_1 - \vec{u}_2\|_2^2 =\|\vec{u}_1\|_2^2 + \|\vec{u}_2\|_2^2 - 2\inner{\vec{u}_1, \vec{u}_2}$, we have
	\begin{equation}
	\label{eq:stuff}
		0= \|\vec{x} \|_2^2 - 2\inner{\vec{x}, \vec{t}} = \|\vec{y} \|_2^2 - 2\inner{\vec{y}, \vec{t}} = \|\vec{x} + \vec{y}- 2\vec{w} \|_2^2 - 2\inner{\vec{x}, \vec{t}} - 2\inner{\vec{y}, \vec{t}} + 4 \inner{\vec{w}, \vec{t}}
		\; .
		\end{equation}

	Furthermore, since
	$\vec0 \in \CVP(\vec{t}, \lat)$, and since $\vec{w}, \vec{x} - \vec{w}, \vec{y} - \vec{w}, \vec{x} + \vec{y} - \vec{w}$ are lattice vectors, we must have
	\begin{equation*}
	\label{eq:no_closer_vectors}
		\|\vec{w} - \vec{t}\|_2^2 \geq \|\vec{t}\|_2^2 \; , \quad \|\vec{x} - \vec{w} - \vec{t}\|_2^2 \geq \|\vec{t}\|_2^2 \; , \quad \|\vec{y} - \vec{w} - \vec{t}\|_2^2 \geq \|\vec{t}\|_2^2 \; , \quad \|\vec{x} + \vec{y} - \vec{w} - \vec{t}\|_2^2 \geq \|\vec{t}\|_2^2
		\; .
	\end{equation*}
	(Otherwise, there would be a lattice vector closer to $\vec{t}$ than $\vec0$.)
	Rearranging as above, we have
	\begin{align*}
		\delta_1 &:= \|\vec{w}\|_2^2 - 2\inner{\vec{w}, \vec{t}} \geq 0 \; , \\ 
		\delta_2 &:= \|\vec{x} - \vec{w}\|_2^2 - \|\vec{x}\|_2^2 + 2\inner{\vec{w}, \vec{t}} = \|\vec{x} - \vec{w}\|_2^2 - 2\inner{\vec{x}, \vec{t}} + 2\inner{\vec{w}, \vec{t}} \geq 0 \; , \\ 
		\delta_3 &:=  \|\vec{y} - \vec{w}\|_2^2 - \|\vec{y}\|_2^2 + 2\inner{\vec{w}, \vec{t}} =  \|\vec{y} - \vec{w}\|_2^2 - 2\inner{\vec{y}, \vec{t}} + 2\inner{\vec{w}, \vec{t}} \geq 0 \; , \\
		\delta_4 &:=  \|\vec{x} + \vec{y} - \vec{w}\|_2^2 - \|\vec{x} + \vec{y} - 2\vec{w}\|^2 - 2\inner{\vec{w}, \vec{t}}=  \|\vec{x} + \vec{y} - \vec{w}\|_2^2 - 2\inner{\vec{x}, \vec{t}} - 2\inner{\vec{y}, \vec{t}}+  2\inner{\vec{w}, \vec{t}} \geq 0
		\; ,
	\end{align*}
	where we have used Eq.~\eqref{eq:stuff}.
	Then, 
	\begin{align*}
	\delta_1 + \delta_2 + \delta_3 + \delta_4 
		&= \|\vec{w}\|_2^2 + \|\vec{x} - \vec{w}\|_2^2 + \|\vec{y} - \vec{w}\|_2^2 + \|\vec{x} + \vec{y} - \vec{w}\|_2^2 \\
		&\qquad - \|\vec{x}\|_2^2 - \|\vec{y}\|_2^2 - \|\vec{x} + \vec{y} - 2\vec{w}\|_2^2\\
		&= 0
	\; .
	\end{align*}
	Since the $\delta_i$ are all non-negative and they sum to zero, they must all be zero. In other words, $\vec{z}_1', \vec{z}_2', \vec{z}_3',\vec{z}_4' \in \CVP(\vec{t}, \lat)$ as needed.
	
	Finally, notice that $2\vec{z}_j' = \vec{z}_1 + \vec{z}_2 + \vec{z}_3 + \vec{z}_4 - 2\vec{z}_j$. If $|C| < 8$, then there exists $i,j$ such that $\vec{z}_i = \vec{z}_j'$. If $i \neq j$, then we see that $\vec{z}_{i} + \vec{z}_{j} = \vec{z}_{k} + \vec{z}_{\ell}$, i.e., the $\vec{z}_{i'}$ form a parallelogram. Furthermore, we must have $\vec{z}_j = \vec{z}_i'$, $\vec{z}_k = \vec{z}_\ell'$, and $\vec{z}_{\ell} = \vec{z}_k'$, i.e., $|C| = 4$. On the other hand, if $i = j$, then we have $4\vec{z}_i = \vec{z}_1 + \vec{z}_2 + \vec{z}_3 + \vec{z}_4$, which yields a contradiction because then $\vec{z}_i$ lies in the convex hull of the other vectors, which means that $\basis \vec{z}_i$ cannot be distinct vectors equidistant from $\vec{t}$.
\end{proof}

The next two lemmas show some basic properties about the expressiveness of $3$-SAT.

\begin{lemma}
	\label{lem:7outof8}
	For any $k\geq1$ and non-empty set $S\subseteq \{0,1\}^n$ with $|S| \leq 2^k$, there exists a $k$-CNF on $n$ variables such that exactly $|S|-1$ of the elements in $S$ are satisfying assignments.
\end{lemma}
\begin{proof}
	We show how to find a $k$-clause that is satisfied by exactly $|S|-1$ elements. The proof is by induction on $k$. The base case $k = 1$ is trivial. So, we suppose that the result holds for $k-1$. We assume without loss of generality that the number of strings in $S$ whose first coordinate is one is between $1$ and $2^{k-1}$. (I.e., we assume that there are at least as many zeros as ones and that not all strings are the same on this coordinate.) Let $S_1$ be the set of strings with non-zero first coordinate. By induction, there is a $(k-1)$-clause $\phi$ such that exactly $|S_1|-1$ elements in $S_1$ satisfy $\phi$. Then $\phi \vee \neg x_1$ is a $k$-clause satisfied by exactly $|S|-1$ elements in $S$, as needed.
\end{proof}

\begin{lemma}
	\label{lem:3SAT_separates}
	For any non-empty disjoint sets $S, T \subseteq \{0,1\}^n$ with $|S| = 4$ and $|T| \geq 2$, there exists a $3$-CNF on $n$ variables such that all elements in $S$ are satisfying assignments and at least one element in $T$ is not a satisfying assignment.
\end{lemma}

\begin{proof}
We will find an assignment of $3$ variables that satisfies $S$, but doesn't satisfy at least one element of $T$. %

Define the majority string $s \in \{0,1\}^n$ of $S$ to be such that $s_i = 0$, if for at least $2$ strings in $S$, the $i$-th coordinate is $0$, and $s_i = 1$, otherwise. Let $t \in T \setminus \{s\}$. Consider a position $j$ where $t$ differs from $s$. Set the $j$-th variable $x_j = s_j$. This satisfies at least $2$ of the strings in $S$. Let $a, b$ be the two strings in $S$ such that $a_j = b_j \neq s_j$. Note that $t_j \neq s_j$, and hence $t_j = a_j = b_j$. Since $t$ is different from $a$ and $b$, there exist positions $k$ and $\ell$ such that $t_k \neq a_k$ and $t_\ell \neq b_\ell$. We set $x_k = a_k$ and $x_\ell = b_\ell$. Thus, we satisfy every element of $S$ but do not satisfy $t$.   
\end{proof}

Finally, we prove Theorem~\ref{thm:no_natural}. To do so, we first use Lemmas~\ref{lem:4to8} and~\ref{lem:3SAT_separates} to argue that if $\vec{z}_1,\vec{z}_2, \vec{z}_3, \vec{z}_4 \in A$ satisfy $\vec{z}_1 + \vec{z}_2 + \vec{z}_3 + \vec{z}_4 = \vec0 \bmod 2$ as in Lemma~\ref{lem:4to8}, then the $\vec{z}_i$ must form a parallelogram, where $A := f(\{0,1\}^n)$ is the image of $f$. We therefore conclude that if $\vec{z}_1,\ldots, \vec{z}_8 \in A$ form an affine $3$-cube modulo two, then they must actually form a parallelepiped. From this and Lemma~\ref{lem:7outof8}, we derive a contradiction by Lemma~\ref{lem:same-distances-imply-cube}. Therefore, $A \bmod 2$ cannot contain any affine $3$-cube, which means that $n' > 4(n-2)/3$ by Lemma~\ref{lem:additive_combinatorics}.

\begin{proof}[Proof of Theorem~\ref{thm:no_natural}]
	Let $R$ be a natural reduction from $3$-SAT on $n$ variables to $\CVP_2$ on rank $n'$ lattices. I.e., $R$ maps $3$-SAT instances $\phi$ to $\basis \in \R^{d \times n'}$ and $\vec{t}\in \R^d$. First, notice that $f$ must be injective. In particular, if $f$ is not injective, then the reduction cannot possibly be valid because for every two distinct assignments $\vec{x}, \vec{x}' \in \{0,1\}^n$, there exists a $3$-SAT instance $\phi$ that is satisfied by one but not the other. Let $A := f(\{0,1\}^n) \subset \Z^{n'}$ be the image of $f$.
	
	Suppose that there exist distinct $\vec{z}_1 := f(\vec{x}_1), \vec{z}_2  := f(\vec{x}_2), \vec{z}_3  := f(\vec{x}_3), \vec{z}_4  := f(\vec{x}_4) = \vec{z}_1 + \vec{z}_2+ \vec{z}_3 - 2\vec{v} \in A$ for some $\vec{v} \in \Z^{n'}$. Then, for any $\basis, \vec{t}$, if $\vec{z}_1, \vec{z}_2, \vec{z}_3, \vec{z}_4 \in \CVP(\vec{t}, \basis)$, by Lemma~\ref{lem:4to8}, we must also have $\vec{z}_1',\vec{z}_2',\vec{z}_3',\vec{z}_4' \in \CVP(\vec{t}, \basis)$ as well, where
	\[\vec{z}_1' := \vec{z}_2 + \vec{z}_3 - \vec{v}, \quad 
	\vec{z}_2' := \vec{z}_1 + \vec{z}_3 - \vec{v}, \quad 
	\vec{z}_3' := \vec{z}_1 + \vec{z}_2 - \vec{v}, \quad 
	\vec{z}_4' := \vec{v} 
	\; .
	\]
	Therefore, by applying $R$ to, e.g., the empty formula $\emptyset$, we see that $\vec{z}_1', \ldots, \vec{z}_4' \in A$ must also lie in the image of $f$, i.e., $\vec{z}_j' = f(\vec{x}_j')$. Again by Lemma~\ref{lem:4to8}, either the $\vec{z}_i$ form a parallelogram, or the sets $S := \{\vec{x}_1,\ldots, \vec{x}_4\}$ and $S' := \{ \vec{x}_1',\ldots, \vec{x}_4'\}$ are disjoint. But, if $S, S'$ are disjoint, then by Lemma~\ref{lem:3SAT_separates}, there exists a $3$-clause $\phi$ such that $\phi(\vec{x}_i) = 1$ for all $i$ but there exists a $j$ such that $\phi(\vec{x}_j') = 0$. Then, taking $\basis, \vec{t} = R(\phi)$, we see that $\vec{z}_j' \notin \CVP(\vec{t}, \basis)$, a contradiction.
	
We conclude that any such $\vec{z}_1, \vec{z}_2 , \vec{z}_3, \vec{z}_4 = \vec{z}_1 + \vec{z}_2 + \vec{z}_3 - 2\vec{v} \in A$ must form a parallelogram. I.e., if the $\vec{z}_i$ form an affine subspace mod two, then they form a parallelogram $\{\vec{z}_1,\vec{z}_2, \vec{z}_3, \vec{z}_4\} = \{\vec{y}_0, \vec{y}_0 + \vec{y}_1, \vec{y}_0 + \vec{y}_2, \vec{y}_0 + \vec{y}_1 + \vec{y}_2\}$. Now, suppose that $A$ modulo two contains an affine $3$-cube. I.e., suppose that it contains distinct $\vec{z}_1 := f(\vec{x}_1),\ldots, \vec{z}_8 := f(\vec{x}_8)$ such that $\basis(\vec{z}_i - \vec{y}_0 - \sum_{j \in W_i} \vec{y}_j) \in 2\lat$ for some $\vec{y}_0, \vec{y}_1, \vec{y}_2, \vec{y}_3 \in \Z^n$ and distinct $W_i \subseteq \{1,2,3\}$. Then, by the above, we see that $\vec{z}_i = \vec{y}_0 + \sum_{j \in W_i} \vec{y}_j$. I.e., the $\vec{z}_i$ form a parallelepiped. But, by Lemma~\ref{lem:7outof8}, there exists a $3$-clause $\phi$ such that exactly seven out of the eight $\vec{x}_i$ satisfy $\phi$. Therefore, $(\basis, \vec{t}) := R(\phi)$ must have $\|\basis \vec{z}_i - \vec{t}\| = \dist(\vec{t}, \basis \Z^{n'})$ for seven out of the eight $\vec{z}_i$. But, by Lemma~\ref{lem:same-distances-imply-cube}, this is not possible.
	
	Finally, we conclude that $A$ cannot include any affine $3$-cube modulo two. Therefore, by Lemma~\ref{lem:additive_combinatorics}, we see that $n' > 4(n-2)/3$.
\end{proof}
\section{On certain weighted sums of binomial coefficients}

\label{sec:crazy_binomial}

In this section, we study two different classes of sums:
\[
	\sum_{i=0}^{k} \binom{k}{i} |i - \tau|^p
	\; ,
\]
and 
\[
\sum_{i=0}^{k} (-1)^i \binom{k}{i} |i - \tau|^p
\; .
\]
Recall that these correspond to the eigenvalues of $H_{k,p}(t^*)$ from Corollary~\ref{cor:H-eigenvalues}, and particularly Eqs.~\eqref{eq:lambda} and~\eqref{eq:lambda-par}. The bounds derived in this section allow us to build good isolating parallelepipeds for parity, as in Theorem~\ref{thm:ips_parity}. There, we use $\tau = \floor{k/2}$ (for which some of our results are cleaner), but we leave $\tau$ as a variable when we can.

\subsection{Identities and bounds for the alternating sum}

\noindent We first address the alternating sum. It will be convenient to prove our main result concerning this sum in a slightly different parameterization. The case $m = 0$ is due to~\cite{stack_exchange_2827591} (and also appeared in~\cite{conf/soda/LiWW20}), and the proof uses a contour integral suggested in~\cite{stack_exchange_comment}.

\begin{theorem}
	\label{thm:crazy_binomial}
	For any integers $n \geq 1$ and $0 \leq m \leq n$ and $p \in \C$ satisfying $1 \leq \Re(p) < 2n-m$,
	\begin{align*}
	\sum_{i=0}^{2n-m} &(-1)^{n-i} \binom{2n-m}{i}  |n-i|^p \\
	&= -2\sin(\pi p/2) \binom{2n-m}{n} \int_0^\infty \frac{x^p}{\sinh(\pi x)} \cdot \Re\Big( \frac{\Gamma(n - m + 1) \Gamma(n+1)}{\Gamma(n-m + ix + 1)\Gamma(n-ix + 1)} \Big) {\rm d} x
	\; .
	\end{align*}
\end{theorem}
\begin{proof}
	Let
	\begin{align*}
	f_{n,m, p}(x) &:= \frac{e^{p \log x}}{\sin(\pi x)} \cdot \frac{1}{\Gamma(n + x + 1)\Gamma(n-x-m+1)} \\
	&= f_{n,0,p}(x)\cdot \frac{\Gamma(n-x+1)}{\Gamma(n-x-m+1)}
	\; ,
	\end{align*}
	where $\log x$ is the principal branch of the logarithm function, which satisfies $-\pi < \Im(\log(x)) \leq \pi$ and is analytic except for a branch cut along the non-positive part of the real axis, $x \leq 0$ (where the imaginary part jumps from $\pi$ to $-\pi$). $f_{n,m, s}(x)$ itself is analytic except for this branch cut and simple poles at $x = k$ for integers $k$ with $1 \leq k \leq n-m$ (or no poles at all if $n = m$). For such $k$, the residue is given by
	\begin{align*}
	\mathsf{Res}_{x\to k}(f_{n,m,p}(x)) &= \mathsf{Res}_{x\to k}(1/\sin(\pi x)) \cdot k^p \cdot \frac{1}{\Gamma(n+k+1) \Gamma(n-k-m+1)} \\
	&=  \frac{(-1)^k}{\pi} \cdot k^p \cdot \frac{1}{(n+k)!(n-k-m)!}
	\; .
	\end{align*}
	
	For $R > n$ and $\eps \in (0,1)$, let $C_{R, \eps}$ be the contour defined as the (counter-clockwise oriented) union of the following four curves: (1) the line from $iR$ to $i\eps$; (2) the small half circle $\{ i\eps e^{ix} \ : \ 0 \leq x \leq \pi\}$; (3) the line from $-i\eps$ to $-iR$; and (4) the large half circle $\{ -iR e^{ix} \ : \ 0 \leq x \leq \pi \}$.
	By Cauchy's residue theorem,
	\[
	\oint_{C_{R, \eps}} f_{n,m, p}(x) {\rm d} x  = 2\pi i \sum_{k=1}^{n-m}  \mathsf{Res}_{x\to k}(f_{n,m, p}(x))  = 2 i  \sum_{k=1}^{n-m} \frac{(-1)^k k^p}{(n+k)!(n-k-m)!}
	\; .
	\]
	
	For sufficiently large $|x|$ we have $|f_{n,m,p}(x)| = O_{n,m,p}(|x|^{p+m-2n -1})$ (this follows, e.g., from Theorem~\ref{thm:ramanujan}), which implies that
	\[
	\lim_{R \to \infty} \Big( R \int_{-\pi}^\pi |f_{n,m,p}(R e^{ix})|{\rm d} x \Big) = 0
	\; ,
	\]
	provided that $\Re(p) < 2n-m$.
	Similarly, for $|x| \leq 1/2$ and $\Re(p) \geq 1$, $|f_{n,m, p}(x)|$ is bounded (because it is continuous over this compact region of the complex plane), which implies that
	\[
	\lim_{\eps \to 0} \Big( \eps \int_{-\pi}^\pi |f_{n,m, p}(\eps e^{ix})| {\rm d} x \Big)= 0
	\; .
	\]
	It follows that
	\begin{align}
	2 i  \sum_{k=0}^{n-m} \frac{(-1)^k k^p}{(n+k)!(n-k-m)!} &=
	\lim_{R \to \infty, \eps \to 0^+}\Big( \oint_{C_{R,\eps}} f_{n,m, p}(x) {\rm d} x\Big) \nonumber\\
	&= -i\int_{0}^\infty f_{n,m, p}(ix) {\rm d} x - i \int_{0}^{\infty} f_{n,m, p}(-ix) {\rm d} x \nonumber\\
	&= -i\int_{0}^\infty (f_{n,m, p}(ix) + f_{n,m, p}(-ix)) {\rm d} x \nonumber\\
	&= -\int_0^\infty   \frac{x^p}{\sinh(\pi x)\Gamma(n+ix + 1)\Gamma(n-ix+1)} \nonumber\\
	&\quad\cdot \Big( e^{\pi i p/2} \frac{\Gamma(n-ix+1)}{\Gamma(n-ix-m+1)} - e^{-\pi ip/2} \frac{\Gamma(n+ix+1)}{\Gamma(n+ix-m+1)} \Big) {\rm d} x \label{eq:f}
	\; .
	\end{align}
	
	Now, let
	\begin{align*}
		g_{n,m, p}(x) &:= \frac{e^{p \log x}}{\sin(\pi x)} \cdot \frac{1}{\Gamma(n -m + x + 1)\Gamma(n-x+1)} \\
		&= f_{n,0,p}(x)\cdot \frac{\Gamma(n+x+1)}{\Gamma(n-m+x+1)}
		\; .
	\end{align*}
	Notice that $g_{n,m,s}$ has poles at $x = k$ for integers $1 \leq k \leq n$. An essentially identical analysis then shows that
	\begin{align}
		2 i  \sum_{k=1}^{n} \frac{(-1)^k k^p}{(n-m+k)!(n-k)!} &= 
		-i\int_{0}^\infty (g_{n,m, p}(ix) + g_{n,m, p}(-ix)) {\rm d} x \nonumber\\
		&= -\int_0^\infty  \frac{x^p}{\sinh(\pi x)\Gamma(n+ix + 1)\Gamma(n-ix+1)} \nonumber\\
		&\quad \cdot \Big( e^{\pi i p/2} \frac{\Gamma(n+ix+1)}{\Gamma(n-m+ix+1)} - e^{-\pi ip/2} \frac{\Gamma(n-ix+1)}{\Gamma(n-m-ix+1)} \Big) {\rm d} x \label{eq:g}
		\; .
	\end{align}
	
	Summing Eqs.~\eqref{eq:f} and~\eqref{eq:g} and multiplying by $-(2n-m)!/(2i)$ gives
	\begin{align*}
	&\sum_{k=0}^{n-m} (-1)^{k+1} k^p \binom{2n-m}{n+k} 
			+ \sum_{k=1}^{n} (-1)^{k+1} k^p \binom{2n-m}{n-k} \\
	&\qquad = \sum_{j=0}^{2n-m} (-1)^{n-j+1} \binom{2n-m}{j}  |n-j|^p\\
	&\qquad = \sin(\pi p/2) (2n-m)! \cdot \int_0^\infty \frac{x^p}{\sinh(\pi x)\Gamma(n+ix + 1)\Gamma(n-ix+1)} \\
	&\qquad\qquad\qquad  \cdot \Big( \frac{\Gamma(n+ix+1)}{\Gamma(n-m+ix+1)}  + \frac{\Gamma(n-ix+1)}{\Gamma(n-m-ix+1)} \Big) {\rm d} x\\
	&\qquad = \sin(\pi p/2) \binom{2n-m}{n} \int_0^\infty \frac{x^p}{\sinh(\pi x)} \\
			&\qquad\qquad\qquad \cdot  \Big( \frac{n! (n-m)!}{\Gamma(n-m+ix+1)\Gamma(n-ix + 1)}  + \frac{n! (n-m)!}{\Gamma(n+ix+1)\Gamma(n-m-ix+1)} \Big) {\rm d} x\\
	&\qquad = 2\sin(\pi p/2) \binom{2n-m}{n} \int_0^\infty \frac{x^p}{\sinh(\pi x)} \cdot \Re\Big( \frac{\Gamma(n-m+1)\Gamma(n+1)}{\Gamma(n-m+ix+1) \Gamma(n-ix+1)} \Big) {\rm d} x
	\; ,
	\end{align*}
	as needed.
\end{proof}

\begin{corollary}
	\label{cor:somecorollary}
	For any integers $k \geq 1$ and $0 \leq \tau \leq k$, and $p \in \C$ with $1 \leq \Re(p) < k$,
	\begin{align*}
	\sum_{i=0}^{k}& (-1)^{i} \binom{k}{i}  |i-\tau|^p = \\
	&(-1)^{\tau+1}\cdot 2\sin(\pi p/2) \binom{k}{\tau} \int_0^\infty \frac{x^p}{\sinh(\pi x)} \cdot \Re\Big( \frac{\Gamma(\tau+1)\Gamma(k-\tau+1)}{\Gamma(\tau+ix+1)\Gamma(k-\tau-ix+1)} \Big) {\rm d} x
	\; .
	\end{align*}
\end{corollary}
\begin{proof}
	For $0 \leq \tau \leq k/2$, this follows immediately from plugging $n := k-\tau$ and $m := k-2\tau$ into Theorem~\ref{thm:crazy_binomial}, and multiplying both sides by $(-1)^\tau$. For $k/2 \leq \tau \leq k$, the result follows by noting that both the left-hand side and the right-hand side satisfy the equation $f(k-\tau) = (-1)^k f(k)$.
\end{proof}

\noindent Finally, we derive the corollary that we need for our application.

\begin{corollary}
\label{cor:crazy_binomial}
	For any integer $k \geq 3$ and $1 \leq p < k$, 
	let 
	\begin{align*}
		S_{k,p} &:= \binom{k}{\floor{k/2}}^{-1} \cdot \sum_{i=0}^{k} (-1)^{i} \binom{k}{i}  |i-\floor{k/2}|^p
		\; .
	\end{align*}
Then, for all integers $k \geq 2$, $ |S_{2k,p}| = |S_{2k-1,p}|$ is monotonically decreasing in $k$  with 
	\begin{align*}
		|S_{k,p}|  &\geq \lim_{k \to \infty} |S_{k,p}| \\
		&= 2 |\sin(\pi p/2)| \zeta(1+p) (2-2^{-p}) \frac{\Gamma(p+1)}{\pi^{p+1}} \\
		&\geq 4|\sin(\pi p/2)| (p/(e\pi))^p
		\; ,
	\end{align*}
	where $\zeta(s) := 1+1/2^s + 1/3^s + \cdots$ is Riemann's $\zeta$ function. Furthermore, 
	\[
		\sign(S_{k,p})=(-1)^{\floor{k/2} + \floor{p/2} + 1}
		\; .
	\]
\end{corollary}
\begin{proof}
	By plugging $\tau := \floor{k/2}$ into Corollary~\ref{cor:somecorollary}, we have that for even $k$,
	\[
	S_{k,p} = (-1)^{k/2+1}\cdot2\sin(\pi p/2) \int_0^\infty \frac{x^p}{\sinh(\pi x)} \cdot \frac{\Gamma(k/2+1)^2}{\Gamma(k/2+ix+1)\Gamma(k/2-ix+1)} {\rm d} x
	\; .
	\]
	For odd $k$, 
	notice that 
	we have
	\begin{align*}
		&\Re\Big( \frac{\Gamma(k/2+1/2)\Gamma(k/2+3/2)}{\Gamma(k/2+ix+1/2)\Gamma(k/2-ix+3/2)} \Big) \\
		&\qquad =
		\frac{\Gamma(k/2+1/2)^2}{\Gamma(k/2+ix+1/2)\Gamma(k/2-ix+1/2)} \cdot \frac{(k/2+1/2)^2}{(k/2+1/2)^2+x^2}\\
		&\qquad = \frac{\Gamma((k+1)/2+1)^2}{\Gamma((k+1)/2+ix+1)\Gamma((k+1)/2-ix+1)}
		\; ,
	\end{align*}
	where we have applied the functional equation $\Gamma(r+1) = r\Gamma(r)$ repeatedly, and explicitly computed the real part of $(k/2-ix+1/2)^{-1}$. Plugging this in to Corollary~\ref{cor:somecorollary}, we see that $S_{k,p} = -S_{k+1,p}$ for odd $k$.
	
	Now, let
	\[
		P_{2k}(x) :=
		\frac{\Gamma(k+1)^2}{\Gamma(k+ix+1)\Gamma(k-ix+1)} 
	\]
	By Theorem~\ref{thm:ramanujan}, we have the surprising identity
	\[
		P_{2k}(x) = \frac{\sinh(\pi x)}{\pi x} \cdot \prod_{j=1}^{k}(1+x^2/j^2)^{-1}
		\; .
	\]
	It follows immediately by inspection that $P_{2k}$ is decreasing in $k$ (for real $x \neq 0$), which implies the monotonicity result.
	 
	 To obtain the asymptotic result for $|S_{k,p}|$, we note that the monotonicity described above allows us to apply the dominated convergence theorem and exchange the limit and the integral to obtain
	 \begin{align*}
	 	\lim_{k \to \infty} |S_{k,p}|
	 		&=  2 |\sin(\pi p/2)| \int_0^\infty \frac{x^p}{\sinh(\pi x)} {\rm d} x \\
	 		&= \frac{2 |\sin(\pi p/2)|}{\pi^{p+1}} \cdot \int_0^\infty \frac{x^p}{\sinh(x)} {\rm d} x 
	 		\; ,
	 \end{align*}
	 where we have used the fact that $\lim_{k \to \infty} P_{2k}(x) = 1$. To compute the integral, we can note that for~$x > 0$,
	 \[
	 	\frac{1}{\sinh(x)} = \frac{2 \exp(-x)}{1-\exp(-2x)} = 2\exp(-x) \sum_{j=0}^\infty \exp(-2jx)
	 	\; , 
	 \]
	 so that 
	 \[
	 	\int_0^\infty \frac{x^p}{\sinh(x)} {\rm d} x = 2\sum_{j=0}^\infty \int_0^\infty \exp(-(2j+1)x) x^p {\rm d} x = 2 \sum_{j=0}^\infty \frac{\Gamma(p+1)}{(1+2j)^{p+1}} = (2-2^{-p}) \Gamma(p+1) \zeta(p+1)
	 	\; ,
	 \]
	 as needed.
	 
Finally, $\sign(S_{k,p})=(-1)^{\floor{k/2} +1} \cdot\sign(\sin(\pi p/2)) =(-1)^{\floor{k/2} + \floor{p/2} + 1}$.
\end{proof}

\subsection{Bounds for the non-alternating sum}

\begin{lemma}
\label{lem:simple_binomial}
For any integer $k \geq 2$, and $p\in \mathbb{R}$ with $1\leq p < k$,
\[
\sum_{i=0}^k \binom{k}{i}  |i-k/2|^p \leq 11 \binom{k}{\floor{k/2}} (pk/2)^{(p+1)/2} 
\, .
\]
\end{lemma}

\begin{proof}
If $p \geq k/2$, then
\begin{align*}
\sum_{i=0}^k \binom{k}{i}  |i-k/2|^p 
\leq 2^k (k/2)^p 
\leq 4 \frac{2^k}{\sqrt{4k}} \cdot (k/2)^{p+1}
\leq 4 \binom{k}{\floor{k/2}} (pk/2)^{(p+1)/2} 
\, .
\end{align*}

If $p< k/2$, then let $t_1=k/2-\sqrt{pk}$ and $t_2=k/2-\sqrt{pk/2}\geq 0$. Note that for any $1\leq i \leq t_1$,
\begin{align*}
\frac{ \binom{k}{i-1} |i-1-k/2|^p}  { \binom{k}{i} |i-k/2|^p} 
&= \frac{i}{k-i+1}\cdot\left(1+\frac{1}{k/2-i}\right)^p \\
&\leq  \frac{t_1}{k-t_1} \cdot\left(1+\frac{1}{k/2-t_1}\right)^p \\
&=\left(1-\frac{2\sqrt{pk}}{k/2+\sqrt{pk}} \right)\cdot \left(1+\frac{1}{\sqrt{pk}} \right)^p \\
& \leq (1-4(\sqrt{2}-1)\sqrt{p/k}) \cdot (1+\sqrt{p/k}+p/k) \\
& \leq 1-(4\sqrt{2}-5)\sqrt{p/k}
\, ,
\end{align*}
where we used $(1+x)^p\leq e^{px} \leq 1+px+(px)^2$ for every $px\leq 1$.
Similarly, for $1\leq i\leq t_2$,
\begin{align*}
\frac{ \binom{k}{i-1} |i-1-k/2|^p}  { \binom{k}{i} |i-k/2|^p} 
&\leq  \frac{t_2}{k-t_2} \cdot\left(1+\frac{1}{k/2-t_2}\right)^p \\
&=\left(1-\frac{2\sqrt{pk/2}}{k/2+\sqrt{pk/2}} \right)\cdot \left(1+\frac{1}{\sqrt{pk/2}} \right)^p \\
& \leq (1-\sqrt{2p/k}) \cdot (1+\sqrt{2p/k}+2p/k) \\
& < 1
\, .
\end{align*}

We have that 
\begin{align*}
    \sum_{i=0}^{\floor{t_1}} \binom{k}{i}  (k/2-i)^p
    &\leq \binom{k}{\floor{t_2}}  (k/2-t_2)^p \cdot \sum_{i=0}^{\floor{t_1}} \left(1-(4\sqrt{2}-5)\sqrt{p/k}\right)^i\\
    &\leq \binom{k}{\floor{k/2}}  (pk/2)^{p/2} \cdot \frac{1}{(4\sqrt{2}-5)\sqrt{p/k}}\\
    &=\binom{k}{\floor{k/2}}  (pk/2)^{(p+1)/2} \cdot \frac{\sqrt{2}}{(4\sqrt{2}-5)p} 
    \,,
\end{align*}
\begin{align*}
    \sum_{i=\ceil{t_1}}^{\floor{t_2}} \binom{k}{i}  (k/2-i)^p
    &\leq \binom{k}{\floor{t_2}}  (k/2-t_2)^p \cdot (t_2-t_1+1)\\
    &\leq\binom{k}{\floor{k/2}}  (pk/2)^{p/2} \cdot \left(\sqrt{pk}-\sqrt{pk/2}+1\right) \\
    &=\binom{k}{\floor{k/2}}  (pk/2)^{(p+1)/2} \cdot \left(\sqrt{2}-1+\sqrt{2/(pk)}\right)
    \,,
\end{align*}
and
\begin{align*}
\sum_{i=\ceil{t_2}}^{k-\ceil{t_2}} \binom{k}{i}  |k/2-i|^p
&\leq \binom{k}{\floor{k/2}}(pk/2)^{p/2}\cdot(k-2t_2+1)\\
&\leq \binom{k}{\floor{k/2}}(pk/2)^{p/2}\cdot\left(2\sqrt{pk/2}+1\right)\\
&\leq \binom{k}{\floor{k/2}}(pk/2)^{(p+1)/2}\cdot\left(2+\sqrt{2/(pk)}\right)\\
\,.
\end{align*}
Finally, using the above inequalities,
\begin{align*}
\sum_{i=0}^k \binom{k}{i}  |i-k/2|^p
&\leq 
2\sum_{i=0}^{\floor{t_1}} \binom{k}{i}  (k/2-i)^p
+2\sum_{i=\ceil{t_1}}^{\floor{t_2}} \binom{k}{i}  (k/2-i)^p
+\sum_{i=\ceil{t_2}}^{k-\ceil{t_2}} \binom{k}{i}  |i-k/2|^p \\
&\leq 
\binom{k}{\floor{k/2}}  (pk/2)^{(p+1)/2} \times\\
&\;\;\;\;\;\;\;\;\;\;\times \left( \frac{2\sqrt{2}}{(4\sqrt{2}-5)p} + \left(2\sqrt{2}-2+2\sqrt{2/(pk)}\right)+\left(2+\sqrt{2/(pk)}\right)\right)\\
&\leq 11 \binom{k}{\floor{k/2}}  (pk/2)^{(p+1)/2} 
\,.
\end{align*}
\end{proof}

\begin{corollary}
\label{cor:simple_binomial}
For any integers $k$ and $c\geq 0$, and $p\in \mathbb{R}$ with $1\leq p < k$,
\[
\sum_{i=0}^k \binom{k}{i}  |i-(k-c)/2|^p \leq 11 \binom{k+c}{\floor{(k+c)/2}} (p(k+c)/2)^{(p+1)/2} 
\, .
\]
In particular, for $c = 1$,
\[
	\sum_{i=0}^k \binom{k}{i}  |i-k/2+1/2|^p \leq 44 \binom{k}{\floor{k/2}} (pk/2)^{(p+1)/2} 
	\; .
\]
\end{corollary}
\begin{proof}
	We have
\begin{align*}
\sum_{i=0}^k \binom{k}{i}  |i-(k-c)/2|^p 
\leq 
\sum_{i=0}^{k+c} \binom{k+c}{i}  |i-(k+c)/2|^p
\leq 11 \binom{k+c}{\floor{(k+c)/2}} (p(k+c)/2)^{(p+1)/2} 
\, .
\end{align*}

For the special case of $c = 1$, we note that the ratio
\[
	\frac{\binom{k+1}{\floor{(k+1)/2}}}{\binom{k}{\floor{k/2}}} 
\]
is exactly $2$ if $k$ is odd and is $2(k+1)/(k+2) \leq 2$ if $k$ is even.
Furthermore, $(k+1)^{(p+1)/2}/k^{(p+1)/2} \leq (1+1/k)^{(k+1)/2} \leq 2$ for $k \geq 1$, which can be verified by noting that $f(x) := (1+1/x)^{(x+1)/2}$ is decreasing as a function of $x > 0$. The result follows.
\end{proof}

\bibliographystyle{alpha}

\newcommand{\etalchar}[1]{$^{#1}$}

\appendix
\section{Hardness of SVP}
\label{app:svp}

We notice that Theorem~\ref{thm:CVP_hard_intro}, or more specifically Corollary~\ref{cor:cvp-hardness}, immediately implies an improvement to the main result in~\cite{ASGapETH18}. Specifically, while~\cite[Theorem 4.3]{ASGapETH18} previously only applied to some non-explicit set of $p$, we can now extend it to all $p \gtrsim 2.14$ with $p \notin 2\Z$. 

We give the formal statement below for completeness. The proof is essentially identical to the original. We simply substitute our Corollary~\ref{cor:cvp-hardness} for the main result from~\cite{conf/focs/BennettGS17} (noting, as in~\cite{ASGapETH18} that the hard $\CVP_p$ instance promised by Corollary~\ref{cor:cvp-hardness} has a particularly nice form). We also include a plot of $C_p$ in Figure~\ref{fig:Cp}, which is taken from~\cite{ASGapETH18}. (\cite{ASGapETH18} also proved that there is no $2^{o(n)}$-time algorithm for $\SVP_p$ for any $p$ assuming Gap-ETH.)

		\begin{figure}
	\begin{center}
		\includegraphics[width=0.4 \textwidth]{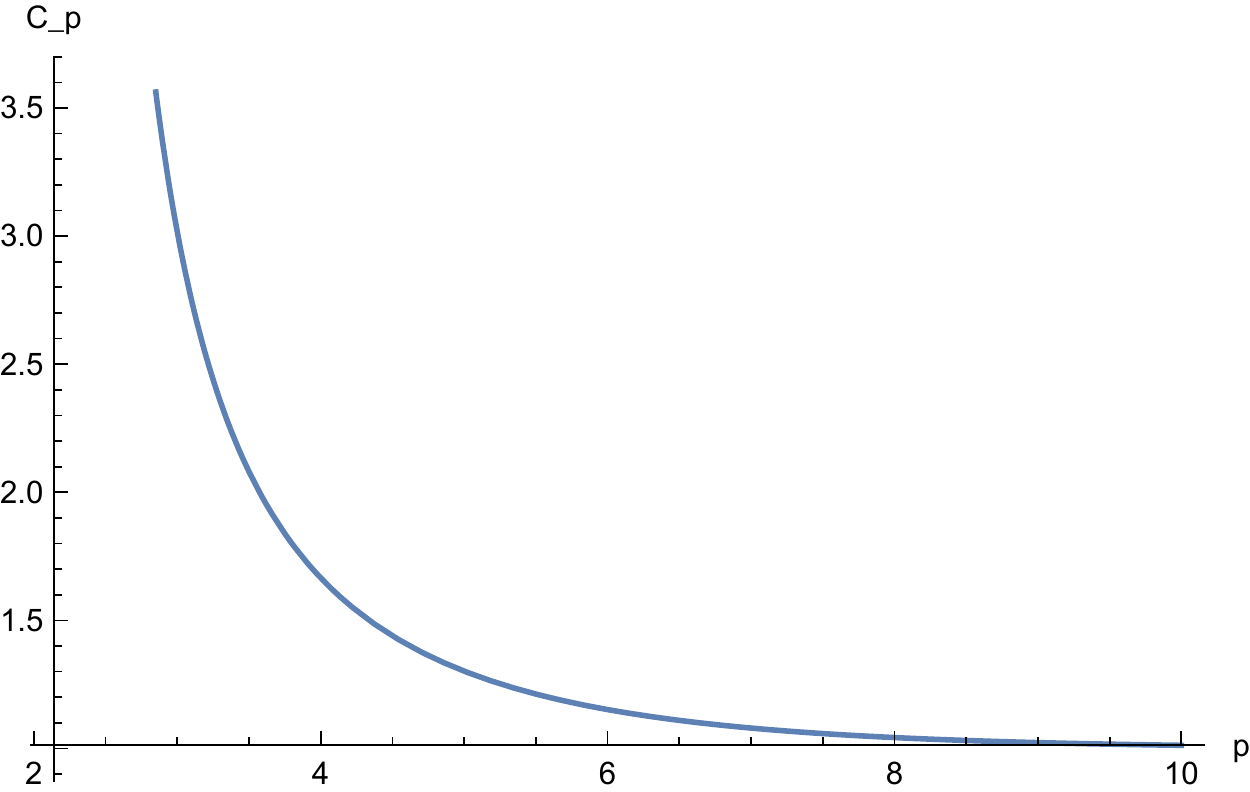}
		\qquad
		\includegraphics[width=0.4 \textwidth]{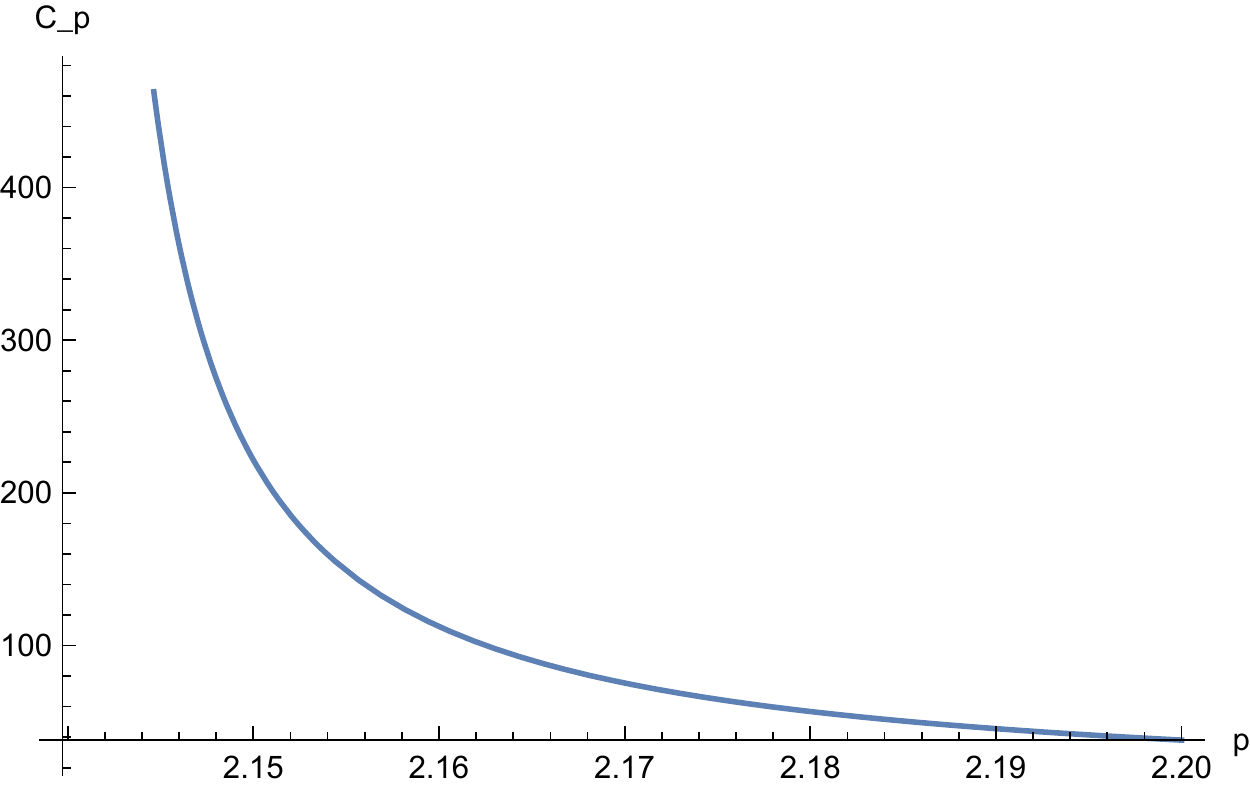}
		\caption{\label{fig:Cp} 
			The value $C_p$ for different values of $p > p_0$. In particular, for $p \notin 2\Z$, there is no $2^{n/C_p}$-time algorithm for $\SVP_p$ unless SETH is false. The plot on the left shows $C_p$ over a wide range of $p$, while the plot on the right shows the behavior when $p$ is close to the threshold $p_0 \approx 2.13972$. (Figure taken from~\cite{ASGapETH18}.)}
	\end{center}
\end{figure}

		\begin{theorem}
	\label{thm:SETH_hardness_centered_theta}
	For any integer $k \geq 2$ and $p > p_0$ with $p \notin 2\Z$, there is an efficient randomized reduction from Max-$k$-SAT on $n$ variables to $\SVP_p$ on a lattice of rank $\ceil{C_p n + \log^2 n}$, where 
	\[
	C_p := \frac{1}{1-\log_2W_p} \qquad \text{ and } \qquad W_p := \min_{\tau > 0} \exp(\tau/2^p) \Theta_p(\tau)
	\; .
	\] 
	Here, $\Theta_p(\tau) := \sum_{z \in \Z} \exp(-\tau |z|^p)$, and $p_0 \approx 2.13972$ is the unique solution to the equation $W_{p_0} = 2$.
	
	In particular, for every $\eps > 0$ and $p > p_0$ with $p \notin 2\Z$ there is no $2^{(1-\eps)n/C_p}$-time algorithm for $\CVP_p$ unless SETH is false.
\end{theorem}%

\end{document}